\documentclass[11pt]{article}

\usepackage{graphicx}
\usepackage{subcaption}
\usepackage{mwe}
\usepackage{xcolor}
\usepackage{longtable}
\usepackage{supertabular}
\usepackage{pspicture}
\usepackage{filecontents}
\usepackage{pst-bar}
\usepackage{bm,amsfonts,amssymb,amsmath,float}
\usepackage{sectsty}
\usepackage{colortbl}
\usepackage{multirow}
\usepackage{latexsym}
\usepackage{epsfig}
\usepackage{epstopdf}
\usepackage{threeparttable}
\usepackage{caption}
\usepackage{enumerate}
\usepackage{diagbox}
\usepackage[normalem]{ulem}
\usepackage{bm}
\usepackage{mathtools}
\usepackage{enumitem}
\usepackage{algorithm}
\usepackage{algpseudocode}
\usepackage{afterpage}
\usepackage{bbm}
\usepackage{algpseudocodex}
\usepackage{placeins}
\usepackage{makecell}
\usepackage{appendix}
\usepackage{indentfirst} 

\newenvironment{proof}{\noindent{\bf Proof:}}{\hfill\fbox{}\vspace*{1mm}}
\newtheorem{theorem}{Theorem}[section]
\newtheorem{example}[theorem]{Example}
\newtheorem{lemma}[theorem]{Lemma}
\newtheorem{definition}[theorem]{Definition}
\newtheorem{remark}[theorem]{Remark}
\newtheorem{proposition}[theorem]{Proposition}
\newtheorem{corollary}[theorem]{Corollary}
\newtheorem{problem}[theorem]{Problem}

\makeatletter
\newcommand\nextitem[1]{%
  \setcounter{\@enumctr}{#1}%
  \addtocounter{\@enumctr}{-1}%
}
\makeatother

\usepackage{ulem} 
\usepackage{hyperref}
\usepackage[mathlines, displaymath]{lineno} 

\usepackage{graphics}
\usepackage{amstext}
\usepackage{amscd}
\usepackage{rotating}
\usepackage{makeidx}
\usepackage{MnSymbol}
\usepackage{mdwlist}
\usepackage{fancyvrb}
\usepackage{mdframed}

\usepackage{tikz}
\usetikzlibrary{arrows.meta}

\begin{document}
\renewcommand{\labelenumi}{\arabic{enumi}.}
\renewcommand{\labelenumii}{\arabic{enumi}.\arabic{enumii}}
\renewcommand{\labelenumiii}{\arabic{enumi}.\arabic{enumii}.\arabic{enumiii}}
\renewcommand{\labelenumiv}{\arabic{enumi}.\arabic{enumii}.\arabic{enumiii}.\arabic{enumiv}}

\title{On the Number of Control Nodes of Threshold and XOR Boolean Networks}

\author{Christopher H. Fok$^1$$^*$, Liangjie Sun$^2$, Tatsuya Akutsu$^2$, Wai-Ki Ching$^1$}
\date{
   $^1$Department of Mathematics, The University of Hong Kong, Pokfulam Road, Hong Kong\\
   $^2$Bioinformatics Center, Institute for Chemical Research, Kyoto University, Kyoto 611-0011, Japan\\
   $^*$christopherfok2015@outlook.com}

\maketitle

\begin{abstract}
Boolean networks (BNs) are important models for gene regulatory networks and many other biological systems. In this paper, we study the minimal controllability problem of threshold and XOR BNs with degree constraints. Firstly, we derive lower-bound-related inequalities and some upper bounds for the number of control nodes of several classes of controllable majority-type threshold BNs. Secondly, we construct controllable majority-type BNs and BNs involving Boolean threshold functions with both positive and negative coefficients such that these BNs are associated with a small number of control nodes. Thirdly, we derive a linear-algebraic necessary and sufficient condition for the controllability of general XOR-BNs, whose update rules are based on the XOR logical operator, and construct polynomial-time algorithms for computing control-node sets and control signals for general XOR-BNs. Lastly, we use ring theory and linear algebra to establish a few best-case upper bounds for a type of degree-constrainted XOR-BNs called $k$-$k$-XOR-BNs. In particular, we show that for any positive integer $m \geq 2$ and any odd integer $k \in [3, 2^{m} - 1]$, there exists a $2^{m}$-node controllable $k$-$k$-XOR-BN with 1 control node. Our results offer theoretical insights into minimal interventions in networked systems such as gene regulatory networks.

{\bf Keywords:} Boolean Networks, Controllability, Minimum Control Node Set, Boolean Threshold Functions, XOR Operators
\end{abstract}



\section{Introduction}\label{section:intro}

Boolean networks (BNs) were first proposed by \cite{KauffmanS.A.1969Msae} for modeling gene regulatory networks. Since then, BNs have also been applied to study biological mechanisms such as T cell receptor signaling \cite{Saez2007}, apoptosis \cite{apoptosis} and the yeast cell-cycle network \cite{yeastrobust}.
By using BNs to model biological systems, we can gain insight into the systems and then design intervention strategies to drive the systems from undesirable states to desirable ones. This motivates the research community to study the control theory of BNs \cite{ControlAndObserve, AkutsuTree, MinRealization, ZhongAsynchronous, StuckAt, AsympObserve, CoarsestPartition}, which has gained increasing attention in recent years.

One of the important directions of research on BN control is the minimal controllability problem of BNs. This problem (defined in Problem \ref{problem:BN_min_control}, Section \ref{section:prelim}) asks to find the smallest set of control nodes in a BN which renders the BN controllable. There have been a number of important studies related to the BN minimal controllability problem. For example, \cite{ControlKernel} introduced a general algorithm for computing a control kernel, which is a minimal set of nodes that must be controlled to drive the network to a desirable state. \cite{HouAttractor} studied a special case of the BN control problem in which the target states are restricted to be attractors, and proved under a reasonable assumption that a set of driver nodes of size $O(\log_{2}N + \log_{2}M)$ is enough for controlling BNs, where $N$ is the total number of nodes and $M$ is the number of attractors. \cite{Borriello} analyzed 49 biological network models and found numerical evidence that the average control-kernel size depends logarithmically on the number of attractors. \cite{WeissConjunctive} studied the minimal controllability problem of conjunctive BNs, and used graph theory to come up with a necessary and sufficient condition for the controllability of such BNs as well as a quadratic-time algorithm for testing controllability. \cite{BNStrongStruct} introduced the concept of strong structural controllability of BNs, and 
proved that the minimal strong structural controllability problem of BNs is NP-hard. \cite{ZhuStabilizable} proposed an efficient procedure for determining a minimal set of controlled nodes which leads to the global stabilization of BNs. \cite{SunMinControl} studied the minimal controllability problem of BNs involving the XOR operator, the $k$-ary AND function, the $k$-ary AND function combined with negation and nested canalyzing Boolean functions, and derived general lower bounds, general upper bounds, worst-case lower bounds and best-case upper bounds for such BNs.

However, to the best of our knowledge, there has been no research on the minimal controllability problem of BNs consisting of Boolean threshold functions (which we call threshold BNs). Threshold BNs and its variations can in many cases successfully model the activation and inhibition effects in biological systems \cite{yeastrobust, FissionYeast, FlowerEvo, PBTN, Arabidopsis}. Therefore, we study the minimal controllability problem of threshold BNs in the hope that our results can provide theoretical insights into the control of real biological systems.
Besides models of biological systems, Boolean threshold functions are also used in theoretical models of artificial neural networks \cite{AnthonyNeural, SiuNeural}.

Most of our results on threshold BNs are related to majority-type Boolean functions, which also have important applications such as social network modeling and digital circuit design \cite{TieBreak, MajorityVote, MajorityLogic}. Hence, our results on threshold BNs may also benefit the research in these areas.

Our contributions in this paper can be divided into four parts. The first and the second parts are related to threshold BNs, whereas the third and the fourth parts are related to XOR-BNs.

In the first part of our contribution, we derive general bounds for BNs involving majority-type Boolean functions. More specifically, for any controllable $k$-in-MAJORITY-BN, we derive an inequality for the size of its control-node set (Theorem \ref{thm:k_in_MAJOR_BN_ineq_thm}), and use the inequality to derive a lower bound for the size of its control-node set (Corollary \ref{coroll:k_in_MAJOR_BN_lower_bound}). 
We also derive a similar inequality and a similar lower bound for $2k$-in-MTBI-BNs (Theorem \ref{thm:2k_in_MTBI_BN_ineq_thm}).
We also prove a result about directed graphs with degree constraints (Proposition \ref{prop:special_subset_of_K-K-regular_graph}), and use the result to derive general upper bounds for the sizes of the minimal control-node sets of $(2k+1)$-$(2k+1)$-MAJORITY-BNs, $2k$-$2k$-MAJORITY-BNs and $2k$-$2k$-MTBI-BNs (Theorems \ref{thm:3k2_6k_2_3k2_7k_2} and \ref{thm:3k2_4k_1_3k2_5k_2}).

In the second part of our contribution, we construct special controllable $(2k+1)$-$(2k+1)$-MAJORITY-BNs, $2k$-$2k$-MAJORITY-BNs and $2k$-$2k$-MTBI-BNs (Section \ref{section:majority_BNs_bcub}). In Section \ref{section:threshold_BNs_pos_neg_coef}, we construct special controllable BNs that have degree constraints and involve Boolean threshold functions with both positive and negative coefficients. All of these BNs are associated with small control-node sets, and can serve as best-case upper bounds.

In the third part of our contribution, we provide results which connect the study of the controllability of XOR-BNs with linear algebra. We begin by introducing a new way of looking at the functional dependencies of XOR-BNs as a linear map from $F^{n}_{2}$ to $F^{n}_{2}$, where $F_{2}$ is the finite field of order $2$. Then, we prove a necessary and sufficient condition for the controllability of XOR-BNs (Theorem \ref{thm:Wu_spans_Fn_2}). Afterwards, we provide polynomial-time algorithms for constructing control-node sets and control strategies for general XOR-BNs (Algorithms \ref{alg:construct_control_node_set}, \ref{alg:intermediate}, \ref{alg:control_scheme}).

In the fourth part of our contribution, we present best-case upper bounds for the sizes of the minimal control-node sets of $k$-$k$-XOR-BNs (i.e., XOR-BNs with some particular degree constraints). More specifically, we use Theorem \ref{thm:Wu_spans_Fn_2} to show that for all integers $n$ and $k$ such that $n > k \geq 2$, there exists an $n$-node controllable $k$-$k$-XOR-BN whose control-node set is of size $k-1$ (Theorem \ref{thm:k-k-XOR-BN_(k-1)-bound}). Moreover, we use Theorem \ref{thm:Wu_spans_Fn_2} and ring theory to show that when $n$ and $k$ are integers such that $n \geq 4$ is a power of $2$ and $k \in [3, n - 1]$ is odd, there exists an $n$-node controllable $k$-$k$-XOR-BN whose control-node set is of size $1$ (Theorem \ref{thm:1_control_node}). We remark that both of these results are better than the best-case upper bound results for $k$-$k$-XOR-BNs in \cite{SunMinControl}.

Our major results in this paper are summarized in Table \ref{table:summary}. Please refer to Table \ref{table:meaning_of_bounds} for the meaning of general lower bounds, general upper bounds and best-case upper bounds.

\begin{table}[h!]
\centering
\caption{Summary of our major results in this paper}\label{table:summary}
\footnotesize
\begin{tabular}{|l|l|l|}
\cline{2-3}
\multicolumn{1}{l|}{}	& Brief summary of the result		& \makecell[l]{Nature of \\ the result} \\ \hline
Theorem \ref{thm:k_in_MAJOR_BN_ineq_thm}	&
\makecell[l]{For any $k$-in-MAJORITY-BN controllable within the time $0$ to $s+1$, \\
$m \coloneqq |U|$ satisfies 
$2^{n} \leq 2^{(s+1)m} - \left( \sum^{\left\lceil \frac{k}{2} \right\rceil - 1}_{i = 1} \binom{m}{i} \right) \sum^{s}_{t = 1} 2^{tm}$.}	&
\multirow{2}{*}{\makecell[l]{General \\ lower bound}} \\
Corollary \ref{coroll:k_in_MAJOR_BN_lower_bound}	&
By this inequality, $|U| \geq \max(\left\lceil \frac{k}{2} \right\rceil, \left\lceil \frac{n + 1}{s + 1} \right\rceil)$.	& \\ \hline
Theorem \ref{thm:2k_in_MTBI_BN_ineq_thm}	&
\makecell[l]{For any $2k$-in-MTBI-BN controllable within the time $0$ to $s + 1$, \\ 
$m \coloneqq |U|$ satisfies 
$2^{n} \leq 2^{(s+1)m} - \left( \sum^{k - 1}_{i = 1} \binom{m}{i} \right) \sum^{s}_{t = 1} 2^{tm}$. \\
By this inequality, $|U| \geq \max(k, \left\lceil \frac{n + 1}{s + 1} \right\rceil)$.}		&
\makecell[l]{General \\ lower bound} \\ \hline
Theorem \ref{thm:3k2_6k_2_3k2_7k_2}	&
\makecell[l]{For any $n$-node $(2k+1)$-$(2k+1)$-MAJORITY-BN, there exists \\
a control-node set $U$ such that $|U| \leq \left( \frac{3k^{2} + 6k + 2}{3k^{2} + 7k + 2} \right) n$.}	&
\makecell[l]{General \\ upper bound} \\ \hline
Theorem \ref{thm:3k2_4k_1_3k2_5k_2}	&
\makecell[l]{For any $n$-node $2k$-$2k$-MAJORITY-BN or $2k$-$2k$-MTBI-BN, \\
there exists a control-node set $U$ such that $|U| \leq \left( \frac{3k^{2} + 4k - 1}{3k^{2} + 5k - 2} \right) n$.}	&
\makecell[l]{General \\ upper bound} \\ \hline
Theorem \ref{thm:(2k+1)_bcub}	&
\makecell[l]{For any positive integer $m$, there exists a controllable $(2k+2)m$-node \\
$(2k+1)$-$(2k+1)$-MAJORITY-BN such that $|U| = 2k + 2 + (m - 1)k$.}	&
\makecell[l]{Best-case \\ upper bound} \\ \hline
Theorem \ref{thm:2k_MAJOR_bcub}	&
\makecell[l]{For any positive integer $m$, there exists a controllable $(2k+1)m$-node \\
$2k$-$2k$-MAJORITY-BN such that $|U| = 2k + 1 + (m - 1)k$.}	&
\makecell[l]{Best-case \\ upper bound} \\ \hline
Theorem \ref{thm:2k_MTBI_bcub}	&
\makecell[l]{For any positive integer $m$, there exists a controllable $2mk$-node \\
$2k$-$2k$-MTBI-BN such that $|U| = 2k + (m - 1)(k - 1)$.}	&
\makecell[l]{Best-case \\ upper bound} \\ \hline
Theorem \ref{thm:phi_k_BN_bcub}	&
\makecell[l]{For any positive integer $k \geq 3$, there exists a Boolean threshold \\
function $\varphi_{k} : \{ 0, 1 \}^{k} \to \{ 0, 1 \}$ such that for any positive integer $m$, \\
there exists a controllable $mk$-node $k$-$k$-$\varphi_{k}$-BN satisfying $|U| = k$.}	&
\makecell[l]{Best-case \\ upper bound} \\ \hline
Theorem \ref{thm:Wu_spans_Fn_2}	&
\multicolumn{2}{|l|}{\makecell[l]{A linear-algebraic polynomial-time necessary and sufficient condition for the \\ controllability of XOR-BNs.}} \\ \hline
Algorithm \ref{alg:construct_control_node_set}	&
\multicolumn{2}{|l|}{\makecell[l]{A linear-algebraic polynomial-time algorithm which takes any XOR-BN as input \\
and constructs a control-node set rendering the XOR-BN controllable.}} \\ \hline
Algorithm \ref{alg:control_scheme}	&
\multicolumn{2}{|l|}{\makecell[l]{A linear-algebraic polynomial-time algorithm which takes any two states 
$\mathbf{a}$, $\mathbf{b}$ of \\ 
any controllable XOR-BN as inputs and constructs control signals driving the BN \\
from $\mathbf{a}$ to $\mathbf{b}$.}} \\ \hline
Theorem \ref{thm:k-k-XOR-BN_(k-1)-bound}	&
\makecell[l]{For all integers $n > k \geq 2$, there exists a controllable $n$-node \\
$k$-$k$-XOR-BN such that $|U| = k - 1$.}	&
\makecell[l]{Best-case \\ upper bound} \\ \hline
Theorem \ref{thm:1_control_node}	&
\makecell[l]{Let $m \geq 2$ be an integer and $k \in [3, 2^{m} - 1]$ be an odd integer. \\
There exists a controllable $2^{m}$-node $k$-$k$-XOR-BN such that $|U| = 1$.}	&
\makecell[l]{Best-case \\ upper bound} \\ \hline
\end{tabular}
\end{table}

The rest of this paper is organized as follows. Section \ref{section:prelim} introduces the necessary preliminary concepts studied in this paper. In Section \ref{section:threshold_functions}, we provide our results on the minimal controllability problem of threshold BNs. In Section \ref{section:XOR_BNs}, our results on the minimal controllability of XOR-BNs are given. In Section \ref{section:conclusion}, we provide our conclusion, discuss the potential applications of our work and propose some possible future directions of research.

\section{Preliminaries}\label{section:prelim}

A Boolean network (BN) is a discrete-time dynamical system consisting of a set of nodes $x_{1}$, $x_{2}$, \ldots, $x_{n}$ and Boolean functions $f_{1}, f_{2}, \ldots, f_{n} : \{ 0, 1 \}^{n} \to \{ 0, 1 \}$. For all $t \in \mathbb{Z}_{\geq 0}$, the state of a BN at time $t$ is denoted by $\mathbf{x}(t) \coloneqq (x_{1}(t), x_{2}(t), \ldots, x_{n}(t))$. When we do not exert control signals on the BN, the update rules can be described as follows:
\begin{equation}
x_{i}(t + 1) = f_{i}(x_{1}(t), x_{2}(t), \ldots, x_{n}(t)) \quad \text{for $i = 1, 2, \ldots, n$.}
\end{equation}
Further, we let $F \coloneqq (f_{1}, f_{2}, \ldots, f_{n}) : \{ 0, 1 \}^{n} \to \{ 0, 1 \}^{n}$. 

Let $U \subseteq \{ x_{1}, x_{2}, \ldots, x_{n} \} \eqqcolon V$ be the set of control nodes of a BN. The update rules of  the BN under the influence of control signals are given by
\begin{align}
x_{i}(t + 1) &= f_{i}(x_{1}(t), x_{2}(t), \ldots, x_{n}(t)), \quad \forall x_{i} \notin U, \\
x_{i}(t + 1) &= u_{i}(t) \oplus f_{i}(x_{1}(t), x_{2}(t), \ldots, x_{n}(t)), \quad \forall x_{i} \in U,
\end{align}
where $\oplus$ denotes the XOR operator and for all $t \in \mathbb{Z}_{\geq 0}$, $u_{i}(t)$ can take arbitrary values in the set $\{ 0, 1 \}$.

Now, we provide the following important definition:
\begin{definition}
If an $n$-node BN with control-node set $U$ satisfies that for all $\mathbf{a}, \mathbf{b} \in \{ 0, 1\}^{n}$, there exist control signals which drive the BN from $\mathbf{x}(0) = \mathbf{a}$ to $\mathbf{x}(T) = \mathbf{b}$ for some non-negative integer $T$, then the BN is said to be controllable.
\end{definition}

The minimal controllability problem of BNs is defined below.
\begin{problem}[Minimal Controllability Problem of BNs]\label{problem:BN_min_control}
Given a BN, find the smallest control-node set which renders the BN controllable.
\end{problem}

We also consider three concepts related to the size of the minimal control-node sets of BNs \cite{SunMinControl, SunObserve}.
\begin{table}[h!]
  \centering
  \caption{The meaning of general lower bounds, general upper bounds and best-case  upper bounds, where $b$ is the corresponding bound.}\label{table:meaning_of_bounds}
  \begin{tabular}{|l|l|}
   \hline
   General lower bound & \makecell[l]{For any BN satisfying certain properties, \\ the size of the minimal control-node set is at least $b$.} \\ \hline
   Best-case upper bound & \makecell[l]{For some of the BNs satisfying certain properties, \\ the size of the minimal control-node set is at most $b$, \\ which is currently the smallest.} \\ \hline
   General upper bound & \makecell[l]{For any BN satisfying certain properties, \\ the size of the minimal control-node set is at most $b$.} \\ \hline
  \end{tabular}
\end{table}

\FloatBarrier

\section{Boolean Networks Consisting of Threshold Functions}\label{section:threshold_functions}

In this section, we study the number of control nodes of BNs consisting of Boolean threshold functions. In Section \ref{section:threshold_general_bounds}, we derive general lower and upper bounds for BNs involving majority-type threshold functions. In Section \ref{section:majority_BNs_bcub}, we provide best-case upper bounds for BNs with majority-type threshold functions and degree constraints. In Section \ref{section:threshold_BNs_pos_neg_coef}, we provide a best-case upper bound for BNs with degree constraints which involve Boolean threshold functions having both positive and negative coefficients.

Before we present the results, we have to provide some definitions. 

\begin{definition}\label{def:k-ary_majority_function}
Let $k$ be a positive integer. The function $f : \{ 0, 1 \}^{k} \to \{ 0, 1 \}$ defined as
\begin{equation}\label{eq:k-ary_majority_function}
f(\mathbf{x}) \coloneqq
\begin{cases}
1	& \text{if $x_{1} + x_{2} + \cdots + x_{k} \geq \frac{k}{2}$} \\
0	& \text{otherwise}
\end{cases}
\end{equation}
is called the  $k$-ary majority function.
\end{definition}

\begin{definition}\label{def:2k-ary_MTBI_function}
Let $k$ be a positive integer. The function $g : \{ 0, 1 \}^{2k} \to \{ 0, 1 \}$ defined as
\begin{equation}\label{eq:2k-ary_MTBI_function}
g(\mathbf{x}) \coloneqq
\begin{cases}
1	& \text{if $1.1 x_{1} + x_{2} + x_{3} + \cdots + x_{2k} \geq k + 0.05$} \\
0	& \text{otherwise}
\end{cases}
\end{equation}
is called the $2k$-ary majority function with a tie-breaking input (MTBI). $x_{1}$ is said to be the tie-breaking input of $g$.
Equivalently, 
\begin{equation}\label{eq:2k-ary_MTBI_function_intuition}
g(\mathbf{x}) \coloneqq
\begin{cases}
1	& \text{if at least $k+1$ entries of $\mathbf{x}$ equal $1$} \\
1	& \text{if exactly $k$ entries of $\mathbf{x}$ equal $1$ and $x_{1} = 1$} \\
0	& \text{if exactly $k$ entries of $\mathbf{x}$ equal $1$ and $x_{1} = 0$} \\
0 	& \text{if at most $k - 1$ entries of $\mathbf{x}$ equal $1$}
\end{cases}
\end{equation}
\end{definition}

\begin{definition}\label{def:k_in_MAJORITY_BN}
Let $k$ be a positive integer. If a BN satisfies that all of its Boolean functions are the $k$-ary majority function, then the BN is said to be a $k$-in-MAJORITY-BN.
\end{definition}

\begin{definition}\label{def:2k_in_MTBI_BN}
Let $k$ be a positive integer. If a BN satisfies that all of its Boolean functions are the $2k$-ary MTBI function, then the BN is said to be a $2k$-in-MTBI-BN.
\end{definition}

\begin{definition}\label{def:k_k_MAJORITY_BN}
Let $k$ be a positive integer. If a $k$-in-MAJORITY-BN satisfies that each node of its associated directed graph has in-degree $k$ and out-degree $k$, then the BN is said to be a $k$-$k$-MAJORITY-BN.
\end{definition}

\begin{definition}
Let $k$ be a positive integer. If a $2k$-in-MTBI-BN satisfies that each node of its associated directed graph has in-degree $2k$ and out-degree $2k$, then the BN is said to be a $2k$-$2k$-MTBI-BN.
\end{definition}

\subsection{General Lower and Upper Bounds}\label{section:threshold_general_bounds}

In this section, we derive lower-bound-related inequalities for the number of control nodes of $k$-in-MAJORITY-BNs and $2k$-in-MTBI-BNs. We also derive general upper bounds for the number of control nodes of $(2k+1)$-$(2k+1)$-MAJORITY-BNs, $2k$-$2k$-MAJORITY-BNs and $2k$-$2k$-MTBI-BNs.

First, we present our results for lower bounds for $k$-in-MAJORITY-BNs and $2k$-in-MTBI-BNs.

\begin{theorem}\label{thm:k_in_MAJOR_BN_ineq_thm}
Fix arbitrary positive integers $n \geq k \geq 3$, $s \geq 1$.
Consider any $n$-node controllable $k$-in-MAJORITY-BN such that for all states
$\mathbf{a} \neq \mathbf{b} \in \{ 0, 1 \}^{n}$, there exists a control scheme 
$\mathbf{u}(0)$, $\mathbf{u}(1)$, \ldots, $\mathbf{u}(t)$ (where $t \leq s$) which drives the BN from $\mathbf{a}$ (time $0$) to $\mathbf{b}$ (time $t+1$). Let $U \subseteq \{ x_{1}, x_{2}, \ldots, x_{n} \}$ be the control-node set of the BN and $m \coloneqq |U|$.
Then, $m \geq \left\lceil \frac{k}{2} \right\rceil$. 
Moreover, $m$ satisfies the inequality
\begin{equation}\label{eq:k_in_MAJOR_BN_ineq_thm}
2^{n} \leq 2^{(s+1)m} - \left( \sum^{\left\lceil \frac{k}{2} \right\rceil - 1}_{i = 1} \binom{m}{i} \right) \sum^{s}_{t = 1} 2^{tm}.
\end{equation}
\end{theorem}

\begin{proof}
It is easy to see that $m = |U| \geq \left\lceil \frac{k}{2} \right\rceil$ must hold; otherwise, 
if $\mathbf{x}(0) = \mathbf{0}_{n}$, then for all $t \in \mathbb{Z}_{\geq 0}$, $F(\mathbf{x}(t)) = \mathbf{0}_{n}$, and hence there is no control scheme which drives the BN from $\mathbf{0}_{n}$ to $\mathbf{1}_{n}$.

Next, we prove the inequality given by Eq.\@ (\ref{eq:k_in_MAJOR_BN_ineq_thm}).

First, we note that if $m = n$, then
\begin{align}
2^{(s+1)m} - \left( \sum^{\left\lceil \frac{k}{2} \right\rceil - 1}_{i = 1} \binom{m}{i} \right) \sum^{s}_{t = 1} 2^{tm}
&= 2^{(s+1)n} - \left( \sum^{n}_{i = 0} \binom{n}{i} \right) \sum^{s}_{t = 1} 2^{tn} + \sum^{s}_{t = 1} 2^{tn}
+ \left( \sum^{n}_{i = \left\lceil \frac{k}{2} \right\rceil} \binom{n}{i} \right) \sum^{s}_{t = 1} 2^{tn} \\
&= 2^{(s+1)n} - \sum^{s+1}_{t = 2} 2^{tn} + \sum^{s}_{t = 1} 2^{tn}
+ \left( \sum^{n}_{i = \left\lceil \frac{k}{2} \right\rceil} \binom{n}{i} \right) \sum^{s}_{t = 1} 2^{tn} \\
&= 2^{n} + \left( \sum^{n}_{i = \left\lceil \frac{k}{2} \right\rceil} \binom{n}{i} \right) \sum^{s}_{t = 1} 2^{tn} \\
&> 2^{n}.
\end{align}
Hence, the inequality of the theorem holds. From now on, we assume that $\left\lceil \frac{k}{2} \right\rceil \leq m < n$.

WLOG, we assume that $U = \{ x_{1}, x_{2}, \ldots, x_{m} \}$. For all integers $t \in [0, s]$, let $A_{t}$ be the set of all states $\mathbf{b} \in \{ 0, 1 \}^{n}$ such that there exists a control scheme $\mathbf{u}(0)$, $\mathbf{u}(1)$, \ldots, $\mathbf{u}(t)$ which drives the BN from $\mathbf{0}_{n}$ (time 0) to $\mathbf{b}$ (time $t+1$). Define
$B_{0} \coloneqq A_{0}$, $B_{1} \coloneqq A_{1} \setminus A_{0}$, \ldots, 
$B_{s} \coloneqq A_{s} \setminus \left( \bigcup^{s-1}_{i = 0} A_{i} \right)$.
Note that $F(\mathbf{0}_{n}) = \mathbf{0}_{n}$. Therefore,
\begin{equation}\label{eq:explicit_B0_aka_A0}
B_{0} = A_{0} = \{ (a_{1}, a_{2}, \ldots, a_{m}, 0, 0, \ldots, 0) \in \{ 0, 1 \}^{n} : a_{1}, a_{2}, \ldots, a_{m} \in \{ 0, 1\} \}.
\end{equation}
So, $|B_{0}| = 2^{m}$.
On the other hand, by the assumption of the theorem and the fact that $\mathbf{0}_{n} \in B_{0}$,
\begin{equation}\label{eq:union_At_union_Bt}
\{ 0, 1 \}^{n} = \bigcup^{s}_{t = 0} A_{t} = \bigcup^{s}_{t = 0} B_{t}.
\end{equation}

Now, we are going to prove by induction that for $t = 1, 2, \ldots, s$,
\begin{equation}\label{eq:induction_Bt_ineq}
|B_{t}| \leq 2^{tm}\left( 2^{m} - \sum^{\left\lceil \frac{k}{2} \right\rceil - 1}_{i = 0} \binom{m}{i} \right).
\end{equation}
Let's call this statement $P(t)$.

We prove the base case first. Fix arbitrary 
$\mathbf{b} = (b_{1}, b_{2}, \ldots, b_{m}, b_{m+1}, \ldots, b_{n}) \in B_{1} = A_{1} \setminus A_{0}$.
Because $\mathbf{b} \notin A_{0}$, some of $b_{m+1}$, $b_{m+2}$, \ldots, $b_{n}$ equal $1$.
Moreover, there exist $\mathbf{a} \in A_{0} = B_{0}$, $c_{1}, c_{2}, \ldots, c_{m} \in \{ 0, 1 \}$ such that
$F(\mathbf{a}) = (c_{1}, \ldots, c_{m}, b_{m+1}, \ldots, b_{n})$.
Note that $\mathbf{a}$ contains at least $\left\lceil \frac{k}{2} \right\rceil$ entries of 1s; otherwise,
$F(\mathbf{a}) = \mathbf{0}_{n}$ because each constituent Boolean function of the BN is the $k$-ary majority function.
Define $G : \{ 0, 1 \}^{n} \to \{ 0, 1 \}^{n-m}$ such that 
$G(y_{1}, y_{2}, \ldots, y_{n}) \coloneqq (y_{m+1}, y_{m+2}, \ldots, y_{n})$.
Let
\begin{equation}\label{}
C_{0} \coloneqq \left\{ \mathbf{y} \in B_{0} : \text{$\mathbf{y}$ contains at least $\left\lceil \frac{k}{2} \right\rceil$ $1$s} \right\}.
\end{equation}
We see that $G(B_{1}) \subseteq G(F(C_{0}))$.
Hence,
\begin{equation}\label{eq:bound_G(B1)_from_above}
|G(B_{1})| \leq |G(F(C_{0}))| \leq |C_{0}| = 2^{m} - \sum^{\left\lceil \frac{k}{2} \right\rceil - 1}_{i = 0} \binom{m}{i},
\end{equation}
which implies
\begin{equation}\label{eq:bound_B1_from_above}
|B_{1}| = 2^{m} |G(B_{1})| \leq 2^{m} \left[ 2^{m} - \sum^{\left\lceil \frac{k}{2} \right\rceil - 1}_{i = 0} \binom{m}{i} \right].
\end{equation}
Therefore, $P(1)$ is true.

Now, let's prove the inductive step. Assume that $P(t)$ is true for some $t \in \{ 1, 2, \ldots, s - 1 \}$.
Hence, $|B_{t}| \leq 2^{tm} \left( 2^{m} - \sum^{\left\lceil \frac{k}{2} \right\rceil - 1}_{i = 0} \binom{m}{i} \right)$.
Fix arbitrary $\mathbf{d} = (d_{1}, d_{2}, \ldots, d_{m}, d_{m+1}, \ldots, d_{n}) \in B_{t+1}$.
Then, there exist $\mathbf{v} \in B_{t}$, $w_{1}, w_{2}, \ldots, w_{m} \in \{ 0, 1 \}$ such that
$F(\mathbf{v}) = (w_{1}, w_{2}, \ldots, w_{m}, d_{m+1}, \ldots, d_{n})$. Hence, $G(B_{t+1}) \subseteq G(F(B_{t}))$.
Then,
\begin{equation}\label{eq:bound_G(B_{t+1})_from_above}
|G(B_{t+1})| \leq |G(F(B_{t}))| \leq |B_{t}| 
\leq 2^{tm} \left( 2^{m} - \sum^{\left\lceil \frac{k}{2} \right\rceil - 1}_{i = 0} \binom{m}{i} \right),
\end{equation}
which implies
\begin{equation}\label{eq:bound_Btp1_from_above}
|B_{t+1}| = 2^{m} |G(B_{t+1})| \leq 2^{(t+1)m} \left( 2^{m} - \sum^{\left\lceil \frac{k}{2} \right\rceil - 1}_{i = 0} \binom{m}{i} \right).
\end{equation}
Therefore, $P(t+1)$ is true. This completes the inductive step. Hence, for $t = 1, 2, \ldots, s$, $P(t)$ is true.

By Eq.\@ (\ref{eq:union_At_union_Bt}),
\begin{align}
2^{n} = |\{ 0, 1 \}^{n}| = \left| \bigcup^{s}_{t = 0} B_{t} \right| \leq \sum^{s}_{t = 0} |B_{t}|
&= |B_{0}| + \sum^{s}_{t = 1} |B_{t}| \\
&\leq 2^{m} + \sum^{s}_{t = 1} 2^{tm}\left( 2^{m} - \sum^{\left\lceil \frac{k}{2} \right\rceil - 1}_{i = 0} \binom{m}{i} \right) \\
&= 2^{m} + \left( \sum^{s+1}_{j = 2} 2^{jm} \right) - \left( \sum^{\left\lceil \frac{k}{2} \right\rceil - 1}_{i = 0} \binom{m}{i} \right) \sum^{s}_{t = 1} 2^{tm} \\
&= 2^{(s+1)m} - \left( \sum^{\left\lceil \frac{k}{2} \right\rceil - 1}_{i = 1} \binom{m}{i} \right) \sum^{s}_{t = 1} 2^{tm}
\end{align}
which gives Eq.\@ (\ref{eq:k_in_MAJOR_BN_ineq_thm}). The proof is complete.
\end{proof}

\begin{corollary}\label{coroll:k_in_MAJOR_BN_lower_bound}
Fix arbitrary positive integers $n \geq k \geq 3$, $s \geq 1$.
Consider any $n$-node controllable $k$-in-MAJORITY-BN such that for all states
$\mathbf{a} \neq \mathbf{b} \in \{ 0, 1 \}^{n}$, there exists a control scheme 
$\mathbf{u}(0)$, $\mathbf{u}(1)$, \ldots, $\mathbf{u}(t)$ (where $t \leq s$) which drives the BN from $\mathbf{a}$ (time $0$) to $\mathbf{b}$ (time $t+1$). Let $U \subseteq \{ x_{1}, x_{2}, \ldots, x_{n} \}$ be the control-node set of the BN and $m \coloneqq |U|$.
Then, $m \geq \max(\left\lceil \frac{k}{2} \right\rceil, \left\lceil \frac{n + 1}{s + 1} \right\rceil)$.
\end{corollary}

\begin{proof}
In Theorem \ref{thm:k_in_MAJOR_BN_ineq_thm}, we have already shown that $m \geq \left\lceil \frac{k}{2} \right\rceil$. Assume for a contradiction that $(s + 1) m \leq n$. By Theorem \ref{thm:k_in_MAJOR_BN_ineq_thm},
\begin{equation}\label{eq:k_in_MAJOR_BN_lower_bound_conrtadiction}
2^{n} 
\leq 2^{(s+1)m} - \left( \sum^{\left\lceil \frac{k}{2} \right\rceil - 1}_{i = 1} \binom{m}{i} \right) \sum^{s}_{t = 1} 2^{tm} 
< 2^{(s + 1)m} \leq 2^{n},
\end{equation}
which gives a contradiction. Hence, $(s + 1) m \geq n + 1$, which implies 
$m \geq \left\lceil \frac{n + 1}{s + 1} \right\rceil$.
\end{proof}

By using the same line of reasoning, we can prove the following result for $2k$-in-MTBI-BNs:

\begin{theorem}\label{thm:2k_in_MTBI_BN_ineq_thm}
Let $n$, $k$, $s$ be positive integers such that $n \geq 2k \geq 4$.
Consider any $n$-node controllable $2k$-in-MTBI-BN such that for all states
$\mathbf{a} \neq \mathbf{b} \in \{ 0, 1 \}^{n}$, there exists a control scheme 
$\mathbf{u}(0)$, $\mathbf{u}(1)$, \ldots, $\mathbf{u}(t)$ (where $t \leq s$) which drives the BN from $\mathbf{a}$ (time $0$) to $\mathbf{b}$ (time $t+1$). Let $U \subseteq \{ x_{1}, x_{2}, \ldots, x_{n} \}$ be the control-node set of the BN and $m \coloneqq |U|$.
Then, $m \geq k$. 
Moreover, $m$ satisfies the inequality
\begin{equation}\label{eq:2k_in_MTBI_BN_ineq_thm}
2^{n} \leq 2^{(s+1)m} - \left( \sum^{k - 1}_{i = 1} \binom{m}{i} \right) \sum^{s}_{t = 1} 2^{tm}.
\end{equation}
Hence, $m \geq \max(k, \left\lceil \frac{n + 1}{s + 1} \right\rceil)$.
\end{theorem}

Next, we prove some upper bound theorems for $(2k+1)$-$(2k+1)$-MAJORITY-BNs, $2k$-$2k$-MAJORITY-BNs and $2k$-$2k$-MTBI-BNs. To derive these results, the following proposition is required.

\begin{proposition}\label{prop:special_subset_of_K-K-regular_graph}
Let $n$, $K$, $L$ be integers such that $n \geq K \geq 1$ and $L \in [0, K - 1]$.
Let $G = (V, E)$ be an $n$-node directed graph such that each node has in-degree $K$ and out-degree $K$ and for all nodes $x \in V$, there exists at most one self-loop from $x$ to itself.
Let $\Gamma^{-}(x)$ denote the set of all in-neighbours of the node $x$ in $G$.
If $R \subseteq V$ satisfies that for all $y \in R$, 
$\left| (\Gamma^{-}(y) \cap R) \setminus \{ y \} \right| \leq L$,
then $|R| \leq \left( \frac{K}{2K - L - 1} \right) n$.
\end{proposition}

\begin{proof}
For all $X, Y \subseteq V$, let
$E(X, Y) \coloneqq \{ (x_{i}, x_{j}) \in E : \text{$x_{i} \in X$ and $x_{j} \in Y$} \}$.
Because for all $y \in R$, there exist at most one self-loop and at most $L$ edges from $R \setminus \{ y \}$ to $y$,
\begin{equation}\label{eq:E(R,R)_upper_bound}
|E(R, R)| \leq (L + 1) |R|.
\end{equation}
On the other hand, the set $E(R, V)$ is the disjoint union of $E(R, V \setminus R)$ and $E(R, R)$. Therefore,
\begin{equation}
K|R| = |E(R, V)| = |E(R, V \setminus R)| + |E(R, R)|,
\end{equation}
which implies
\begin{align}
|E(R, V \setminus R)| 
&= K|R| - |E(R, R)| \\
&\geq K|R| - (L + 1)|R| \quad \text{(by (\ref{eq:E(R,R)_upper_bound}))} \\
&= (K - L - 1)|R|. \label{eq:2nd_important_ineq}
\end{align}
Furthermore, each vertex in $V \setminus R$ has exactly $K$ in-neighbours in $G$. Therefore, 
\begin{equation}\label{eq:3rd_important_ineq}
|E(R, V \setminus R)| \leq |E(V, V \setminus R)| = K|V \setminus R| = K(n - |R|).
\end{equation}
By Eq.\@ (\ref{eq:2nd_important_ineq}) and (\ref{eq:3rd_important_ineq}),
\begin{align}
&(K - L - 1)|R| \leq |E(R, V \setminus R)| \leq K(n - |R|) \\
\Rightarrow\ & (2K - L - 1)|R| \leq Kn \\
\Rightarrow\ & |R| \leq \left( \frac{K}{2K - L - 1} \right) n.
\end{align}
The proof is complete.
\end{proof}

We are now ready to present the upper bound theorems.

\begin{theorem}\label{thm:3k2_6k_2_3k2_7k_2}
Fix arbitrary positive integers $k$, $n$ such that $n \geq 2k + 1$.
For any $n$-node $(2k+1)$-$(2k+1)$-MAJORITY-BN, there exists a control-node set 
$U \subseteq \{ x_{1}, x_{2}, \ldots, x_{n} \}$ with size at most 
$\left( \frac{3k^{2} + 6k + 2}{3k^{2} + 7k + 2} \right) n$ such that for all states $\mathbf{a}, \mathbf{b} \in \{ 0, 1 \}^{n}$, there exists a control scheme $\mathbf{u}(0)$, $\mathbf{u}(1)$ which drives the BN from $\mathbf{a}$ (time $0$) to $\mathbf{b}$ (time $2$).
\end{theorem}

\begin{proof}
Let $G = (V, E)$ be the directed graph corresponding to the structure of the $(2k+1)$-$(2k+1)$-MAJORITY-BN. To form the required control-node set $U$, we carry out the following procedure:
\begin{description}
\item[Step 1.] Set $A$ to an empty list. Set $B$ to an empty list. Set $R \gets V$ and $i \gets 1$.

\item[Step 2.] Check whether there exist distinct nodes $x_{1}, x_{2}, \ldots, x_{k+1}, y \in R$ such that
$\{ x_{1}, x_{2}, \ldots, x_{k+1} \} \subseteq \Gamma^{-}(y)$ in $G$. If so, go to step 3. If not, go to step 7.

\item[Step 3.] Set $x^{i}_{1} \gets x_{1}$, $x^{i}_{2} \gets x_{2}$, \ldots, $x^{i}_{k+1} \gets x_{k+1}$, $y^{i} \gets y$.

\item[Step 4.] Append the list $[x^{i}_{1}, x^{i}_{2}, \ldots, x^{i}_{k+1}]$ to $A$ and append $y^{i}$ to $B$.

\item[Step 5.] Set $R \gets R \setminus \{ x^{i}_{1}, x^{i}_{2}, \ldots, x^{i}_{k+1}, y^{i} \}$.

\item[Step 6.] Set $i \gets i + 1$. Go to step 2.

\item[Step 7.] Output $A$, $B$, $R$.
\end{description}

Write $A = [[x^{1}_{1}, x^{1}_{2}, \ldots, x^{1}_{k+1}], [x^{2}_{1}, x^{2}_{2}, \ldots, x^{2}_{k+1}], \ldots, 
[x^{p}_{1}, x^{p}_{2}, \ldots, x^{p}_{k+1}]]$ and $B = [y^{1}, y^{2}, \ldots, y^{p}]$.
Note that $x^{1}_{1}$, $x^{1}_{2}$, \ldots, $x^{1}_{k+1}$, $x^{2}_{1}$, $x^{2}_{2}$, \ldots, $x^{2}_{k+1}$, \ldots, 
$x^{p}_{1}$, $x^{p}_{2}$, \ldots, $x^{p}_{k+1}$, $y^{1}$, $y^{2}$, \ldots, $y^{p}$ are distinct nodes and 
for $i = 1, 2, \ldots, p$, $\{ x^{i}_{1}, x^{i}_{2}, \ldots, x^{i}_{k+1} \} \subseteq \Gamma^{-}(y^{i})$. Moreover, the output $R$ satisfies that for all $z \in R$, $|(\Gamma^{-}(z) \cap R) \setminus \{ z \}| \leq k$ in $G$.
Furthermore, each node of $G$ has in-degree $2k+1$ and out-degree $2k+1$, and for all $x \in V$, there exists at most one self-loop in $G$ from $x$ to itself. Therefore, we can apply Proposition \ref{prop:special_subset_of_K-K-regular_graph} to deduce that
\begin{equation}\label{eq:(2k+1)_(3k+1)_bound}
|R| \leq \left( \frac{2k+1}{2(2k+1) - k - 1} \right)n = \left( \frac{2k+1}{3k+1} \right) n.
\end{equation}

We define the control-node set 
$U \coloneqq \{ x^{i}_{1}, x^{i}_{2}, \ldots, x^{i}_{k+1} : i = 1, 2, \ldots, p \} \cup R 
= V \setminus \{ y^{1}, y^{2}, \ldots, y^{p} \}$. Note that
\begin{align}
|U| 
= \frac{k+1}{k+2} |V \setminus R| + |R|
&= \frac{k+1}{k+2}(n - |R|) + |R| \\
&= \left( \frac{k+1}{k+2} \right)n + \frac{1}{k+2}|R| \\
&\leq \left( \frac{k+1}{k+2} \right)n + \frac{2k+1}{(k+2)(3k+1)} n  
\quad \text{(by Eq. (\ref{eq:(2k+1)_(3k+1)_bound}))}\\
&= \left(\frac{3k^{2} + 6k + 2}{3k^{2} + 7k + 2}\right) n. \label{eq:proved_3k2+7k+2_bound}
\end{align}

Finally, we show that for all states $\mathbf{a}, \mathbf{b} \in \{ 0, 1 \}^{n}$, there exists a control scheme 
$\mathbf{u}(0)$, $\mathbf{u}(1)$ (based on the control-node set $U$) which drives the BN 
from $\mathbf{a}$ (time $0$) to $\mathbf{b}$ (time $2$). Let $b^{1}$, $b^{2}$, \ldots, $b^{p}$ be the states of the nodes $y^{1}$, $y^{2}$, \ldots, $y^{p}$ in $\mathbf{b}$. Consider the evolution of the BN with initial state
$\mathbf{x}(0) = \mathbf{a}$. We can pick $\mathbf{u}(0)$ in a way that for $i = 1, 2, \ldots, p$,
$x^{i}_{1}(1) = x^{i}_{2}(1) = \ldots = x^{i}_{k+1}(1) = b^{i}$ in $\mathbf{x}(1)$. Then, $F(\mathbf{x}(1))$ satisfies that for $i = 1, 2, \ldots, p$, the node $y^{i}$ takes up the state $b^{i}$ 
because $x^{i}_{1}$, $x^{i}_{2}$, \ldots, $x^{i}_{k+1}$ are distinct nodes in $\Gamma^{-}(y^{i})$ and $|\Gamma^{-}(y^{i})| = 2k+1$. Now, we pick $\mathbf{u}(1)$ in a way that for all $x \in U$, $x(2)$ takes up the state of node $x$ in $\mathbf{b}$. Clearly, $\mathbf{x}(2) = \mathbf{b}$. 
The proof is complete.
\end{proof}

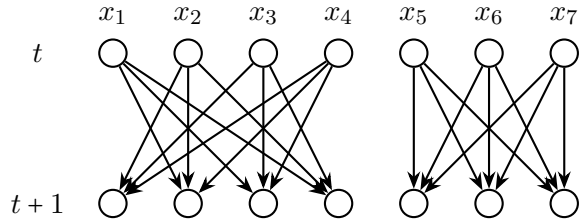
\begin{figure}[h]
\centering
\begin{tikzpicture}
\begin{scope}[every node/.style={circle,thick,draw}]
    \node (x1) at (0,2) {};
    \node (x2) at (1,2) {};
    \node (x3) at (2,2) {};
    \node (x4) at (3,2) {};
    \node (x5) at (4,2) {};
    \node (x6) at (5,2) {};
    \node (x7) at (6,2) {};
    \node (y1) at (0,0) {};
    \node (y2) at (1,0) {};
    \node (y3) at (2,0) {};
    \node (y4) at (3,0) {};
    \node (y5) at (4,0) {};
    \node (y6) at (5,0) {};
    \node (y7) at (6,0) {};
\end{scope}

\begin{scope}
    \node (L1) at (0,2.5) {$x_{1}$};
    \node (L2) at (1,2.5) {$x_{2}$};
    \node (L3) at (2,2.5) {$x_{3}$};
    \node (L4) at (3,2.5) {$x_{4}$};
    \node (L5) at (4,2.5) {$x_{5}$};
    \node (L6) at (5,2.5) {$x_{6}$};
    \node (L7) at (6,2.5) {$x_{7}$};
    \node (T) at (-1,2) {$t$};
    \node (Tp1) at (-1,0) {$t + 1$};
\end{scope}

\begin{scope}[>=Stealth,thick]
    \draw [->] (x1) -- (y2);
    \draw [->] (x1) -- (y3);
    \draw [->] (x1) -- (y4);
    \draw [->] (x2) -- (y1);
    \draw [->] (x2) -- (y2);
    \draw [->] (x2) -- (y4);
    \draw [->] (x3) -- (y1);
    \draw [->] (x3) -- (y3);
    \draw [->] (x3) -- (y4);
    \draw [->] (x4) -- (y1);
    \draw [->] (x4) -- (y2);
    \draw [->] (x4) -- (y3);
    \draw [->] (x5) -- (y5);
    \draw [->] (x5) -- (y6);
    \draw [->] (x5) -- (y7);
    \draw [->] (x6) -- (y5);
    \draw [->] (x6) -- (y6);
    \draw [->] (x6) -- (y7);
    \draw [->] (x7) -- (y5);
    \draw [->] (x7) -- (y6);
    \draw [->] (x7) -- (y7);
\end{scope}
\end{tikzpicture}
\caption{A 7-node 3-3-MAJORITY-BN}\label{fig:7_node_3_3_MAJORITY_BN}
\end{figure}

\FloatBarrier

\begin{example}
Consider the 7-node 3-3-MAJORITY-BN given in Fig.\@ \ref{fig:7_node_3_3_MAJORITY_BN}. After we execute the 7-step procedure in the proof of Theorem \ref{thm:3k2_6k_2_3k2_7k_2}, we get the outputs
$A = [[x_{2}, x_{3}], [x_{5}, x_{6}]]$, $B = [x_{1}, x_{7}]$, $R = \{ x_{4} \}$.
Let $x^{1}_{1} \coloneqq x_{2}$, $x^{1}_{2} \coloneqq x_{3}$, $x^{2}_{1} \coloneqq x_{5}$, 
$x^{2}_{2} \coloneqq x_{6}$, $y^{1} \coloneqq x_{1}$, $y^{2} \coloneqq x_{7}$. Note that $x^{1}_{1}$, $x^{1}_{2}$,
$x^{2}_{1}$, $x^{2}_{2}$, $y^{1}$, $y^{2}$ are distinct nodes of the BN, and that 
$\{ x^{1}_{1}, x^{1}_{2} \} \subseteq \Gamma^{-}(y^{1})$ and
$\{ x^{2}_{1}, x^{2}_{2} \} \subseteq \Gamma^{-}(y^{2})$. The control-node set $U$ is set to
$\{ x^{1}_{1}, x^{1}_{2}, x^{2}_{1}, x^{2}_{2} \} \cup R = \{ x_{2}, x_{3}, x_{4}, x_{5}, x_{6} \}$.

Suppose that we want to drive the BN from the initial state $\mathbf{x}(0) = 1111000$ to the target state 
$\mathbf{x}^{T} = 0101011$. At time $t = 0$, we let $u_{2}(0) = 1$, $u_{3}(0) = 1$, $u_{4}(0) = 0$, $u_{5}(0) = 1$,
$u_{6}(0) = 1$. At time $t = 1$, we let $u_{2}(1) = 0$, $u_{3}(1) = 1$, $u_{4}(1) = 1$, $u_{5}(1) = 1$,
$u_{6}(1) = 0$. Then, we have the following.
\begin{center}
\begin{tabular}{l|lllllll}
\hline
$\mathbf{x}(0)$		& 1 & 1 & 1 & 1 & 0 & 0 & 0 \\
$F(\mathbf{x}(0))$ 	& 1 & 1 & 1 & 1 & 0 & 0 & 0 \\
$\mathbf{x}(1)$ 		& 1 & 0 & 0 & 1 & 1 & 1 & 0 \\
$F(\mathbf{x}(1))$ 	& 0 & 1 & 1 & 0 & 1 & 1 & 1 \\
$\mathbf{x}(2)$ 		& 0 & 1 & 0 & 1 & 0 & 1 & 1 \\
\hline
\end{tabular}
\end{center}
\end{example}

\begin{theorem}\label{thm:3k2_4k_1_3k2_5k_2}
Fix arbitrary positive integers $k$, $n$ such that $n \geq 2k \geq 4$.
For any $n$-node $2k$-$2k$-MAJORITY-BN, there exists a control-node set 
$U \subseteq \{ x_{1}, x_{2}, \ldots, x_{n} \}$ with size at most 
$\left( \frac{3k^{2} + 4k - 1}{3k^{2} + 5k - 2} \right) n$ such that for all states $\mathbf{a}, \mathbf{b} \in \{ 0, 1 \}^{n}$, there exists a control scheme $\mathbf{u}(0)$, $\mathbf{u}(1)$ which drives the BN from $\mathbf{a}$ (time $0$) to $\mathbf{b}$ (time $2$).
This theorem also applies to $2k$-$2k$-MTBI-BNs.
\end{theorem}

\begin{proof}
We will focus on the case of $2k$-$2k$-MAJORITY-BN, since the proof for $2k$-$2k$-MTBI-BN is exactly the same.

Let $G = (V, E)$ be the directed graph corresponding to the structure of the $2k$-$2k$-MAJORITY-BN. To form the required control-node set $U$, we carry out the following procedure:
\begin{description}
\item[Step 1.] Set $A$ to an empty list. Set $B$ to an empty list. Set $R \gets V$ and $i \gets 1$.

\item[Step 2.] Check whether there exist distinct nodes $x_{1}, x_{2}, \ldots, x_{k+1}, y \in R$ such that
$\{ x_{1}, x_{2}, \ldots, x_{k+1} \} \subseteq \Gamma^{-}(y)$ in $G$. If so, go to step 3. If not, go to step 7.

\item[Step 3.] Set $x^{i}_{1} \gets x_{1}$, $x^{i}_{2} \gets x_{2}$, \ldots, $x^{i}_{k+1} \gets x_{k+1}$, $y^{i} \gets y$.

\item[Step 4.] Append the list $[x^{i}_{1}, x^{i}_{2}, \ldots, x^{i}_{k+1}]$ to $A$ and append $y^{i}$ to $B$.

\item[Step 5.] Set $R \gets R \setminus \{ x^{i}_{1}, x^{i}_{2}, \ldots, x^{i}_{k+1}, y^{i} \}$.

\item[Step 6.] Set $i \gets i + 1$. Go to step 2.

\item[Step 7.] Output $A$, $B$, $R$.
\end{description}

Write $A = [[x^{1}_{1}, x^{1}_{2}, \ldots, x^{1}_{k+1}], [x^{2}_{1}, x^{2}_{2}, \ldots, x^{2}_{k+1}], \ldots, 
[x^{p}_{1}, x^{p}_{2}, \ldots, x^{p}_{k+1}]]$ and $B = [y^{1}, y^{2}, \ldots, y^{p}]$.
Note that $x^{1}_{1}$, $x^{1}_{2}$, \ldots, $x^{1}_{k+1}$, $x^{2}_{1}$, $x^{2}_{2}$, \ldots, $x^{2}_{k+1}$, \ldots, 
$x^{p}_{1}$, $x^{p}_{2}$, \ldots, $x^{p}_{k+1}$, $y^{1}$, $y^{2}$, \ldots, $y^{p}$ are distinct nodes and 
for $i = 1, 2, \ldots, p$, $\{ x^{i}_{1}, x^{i}_{2}, \ldots, x^{i}_{k+1} \} \subseteq \Gamma^{-}(y^{i})$. Moreover, the output $R$ satisfies that for all $z \in R$, $|(\Gamma^{-}(z) \cap R) \setminus \{ z \}| \leq k$ in $G$.
Furthermore, each node of $G$ has in-degree $2k$ and out-degree $2k$, and for all $x \in V$, there exists at most one self-loop in $G$ from $x$ to itself. Therefore, we can apply Proposition \ref{prop:special_subset_of_K-K-regular_graph} to deduce that
\begin{equation}\label{eq:(2k)_(3k-1)_bound}
|R| \leq \left( \frac{2k}{2(2k) - k - 1} \right)n = \left( \frac{2k}{3k - 1} \right) n.
\end{equation}

We define the control-node set 
$U \coloneqq \{ x^{i}_{1}, x^{i}_{2}, \ldots, x^{i}_{k+1} : i = 1, 2, \ldots, p \} \cup R 
= V \setminus \{ y^{1}, y^{2}, \ldots, y^{p} \}$.
Using Eq.\@ (\ref{eq:(2k)_(3k-1)_bound}), we can show that
\begin{equation}\label{eq:proved_3k2+5k-2_bound}
|U| 
= \frac{k+1}{k+2} |V \setminus R| + |R|
= \frac{k+1}{k+2}(n - |R|) + |R|
\leq \left(\frac{3k^{2} + 4k - 1}{3k^{2} + 5k - 2}\right) n. 
\end{equation}

The proof that for all states $\mathbf{a}, \mathbf{b} \in \{ 0, 1 \}^{n}$, there exists a control scheme 
$\mathbf{u}(0)$, $\mathbf{u}(1)$ (based on the control-node set $U$) driving the BN 
from $\mathbf{a}$ (time $0$) to $\mathbf{b}$ (time $2$) is the same as the corresponding justification in the proof of Theorem \ref{thm:3k2_6k_2_3k2_7k_2}.
\end{proof}

\begin{example}
Let $f : \{ 0, 1 \}^{4} \to \{ 0, 1 \}$ be the 4-ary majority function. Consider the $8$-node 4-4-MAJORITY-BN defined  as:
\begin{align}
x_{1}(t+1) &= f(x_{1}(t), x_{2}(t), x_{3}(t), x_{4}(t)) \\
x_{2}(t+1) &= f(x_{2}(t), x_{3}(t), x_{4}(t), x_{5}(t)) \\
x_{3}(t+1) &= f(x_{3}(t), x_{4}(t), x_{5}(t), x_{6}(t)) \\
x_{4}(t+1) &= f(x_{4}(t), x_{5}(t), x_{6}(t), x_{7}(t)) \\
x_{5}(t+1) &= f(x_{5}(t), x_{6}(t), x_{7}(t), x_{8}(t)) \\
x_{6}(t+1) &= f(x_{6}(t), x_{7}(t), x_{8}(t), x_{1}(t)) \\
x_{7}(t+1) &= f(x_{7}(t), x_{8}(t), x_{1}(t), x_{2}(t)) \\
x_{8}(t+1) &= f(x_{8}(t), x_{1}(t), x_{2}(t), x_{3}(t))
\end{align}
After we execute the 7-step procedure in the proof of Theorem \ref{thm:3k2_4k_1_3k2_5k_2}, we get the outputs
$A = [[x_{2}, x_{3}, x_{4}], [x_{6}, x_{7}, x_{8}]]$, $B = [x_{1}, x_{5}]$, $R = \varnothing$.
Let 
$x^{1}_{1} \coloneqq x_{2}$, 
$x^{1}_{2} \coloneqq x_{3}$, 
$x^{1}_{3} \coloneqq x_{4}$, 
$x^{2}_{1} \coloneqq x_{6}$, 
$x^{2}_{2} \coloneqq x_{7}$,
$x^{2}_{3} \coloneqq x_{8}$,  
$y^{1} \coloneqq x_{1}$, 
$y^{2} \coloneqq x_{5}$. 
Note that $x^{1}_{1}$, $x^{1}_{2}$, $x^{1}_{3}$, $x^{2}_{1}$, $x^{2}_{2}$, $x^{2}_{3}$, $y^{1}$, $y^{2}$ are distinct nodes of the BN, and that 
$\{ x^{1}_{1}, x^{1}_{2}, x^{1}_{3} \} \subseteq \Gamma^{-}(y^{1})$ and
$\{ x^{2}_{1}, x^{2}_{2}, x^{2}_{3} \} \subseteq \Gamma^{-}(y^{2})$. 
The control-node set $U$ is set to
$\{ x^{1}_{1}, x^{1}_{2}, x^{1}_{3}, x^{2}_{1}, x^{2}_{2}, x^{2}_{3} \} \cup R = \{ x_{2}, x_{3}, x_{4}, x_{6}, x_{7}, x_{8} \}$.

Suppose that we want to drive the BN from the initial state $\mathbf{x}(0) = 10010010$ to the target state 
$\mathbf{x}^{T} = 10100001$. At time $t = 0$, we let $u_{2}(0) = 1$, $u_{3}(0) = 1$, $u_{4}(0) = 0$, $u_{6}(0) = 1$,
$u_{7}(0) = 1$, $u_{8}(0) = 0$. At time $t = 1$, we let $u_{2}(1) = 1$, $u_{3}(1) = 0$, $u_{4}(1) = 0$, $u_{6}(1) = 0$,
$u_{7}(1) = 1$, $u_{8}(1) = 0$. Then, we have the following.
\begin{center}
\begin{tabular}{l|llllllll}
\hline
$\mathbf{x}(0)$		& 1 & 0 & 0 & 1 & 0 & 0 & 1 & 0 \\
$F(\mathbf{x}(0))$ 	& 1 & 0 & 0 & 1 & 0 & 1 & 1 & 0 \\
$\mathbf{x}(1)$ 		& 1 & 1 & 1 & 1 & 0 & 0 & 0 & 0 \\
$F(\mathbf{x}(1))$ 	& 1 & 1 & 1 & 0 & 0 & 0 & 1 & 1 \\
$\mathbf{x}(2)$ 		& 1 & 0 & 1 & 0 & 0 & 0 & 0 & 1 \\
\hline
\end{tabular}
\end{center}
\end{example}

\subsection{Best-Case Upper Bounds for Majority-Type Boolean Networks}\label{section:majority_BNs_bcub}

In this section, we provide three best-case upper bounds for $(2k+1)$-$(2k+1)$-MAJORITY-BNs, $2k$-$2k$-MAJORITY-BNs and $2k$-$2k$-MTBI-BNs.

\subsubsection{$(2k+1)$-$(2k+1)$-MAJORITY-BNs}\label{section:bcub_2k1_2k1_majority_BNs}

Let $k$ be a positive integer. Denote the $(2k+1)$-ary majority function by
$f_{k} : \{ 0, 1 \}^{2k+1} \to \{ 0, 1 \}$. In the following, we construct a $(2k+1)$-$(2k+1)$-MAJORITY-BN with 
$n  =m(2k + 2)$ nodes, where $m$ is a positive integer. We label the nodes as $x^{1}_{1}$, $x^{1}_{2}$, \ldots, $x^{1}_{2k+2}$, $x^{2}_{1}$, $x^{2}_{2}$, \ldots, $x^{2}_{2k+2}$, \ldots, $x^{m}_{1}$, $x^{m}_{2}$, \ldots, $x^{m}_{2k+2}$. The dependencies among the nodes are defined below.
\begin{align}
x^{2}_{1}(t + 1) &=
f_{k}(x^{1}_{2}(t), x^{1}_{3}(t), x^{1}_{4}(t), \ldots, x^{1}_{2k+2}(t)) \\
x^{2}_{2}(t + 1) &=
f_{k}(x^{1}_{1}(t), x^{1}_{3}(t), x^{1}_{4}(t), \ldots, x^{1}_{2k+2}(t)) \\
& \quad \vdots \nonumber \\
x^{2}_{2k+2}(t + 1) &=
f_{k}(x^{1}_{1}(t), x^{1}_{2}(t), x^{1}_{3}(t), \ldots, x^{1}_{2k+1}(t)) \\
x^{3}_{1}(t + 1) &=
f_{k}(x^{2}_{2}(t), x^{2}_{3}(t), x^{2}_{4}(t), \ldots, x^{2}_{2k+2}(t)) \\
x^{3}_{2}(t + 1) &=
f_{k}(x^{2}_{1}(t), x^{2}_{3}(t), x^{2}_{4}(t), \ldots, x^{2}_{2k+2}(t)) \\
& \quad \vdots \nonumber  \\ 
x^{3}_{2k+2}(t + 1) &=
f_{k}(x^{2}_{1}(t), x^{2}_{2}(t), x^{2}_{3}(t), \ldots, x^{2}_{2k+1}(t)) \\
& \quad \vdots \nonumber  \\ 
x^{1}_{1}(t + 1) &=
f_{k}(x^{m}_{2}(t), x^{m}_{3}(t), x^{m}_{4}(t), \ldots, x^{m}_{2k+2}(t)) \\
x^{1}_{2}(t + 1) &=
f_{k}(x^{m}_{1}(t), x^{m}_{3}(t), x^{m}_{4}(t), \ldots, x^{m}_{2k+2}(t)) \\
& \quad \vdots \nonumber \\ 
x^{1}_{2k+2}(t + 1) &=
f_{k}(x^{m}_{1}(t), x^{m}_{2}(t), x^{m}_{3}(t), \ldots, x^{m}_{2k+1}(t))
\end{align}

Let $G_{1} : \{ 0, 1 \}^{2k+2} \to \{ 0, 1 \}^{2k+2}$ be the function such that for all $\mathbf{y} \in \{ 0, 1 \}^{2k+2}$,
\begin{equation}\label{eq:def_of_G1}
G_{1}(\mathbf{y}) = (f_{k}(\mathbf{y}_{-1}), f_{k}(\mathbf{y}_{-2}), \ldots, f_{k}(\mathbf{y}_{-(2k+2)})).
\end{equation}
Then, the dependencies of this BN can be succinctly represented as
\begin{align}
\mathbf{x}^{2}(t + 1) &= 
G_{1}(\mathbf{x}^{1}(t)) \\
\mathbf{x}^{3}(t + 1) &=
G_{1}(\mathbf{x}^{2}(t)) \\
& \vdots \nonumber \\ 
\mathbf{x}^{1}(t + 1) &= 
G_{1}(\mathbf{x}^{m}(t))
\end{align}

The graphical representation of this BN is given in Figure \ref{fig:general_2k1_2k1_majority_BN}.

\begin{figure}[h]
\caption{Graphical Representation of the $(2k+1)$-$(2k+1)$-MAJORITY-BN}\label{fig:general_2k1_2k1_majority_BN}
\centering
\includegraphics[width=0.7\textwidth]{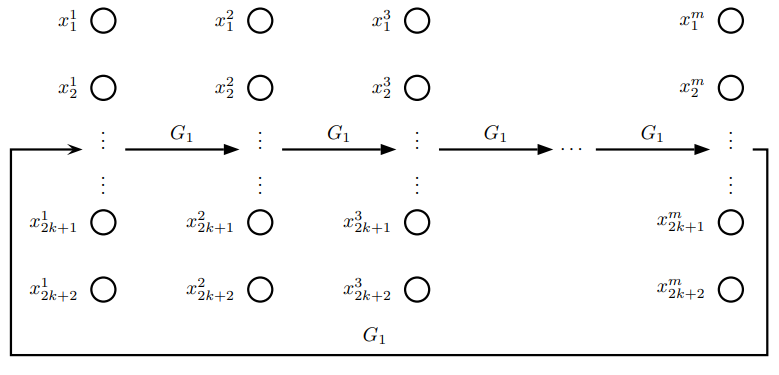}
\end{figure}

\FloatBarrier

The control-node set is defined to be
\begin{equation}\label{eq:2k1_2k1_MAJORITY_BN_set_U}
U \coloneqq 
\{ x^{1}_{1}, x^{1}_{2}, \ldots, x^{1}_{2k+2} \} \cup 
\{ x^{2}_{1}, \ldots, x^{2}_{k} \} \cup
\{ x^{3}_{1}, \ldots, x^{3}_{k} \} \cup \ldots \cup
\{ x^{m}_{1}, \ldots, x^{m}_{k} \}.
\end{equation}
Note that $|U| = (2k + 2) + (m - 1)k$ and $\frac{|U|}{n} = \frac{(2k + 2) + (m - 1)k}{m(2k + 2)}$.

The following proposition is the key to understanding why this BN is controllable.

\begin{proposition}\label{prop:applying_G1}
For all $\mathbf{y} \in \{ 0, 1 \}^{2k+2}$,
if the number of $1$-entries of $\mathbf{y}$ is at least $k + 2$, then $G_{1}(\mathbf{y}) = \mathbf{1}_{2k + 2}$;
if the number of $1$-entries of $\mathbf{y}$ is at most $k$, then $G_{1}(\mathbf{y}) = \mathbf{0}_{2k + 2}$;
if the number of $1$-entries of $\mathbf{y}$ is equal to $k + 1$, then 
$G_{1}(\mathbf{y}) = (1 - y_{1}, 1 - y_{2}, \ldots, 1 - y_{2k+2})$.
\end{proposition}

\begin{proof}
Fix arbitrary $\mathbf{y} \in \{ 0, 1 \}^{2k+2}$. Let $\mathbf{z} \coloneqq G_{1}(\mathbf{y})$.

Suppose that the first case is true. Then, for any $i \in \{ 1, 2, \ldots, 2k + 2 \}$, at least $k + 1$ entries of $\mathbf{y}_{-i}$ equal $1$. Therefore, $z_{i} = f_{k}(\mathbf{y}_{-i}) = 1$. Hence, $\mathbf{z} = \mathbf{1}_{2k+2}$.

Suppose that the second case is true. Then, for any $i \in \{ 1, 2, \ldots, 2k + 2 \}$, at most $k$ entries of $\mathbf{y}_{-i}$ equal $1$. Therefore, $z_{i} = f_{k}(\mathbf{y}_{-i}) = 0$. Hence, $\mathbf{z} = \mathbf{0}_{2k+2}$.

Suppose that the third case is true. Fix arbitrary $i \in \{ 1, 2, \ldots, 2k + 2 \}$. If $y_{i} = 1$, then exactly $k$ entries of $\mathbf{y}_{-i}$ equal $1$. Hence, $z_{i} = f_{k}(\mathbf{y}_{-i}) = 0 = 1 - y_{i}$, 
On the other hand, if $y_{i} = 0$, then exactly $k+1$ entries of $\mathbf{y}_{-i}$ equal $1$. Hence, 
$z_{i} = f_{k}(\mathbf{y}_{-i}) = 1 = 1 - y_{i}$. Therefore, 
$\mathbf{z} = (1 - y_{1}, 1 - y_{2}, \ldots, 1 - y_{2k+2})$.

The proof is complete.
\end{proof}

By considering Proposition \ref{prop:applying_G1}, we can intuitively understand the evolution of the BN. Suppose that we do not exert any control signals, and that at time $t$, the states of the nodes $x^{1}_{1}$, $x^{1}_{2}$, \ldots, $x^{1}_{2k+1}$, $x^{1}_{2k+2}$ are given by some $\mathbf{y}$ in $\{ 0, 1 \}^{2k+2}$. Consider the following three cases:
\begin{enumerate}
\item if the number of 1-entries in $\mathbf{y}$ is at least $k+2$, then for $s = 1, 2, \ldots, m-1$, the states of the nodes $x^{s+1}_{1}$, $x^{s+1}_{2}$, \ldots, $x^{s+1}_{2k+1}$, $x^{s+1}_{2k+2}$ at time $t + s$ will be $\mathbf{1}_{2k+2}$;

\item if the number of 1-entries in $\mathbf{y}$ is at most $k$, then for $s = 1, 2, \ldots, m-1$, the states of the nodes $x^{s+1}_{1}$, $x^{s+1}_{2}$, \ldots, $x^{s+1}_{2k+1}$, $x^{s+1}_{2k+2}$ at time $t + s$ will be $\mathbf{0}_{2k+2}$;

\item if the number of 1-entries in $\mathbf{y}$ equals $k+1$, then for $s = 1, 3, 5, \ldots$, the states of the nodes $x^{s+1}_{1}$, $x^{s+1}_{2}$, \ldots, $x^{s+1}_{2k+1}$, $x^{s+1}_{2k+2}$ at time $t + s$ will be
$\mathbf{1}_{2k+2} - \mathbf{y}$, and for $s = 2, 4 \ldots$, the states of the nodes $x^{s+1}_{1}$, $x^{s+1}_{2}$, \ldots, $x^{s+1}_{2k+1}$, $x^{s+1}_{2k+2}$ at time $t + s$ will be $\mathbf{y}$.
\end{enumerate}

The idea behind our method of controlling the BN is that for $l = 2, 3, \ldots, m$, if we want the nodes $x^{l}_{k+1}$, $x^{l}_{k+2}$, \ldots, $x^{l}_{2k+2}$ in the $l$-th layer to be in some particular states at time $m$, then we appropriately modify the states of the nodes $x^{1}_{1}$, $x^{1}_{2}$, \ldots, $x^{1}_{2k+1}$, $x^{1}_{2k+2}$ at time $m - l$ and let such information pass along the layers of nodes and reach the $l$-th layer.

Before presenting our control strategy of this $(2k+1)$-$(2k+1)$-MAJORITY-BN, we need to provide the following definition.
\begin{definition}\label{def:alpha_one_v}
Let $\mathbf{v}$ be a vector in $\{ 0, 1 \}^{k+2}$ which contains at least one $0$-entry and at least one $1$-entry. Let $d$ be the number of $1$-entries in $\mathbf{v}$. Clearly, $1 \leq d \leq k+1$. We define the length-$k$ vector
\begin{equation}\label{eq:alpha_one_v}
\alpha_{1}(\mathbf{v}) \coloneqq 
\underbrace{11\ldots1}_{\substack{\text{$k + 1 - d$} \\ \text{times}}}\overbrace{00\ldots0}^{\substack{\text{$d - 1$} \\ \text{times}}}
\end{equation}
\end{definition}

\begin{theorem}\label{thm:(2k+1)_bcub}
The $(2k+2)m$-node $(2k+1)$-$(2k+1)$-MAJORITY-BN given in Figure \ref{fig:general_2k1_2k1_majority_BN} with control-node set $U$ of size $2k + 2 + (m - 1)k$ defined in Eq.\@ (\ref{eq:2k1_2k1_MAJORITY_BN_set_U}) is controllable. More specifically, for all 
$\mathbf{a}, \mathbf{b} \in \{ 0, 1 \}^{(2k+2)m}$, there exists a control scheme which drives the BN from $\mathbf{a}$ (time $0$) to $\mathbf{b}$ (time $m$).
\end{theorem}

\begin{proof}
The control strategy is shown in Algorithm \ref{alg:control_scheme_(2k+1)_MAJORITY}.

\begin{algorithm}
\caption{The Control Strategy for a $(2k+1)$-$(2k+1)$-MAJORITY-BN}
\label{alg:control_scheme_(2k+1)_MAJORITY}
\begin{algorithmic}[1]
\Require An initial state $\mathbf{a} \in \{ 0, 1 \}^{m(2k+2)}$ and a target state $\mathbf{b} \in \{ 0, 1 \}^{m(2k+2)}$.

\State Set $\mathbf{x}(0) \gets \mathbf{a}$.

\LComment{Iteration phase}

\For{$t = 0, 1, \ldots, m - 2$}
\State Use the update rules of the BN to update the BN from state $\mathbf{x}(t)$. Denote the resulting state by $\mathbf{y}(t) = F(\mathbf{x}(t))$.

\If{$b^{m-t}_{k+1} = b^{m-t}_{k+2} = \cdots = b^{m-t}_{2k+2} = 1$}

\State Overwrite the states of the nodes $x^{1}_{1}$, $x^{1}_{2}$, \ldots, $x^{1}_{2k+2}$ in $\mathbf{y}(t)$ to form $\mathbf{x}(t+1)$ which satisfies $x^{1}_{1}(t + 1) = x^{1}_{2}(t + 1) = \cdots = x^{1}_{2k+2}(t + 1) = 1$.

\ElsIf{$b^{m-t}_{k+1} = b^{m-t}_{k+2} = \cdots = b^{m-t}_{2k+2} = 0$}

\State Overwrite the states of the nodes $x^{1}_{1}$, $x^{1}_{2}$, \ldots, $x^{1}_{2k+2}$ in $\mathbf{y}(t)$ to form $\mathbf{x}(t+1)$ which satisfies $x^{1}_{1}(t + 1) = x^{1}_{2}(t + 1) = \cdots = x^{1}_{2k+2}(t + 1) = 0$.

\ElsIf{$m - t$ is odd and $\mathbf{p} \coloneqq (b^{m-t}_{k+1}, b^{m-t}_{k+2}, \ldots, b^{m-t}_{2k+2})$ contains both 0-entry and 1-entry}

\State Overwrite the states of the nodes $x^{1}_{1}$, $x^{1}_{2}$, \ldots, $x^{1}_{2k+2}$ in $\mathbf{y}(t)$ to form 
$\mathbf{x}(t+1)$ which satisfies 
$(x^{1}_{1}(t + 1), x^{1}_{2}(t + 1), \ldots, x^{1}_{2k+2}(t + 1)) = (\alpha_{1}(\mathbf{p}), \mathbf{p})$.
Note that
$(\alpha_{1}(\mathbf{p}), \mathbf{p})$
contains exactly $k+1$ 1-entries.

\ElsIf{$m - t$ is even and $\mathbf{p} \coloneqq (b^{m-t}_{k+1}, b^{m-t}_{k+2}, \ldots, b^{m-t}_{2k+2})$ contains both 0-entry and 1-entry}

\State Overwrite the states of the nodes $x^{1}_{1}$, $x^{1}_{2}$, \ldots, $x^{1}_{2k+2}$ in $\mathbf{y}(t)$ to form 
$\mathbf{x}(t+1)$ which satisfies 
$(x^{1}_{1}(t + 1), x^{1}_{2}(t + 1), \ldots, x^{1}_{2k+2}(t + 1)) = \mathbf{1}_{2k+2} - (\alpha_{1}(\mathbf{p}), \mathbf{p})$.
Note that
$\mathbf{1}_{2k+2} - (\alpha_{1}(\mathbf{p}), \mathbf{p})$
contains exactly $k+1$ 1-entries.
\EndIf
\EndFor

\LComment{Termination phase}
\State Use the update rules of the BN to update the BN from state $\mathbf{x}(m-1)$.
Denote the resulting state by $\mathbf{y}(m-1) = F(\mathbf{x}(m-1))$.

\State Overwrite the states of all control nodes in $\mathbf{y}(m-1)$ to form $\mathbf{x}(m)$ which satisfies 
$\mathbf{x}(m) = \mathbf{b}$.
\end{algorithmic}
\end{algorithm}

\FloatBarrier
\end{proof}

\begin{example}
Let $k = 2$ and $m = 4$.
The corresponding $5$-$5$-MAJORITY-BN has $(2k+2)m = 24$ nodes and is shown in Figure \ref{fig:example_5_5_majority_BN}.
\begin{figure}[h]
\caption{Structure of the $24$-node $5$-$5$-MAJORITY-BN. The blue nodes are the control nodes.}\label{fig:example_5_5_majority_BN}
\centering
\includegraphics[width=0.6\textwidth]{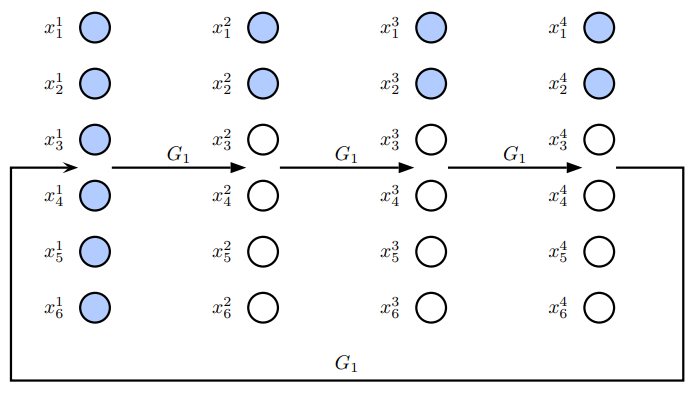}
\end{figure}

\FloatBarrier

Suppose that we want to drive the $5$-$5$-MAJORITY-BN from the initial state $\mathbf{a} \in \{ 0, 1 \}^{24}$ to the target state $\mathbf{b} \in \{ 0, 1 \}^{24}$ defined below.
\begin{table*}[!h]
\centering
\begin{tabular}{ccc}  
  \begin{tabular}{l|llll}
  \multirow{6}{*}{$\mathbf{a}$}	& 1 & 1 & 1 & 0 \\
						& 0 & 1 & 1 & 1 \\
						& 0 & 0 & 1 & 1 \\
						& 1 & 1 & 0 & 0 \\
						& 0 & 0 & 0 & 0 \\
						& 0 & 1 & 0 & 1 \\
  \end{tabular}
  & \phantom{hello} &
  \begin{tabular}{l|llll}
  \multirow{6}{*}{$\mathbf{b}$}	& 0 & 1 & 1 & 1 \\
						& 1 & 0 & 0 & 0 \\
						& 0 & 1 & 1 & 0 \\
						& 1 & 1 & 0 & 1 \\
						& 0 & 1 & 0 & 1 \\
						& 0 & 1 & 1 & 1 \\
  \end{tabular}
  \\
\end{tabular}  
\end{table*}

\FloatBarrier

We can apply our control strategy in Algorithm \ref{alg:control_scheme_(2k+1)_MAJORITY} to drive the BN from $\mathbf{a}$ (time $0$) to $\mathbf{b}$ (time $m = 4$), which is illustrated below.
\begin{table*}[!h]
\centering
\begin{tabular}{rrr}  
  \begin{tabular}{l|llll}
  \multirow{6}{*}{$\mathbf{x}(0)$}	& 1 & 1 & 1 & 0 \\
						& 0 & 1 & 1 & 1 \\
						& 0 & 0 & 1 & 1 \\
						& 1 & 1 & 0 & 0 \\
						& 0 & 0 & 0 & 0 \\
						& 0 & 1 & 0 & 1 \\
  \end{tabular}
  &
  \begin{tabular}{l|llll}
  \multirow{6}{*}{$F(\mathbf{x}(0))$}	& 1 & 0 & 1 & 0 \\
						& 0 & 0 & 1 & 0 \\
						& 0 & 0 & 1 & 0 \\
						& 1 & 0 & 1 & 1 \\
						& 1 & 0 & 1 & 1 \\
						& 0 & 0 & 1 & 1 \\
  \end{tabular}
  &
  \begin{tabular}{l|llll}
  \multirow{6}{*}{$\mathbf{x}(1)$}	& 1 & 0 & 1 & 0 \\
						& 1 & 0 & 1 & 0 \\
						& 1 & 0 & 1 & 0 \\
						& 0 & 0 & 1 & 1 \\
						& 0 & 0 & 1 & 1 \\
						& 0 & 0 & 1 & 1 \\
  \end{tabular}
  \\
  \phantom{a}  & \phantom{a}  & \phantom{a} \\
  \begin{tabular}{l|llll}
  \multirow{6}{*}{$F(\mathbf{x}(1))$}	& 1 & 0 & 0 & 1 \\
						& 1 & 0 & 0 & 1 \\
						& 1 & 0 & 0 & 1 \\
						& 0 & 1 & 0 & 1 \\
						& 0 & 1 & 0 & 1 \\
						& 0 & 1 & 0 & 1 \\
  \end{tabular}
  &
  \begin{tabular}{l|llll}
  \multirow{6}{*}{$\mathbf{x}(2)$}	& 1 & 0 & 0 & 1 \\
						& 0 & 0 & 0 & 1 \\
						& 1 & 0 & 0 & 1 \\
						& 0 & 1 & 0 & 1 \\
						& 0 & 1 & 0 & 1 \\
						& 1 & 1 & 0 & 1 \\
  \end{tabular}
  &
  \begin{tabular}{l|llll}
  \multirow{6}{*}{$F(\mathbf{x}(2))$}	& 1 & 0 & 1 & 0 \\
						& 1 & 1 & 1 & 0 \\
						& 1 & 0 & 1 & 0 \\
						& 1 & 1 & 0 & 0 \\
						& 1 & 1 & 0 & 0 \\
						& 1 & 0 & 0 & 0 \\
  \end{tabular}
  \\
  \phantom{a}  & \phantom{a}  & \phantom{a} \\
  \begin{tabular}{l|llll}
  \multirow{6}{*}{$\mathbf{x}(3)$}	& 1 & 0 & 1 & 0 \\
						& 1 & 1 & 1 & 0 \\
						& 1 & 0 & 1 & 0 \\
						& 1 & 1 & 0 & 0 \\
						& 1 & 1 & 0 & 0 \\
						& 1 & 0 & 0 & 0 \\
  \end{tabular}
  &
  \begin{tabular}{l|llll}
  \multirow{6}{*}{$F(\mathbf{x}(3))$}	& 0 & 1 & 1 & 0 \\
						& 0 & 1 & 0 & 0 \\
						& 0 & 1 & 1 & 0 \\
						& 0 & 1 & 0 & 1 \\
						& 0 & 1 & 0 & 1 \\
						& 0 & 1 & 1 & 1 \\
  \end{tabular}
  &
  \begin{tabular}{l|llll}
  \multirow{6}{*}{$\mathbf{x}(4)$}	& 0 & 1 & 1 & 1 \\
						& 1 & 0 & 0 & 0 \\
						& 0 & 1 & 1 & 0 \\
						& 1 & 1 & 0 & 1 \\
						& 0 & 1 & 0 & 1 \\
						& 0 & 1 & 1 & 1 \\
  \end{tabular}
  \\
\end{tabular}  
\end{table*}

\FloatBarrier
\end{example}

\subsubsection{$2k$-$2k$-MAJORITY-BNs}\label{section:bcub_2k_2k_majority_BNs}

Let $k$ be a positive integer. Denote the $2k$-ary majority function by
$h_{k} : \{ 0, 1 \}^{2k} \to \{ 0, 1 \}$. In the following, we construct a $2k$-$2k$-MAJORITY-BN with 
$n  =m(2k + 1)$ nodes, where $m$ is a positive integer. We remark that although this BN is similar to the $(2k+1)$-$(2k+1)$-MAJORITY-BN constructed in Section \ref{section:bcub_2k1_2k1_majority_BNs}, the control strategies for the two BNs are fairly different. We label the $m(2k + 1)$ nodes as $x^{1}_{1}$, $x^{1}_{2}$, \ldots, $x^{1}_{2k+1}$, $x^{2}_{1}$, $x^{2}_{2}$, \ldots, $x^{2}_{2k+1}$, \ldots, $x^{m}_{1}$, $x^{m}_{2}$, \ldots, $x^{m}_{2k+1}$. The dependencies among the nodes are defined below.
\begin{align}
x^{2}_{1}(t + 1) &=
h_{k}(x^{1}_{2}(t), x^{1}_{3}(t), x^{1}_{4}(t), \ldots, x^{1}_{2k+1}(t)) \\
x^{2}_{2}(t + 1) &=
h_{k}(x^{1}_{1}(t), x^{1}_{3}(t), x^{1}_{4}(t), \ldots, x^{1}_{2k+1}(t)) \\
& \quad \vdots \nonumber \\
x^{2}_{2k+1}(t + 1) &=
h_{k}(x^{1}_{1}(t), x^{1}_{2}(t), x^{1}_{3}(t), \ldots, x^{1}_{2k}(t)) \\
x^{3}_{1}(t + 1) &=
h_{k}(x^{2}_{2}(t), x^{2}_{3}(t), x^{2}_{4}(t), \ldots, x^{2}_{2k+1}(t)) \\
x^{3}_{2}(t + 1) &=
h_{k}(x^{2}_{1}(t), x^{2}_{3}(t), x^{2}_{4}(t), \ldots, x^{2}_{2k+1}(t)) \\
& \quad \vdots \nonumber  \\ 
x^{3}_{2k+1}(t + 1) &=
h_{k}(x^{2}_{1}(t), x^{2}_{2}(t), x^{2}_{3}(t), \ldots, x^{2}_{2k}(t)) \\
& \quad \vdots \nonumber  \\ 
x^{1}_{1}(t + 1) &=
h_{k}(x^{m}_{2}(t), x^{m}_{3}(t), x^{m}_{4}(t), \ldots, x^{m}_{2k+1}(t)) \\
x^{1}_{2}(t + 1) &=
h_{k}(x^{m}_{1}(t), x^{m}_{3}(t), x^{m}_{4}(t), \ldots, x^{m}_{2k+1}(t)) \\
& \quad \vdots \nonumber \\ 
x^{1}_{2k+1}(t + 1) &=
h_{k}(x^{m}_{1}(t), x^{m}_{2}(t), x^{m}_{3}(t), \ldots, x^{m}_{2k}(t))
\end{align}

Let $G_{2} : \{ 0, 1 \}^{2k+1} \to \{ 0, 1 \}^{2k+1}$ be the function such that for all $\mathbf{y} \in \{ 0, 1 \}^{2k+1}$,
\begin{equation}\label{eq:def_of_G2}
G_{2}(\mathbf{y}) = (h_{k}(\mathbf{y}_{-1}), h_{k}(\mathbf{y}_{-2}), \ldots, h_{k}(\mathbf{y}_{-(2k+1)})).
\end{equation}
Then, the dependencies of this BN can be succinctly represented as
\begin{align}
\mathbf{x}^{2}(t + 1) &= 
G_{2}(\mathbf{x}^{1}(t)) \\
\mathbf{x}^{3}(t + 1) &=
G_{2}(\mathbf{x}^{2}(t)) \\
& \vdots \nonumber \\ 
\mathbf{x}^{1}(t + 1) &= 
G_{2}(\mathbf{x}^{m}(t))
\end{align}

The graphical representation of this BN is given in Figure \ref{fig:general_2k_2k_majority_BN}.

\begin{figure}[h]
\caption{Graphical Representation of the $2k$-$2k$-MAJORITY-BN}\label{fig:general_2k_2k_majority_BN}
\centering
\includegraphics[width=0.7\textwidth]{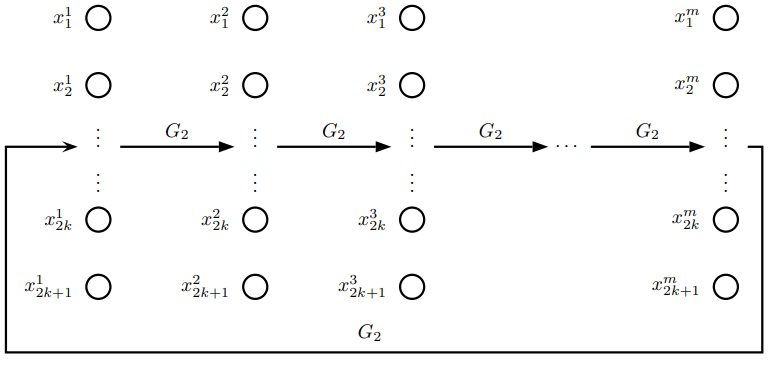}
\end{figure}

\FloatBarrier

The control-node set is defined to be
\begin{equation}\label{eq:2k_2k_MAJORITY_BN_set_U}
U \coloneqq 
\{ x^{1}_{1}, x^{1}_{2}, \ldots, x^{1}_{2k+1} \} \cup 
\{ x^{2}_{1}, \ldots, x^{2}_{k} \} \cup
\{ x^{3}_{1}, \ldots, x^{3}_{k} \} \cup \ldots \cup
\{ x^{m}_{1}, \ldots, x^{m}_{k} \}.
\end{equation}
Note that $|U| = (2k + 1) + (m - 1)k$ and $\frac{|U|}{n} = \frac{(2k + 1) + (m - 1)k}{m(2k + 1)}$.

The following proposition and definition are the key to understanding why this BN is controllable.

\begin{proposition}\label{prop:applying_G2}
For all $\mathbf{y} \in \{ 0, 1 \}^{2k+1}$,
if the number of $1$-entries of $\mathbf{y}$ is at least $k + 1$, then $G_{2}(\mathbf{y}) = \mathbf{1}_{2k + 1}$;
if the number of $1$-entries of $\mathbf{y}$ is at most $k - 1$, then $G_{2}(\mathbf{y}) = \mathbf{0}_{2k + 1}$;
if the number of $1$-entries of $\mathbf{y}$ is equal to $k$, then 
$G_{2}(\mathbf{y}) = (1 - y_{1}, 1 - y_{2}, \ldots, 1 - y_{2k+1})$.
\end{proposition}

\begin{proof}
The proof of this proposition follows the same line of reasoning as the proof of Proposition \ref{prop:applying_G1}.
\end{proof}

\begin{definition}\label{def:alpha_two_v}
Let $\mathbf{v}$ be a vector in $\{ 0, 1 \}^{k+1}$ which contains at least one $0$-entry and at least one $1$-entry. Let $d$ be the number of $1$-entries in $\mathbf{v}$. Clearly, $1 \leq d \leq k$. We define the length-$k$ vector
\begin{equation}\label{eq:alpha_two_v}
\alpha_{2}(\mathbf{v}) \coloneqq 
\underbrace{11\ldots1}_{\substack{\text{$k - d$} \\ \text{times}}}\overbrace{00\ldots0}^{\substack{\text{$d$} \\ \text{times}}}
\end{equation}
\end{definition}

By considering Proposition \ref{prop:applying_G2}, we can intuitively understand the evolution of the $2k$-$2k$-MAJORITY-BN. Suppose that we do not exert any control signals, and that at time $t$, the states of the nodes $x^{1}_{1}$, $x^{1}_{2}$, \ldots, $x^{1}_{2k}$, $x^{1}_{2k+1}$ are given by some $\mathbf{y}$ in $\{ 0, 1 \}^{2k+1}$. Consider the following three cases:
\begin{enumerate}
\item if the number of 1-entries in $\mathbf{y}$ is at least $k+1$, then for $s = 1, 2, \ldots, m-1$, the states of the nodes $x^{s+1}_{1}$, $x^{s+1}_{2}$, \ldots, $x^{s+1}_{2k}$, $x^{s+1}_{2k+1}$ at time $t + s$ will be $\mathbf{1}_{2k+1}$;

\item if the number of 1-entries in $\mathbf{y}$ is at most $k - 1$, then for $s = 1, 2, \ldots, m-1$, the states of the nodes $x^{s+1}_{1}$, $x^{s+1}_{2}$, \ldots, $x^{s+1}_{2k}$, $x^{s+1}_{2k+1}$ at time $t + s$ will be $\mathbf{0}_{2k+1}$;

\item if the number of 1-entries in $\mathbf{y}$ equals $k$, then the states of the nodes $x^{2}_{1}$, $x^{2}_{2}$, \ldots, $x^{2}_{2k}$, $x^{2}_{2k+1}$ at time $t + 1$ will be $\mathbf{1}_{2k+1} - \mathbf{y}$. 
Note that $\mathbf{1}_{2k+1} - \mathbf{y}$ has $k+1$ 1-entries. Therefore, for $s = 2, 3, \ldots, m - 1$, 
the states of the nodes $x^{s+1}_{1}$, $x^{s+1}_{2}$, \ldots, $x^{s+1}_{2k}$, $x^{s+1}_{2k+1}$ at time $t + s$ will be $\mathbf{1}_{2k+1}$.
\end{enumerate}

Now, we explain a key idea behind how we control the $2k$-$2k$-MAJORITY-BN. 

Fix arbitrary $l \in \{ 2, 3, \ldots, m \}$ and $b_{k+1}, b_{k+2}, \ldots, b_{2k+1} \in \{ 0, 1 \}$ such that $l$ is odd and 
$b_{k+1}$, $b_{k+2}$, \ldots, $b_{2k+1}$ contain both $0$ and $1$. Let $d \in [1, k]$ be the number of $1$-entries in $\mathbf{b} \coloneqq (b_{k+1}, b_{k+2}, \ldots, b_{2k+1})$. Suppose that we want to exert control signals to
achieve $x^{l}_{k+1}(m) = b_{k+1}$, $x^{l}_{k+2}(m) = b_{k+2}$, \ldots, $x^{l}_{2k+1}(m) = b_{2k+1}$.
To do so, first we exert control signals on the control nodes $x^{1}_{1}$, $x^{1}_{2}$, \ldots, $x^{1}_{2k+1}$ at time $m - l$ so that $(x^{1}_{1}(m - l + 1), x^{1}_{2}(m - l + 1), \ldots, x^{1}_{2k+1}(m - l + 1)) = 
(\alpha_{2}(\mathbf{b}), \mathbf{b})$, which has $k$ $1$-entries. Then, the BN evolves so that the nodes
$x^{2}_{1}$, $x^{2}_{2}$, \ldots, $x^{2}_{2k+1}$ are in the state
\begin{equation}
G_{2}(\alpha_{2}(\mathbf{b}), \mathbf{b})
= \mathbf{1}_{2k+1} - (\alpha_{2}(\mathbf{b}), \mathbf{b})
= (\mathbf{1}_{k} - \alpha_{2}(\mathbf{b}), \mathbf{1}_{k+1} - \mathbf{b}).
\end{equation}
Note that $(\mathbf{1}_{k} - \alpha_{2}(\mathbf{b}), \mathbf{1}_{k+1} - \mathbf{b})$ has $k+1$ $1$-entries.
We want to reduce the number of $1$-entries to $k$. To do so, we overwrite the states of the control nodes 
$x^{2}_{1}$, $x^{2}_{2}$, \ldots, $x^{2}_{k}$ at time $m - l + 1$ so that
$(x^{2}_{1}(m - l + 2), x^{2}_{2}(m - l + 2), \ldots, x^{2}_{2k+1}(m - l + 2))
= (\alpha_{2}(\mathbf{1}_{k+1} - \mathbf{b}), \mathbf{1}_{k+1} - \mathbf{b})$, which clearly has $k$ $1$-entries.
Afterwards, the BN evolves so that the nodes $x^{3}_{1}$, $x^{3}_{2}$, \ldots, $x^{3}_{2k+1}$ are in the state
\begin{equation}
G_{2}(\alpha_{2}(\mathbf{1}_{k+1} - \mathbf{b}), \mathbf{1}_{k+1} - \mathbf{b})
= \mathbf{1}_{2k+1} - (\alpha_{2}(\mathbf{1}_{k+1} - \mathbf{b}), \mathbf{1}_{k+1} - \mathbf{b})
= (\mathbf{1}_{k} - \alpha_{2}(\mathbf{1}_{k+1} - \mathbf{b}), \mathbf{b}),
\end{equation}
which has $k + 1$ $1$-entries. To reduce the number of $1$-entries to $k$, we can overwrite the states of the control nodes $x^{3}_{1}$, $x^{3}_{2}$, \ldots, $x^{3}_{k}$ at time $m - l + 2$ so that
$(x^{3}_{1}(m - l + 3), x^{3}_{2}(m - l + 3), \ldots, x^{3}_{2k+1}(m - l + 3))
= (\alpha_{2}(\mathbf{b}), \mathbf{b})$, which has $k$ $1$-entries.
We repeat this process iteratively until the nodes $x^{l}_{k+1}$, $x^{l}_{k+2}$, \ldots, $x^{l}_{2k+1}$ are in states 
$b_{k+1}$, $b_{k+2}$, \ldots, $b_{2k+1}$ respectively at time $m$.

As for the case that $l \in \{ 2, 3, \ldots, m \}$ is even, we exert control signals in a similar fashion, 
but at time $m - l$, we exert control signals on the nodes $x^{1}_{1}$, $x^{1}_{2}$, \ldots, $x^{1}_{2k+1}$ so that
$(x^{1}_{1}(m - l + 1), x^{1}_{2}(m - l + 1), \ldots, x^{1}_{2k+1}(m - l + 1))
= (\alpha_{2}(\mathbf{1}_{k+1} - \mathbf{b}), \mathbf{1}_{k+1} - \mathbf{b})$.

We can now present the theorem on the control of this $2k$-$2k$-MAJORITY-BN and the associated control strategy.

\begin{theorem}\label{thm:2k_MAJOR_bcub}
The constructed $(2k+1)m$-node $2k$-$2k$-MAJORITY-BN with control-node set $U$ 
of size $2k + 1 + (m - 1)k$ defined in Eq.\@ (\ref{eq:2k_2k_MAJORITY_BN_set_U}) is controllable. More specifically, for all 
$\mathbf{a}, \mathbf{b} \in \{ 0, 1 \}^{(2k+1)m}$, there exists a control scheme which drives the BN from 
$\mathbf{a}$ (time $0$) to $\mathbf{b}$ (time $m$).
\end{theorem}

\begin{proof}
The control strategy is shown in Algorithm \ref{alg:control_scheme_2k_MAJORITY}.

\begin{algorithm}
\caption{The Control Strategy for a $2k$-$2k$-MAJORITY-BN}
\label{alg:control_scheme_2k_MAJORITY}
\begin{algorithmic}[1]
\Require An initial state $\mathbf{a} \in \{ 0, 1 \}^{m(2k+1)}$ and a target state $\mathbf{b} \in \{ 0, 1 \}^{m(2k+1)}$.

\State Set $\mathbf{x}(0) \gets \mathbf{a}$.

\LComment{Iteration phase}

\For{$t = 0, 1, \ldots, m - 2$}
\State Use the update rules of the BN to update the BN from state $\mathbf{x}(t)$. Denote the resulting state by $\mathbf{y}(t) = F(\mathbf{x}(t))$.

\If{$b^{m-t}_{k+1} = b^{m-t}_{k+2} = \cdots = b^{m-t}_{2k+1} = 1$}

\State Overwrite the states of the nodes $x^{1}_{1}$, $x^{1}_{2}$, \ldots, $x^{1}_{2k+1}$ in $\mathbf{y}(t)$ so that $x^{1}_{1}(t + 1) = x^{1}_{2}(t + 1) = \cdots = x^{1}_{2k+1}(t + 1) = 1$.

\ElsIf{$b^{m-t}_{k+1} = b^{m-t}_{k+2} = \cdots = b^{m-t}_{2k+1} = 0$}

\State Overwrite the states of the nodes $x^{1}_{1}$, $x^{1}_{2}$, \ldots, $x^{1}_{2k+1}$ in $\mathbf{y}(t)$ so that $x^{1}_{1}(t + 1) = x^{1}_{2}(t + 1) = \cdots = x^{1}_{2k+1}(t + 1) = 0$.

\ElsIf{$m - t$ is odd and $\mathbf{p} \coloneqq (b^{m-t}_{k+1}, b^{m-t}_{k+2}, \ldots, b^{m-t}_{2k+1})$ contains both 0-entry and 1-entry}

\State Overwrite the states of the nodes $x^{1}_{1}$, $x^{1}_{2}$, \ldots, $x^{1}_{2k+1}$ in $\mathbf{y}(t)$ so that
$(x^{1}_{1}(t + 1), x^{1}_{2}(t + 1), \ldots, x^{1}_{2k+1}(t + 1)) = (\alpha_{2}(\mathbf{p}), \mathbf{p})$.
Note that
$(\alpha_{2}(\mathbf{p}), \mathbf{p})$
contains exactly $k$ 1-entries.

\ElsIf{$m - t$ is even and $\mathbf{p} \coloneqq (b^{m-t}_{k+1}, b^{m-t}_{k+2}, \ldots, b^{m-t}_{2k+1})$ contains both 0-entry and 1-entry}

\State Overwrite the states of the nodes $x^{1}_{1}$, $x^{1}_{2}$, \ldots, $x^{1}_{2k+1}$ in $\mathbf{y}(t)$ so that
$(x^{1}_{1}(t + 1), x^{1}_{2}(t + 1), \ldots, x^{1}_{2k+1}(t + 1)) 
= (\alpha_{2}(\mathbf{1}_{k+1} - \mathbf{p}), \mathbf{1}_{k+1} - \mathbf{p})$.
Note that
$(\alpha_{2}(\mathbf{1}_{k+1} - \mathbf{p}), \mathbf{1}_{k+1} - \mathbf{p})$
contains exactly $k$ 1-entries.
\EndIf

\For{$l = 2, 3, \ldots, m$}
\If{$\mathbf{q} \coloneqq (y^{l}_{k+1}(t), y^{l}_{k+2}(t), \ldots, y^{l}_{2k+1}(t))$ contains both 0-entry and 1-entry}

\State Overwrite the states of the nodes $x^{l}_{1}$, $x^{l}_{2}$, \ldots, $x^{l}_{k}$ in $\mathbf{y}(t)$ so that 
$(x^{l}_{1}(t + 1), x^{l}_{2}(t + 1), \ldots, x^{l}_{k}(t + 1)) = \alpha_{2}(\mathbf{q})$.
\EndIf
\EndFor 

\LComment{We use $\mathbf{x}(t + 1)$ to denote the state of the whole BN after the overwriting in lines 4--15.}
\EndFor 

\LComment{Termination phase}
\State Use the update rules of the BN to update the BN from state $\mathbf{x}(m-1)$.
Denote the resulting state by $\mathbf{y}(m-1) = F(\mathbf{x}(m-1))$.

\State Overwrite the states of all control nodes in $\mathbf{y}(m-1)$ to form $\mathbf{x}(m)$ which satisfies 
$\mathbf{x}(m) = \mathbf{b}$.
\end{algorithmic}
\end{algorithm}

\FloatBarrier
\end{proof}

\begin{remark}\label{rmk:diff_between_odd_control_even_control}
The major difference between the control strategies in Algorithms \ref{alg:control_scheme_(2k+1)_MAJORITY} and \ref{alg:control_scheme_2k_MAJORITY} lies in lines 13--15 of Algorithm \ref{alg:control_scheme_2k_MAJORITY}.
\end{remark}

\begin{example}
Let $k = 2$ and $m = 4$. The corresponding $4$-$4$-MAJORITY-BN has $(2k + 1)m = 20$ nodes and is shown in Figure \ref{fig:example_4_4_majority_BN}.

\begin{figure}[h]
\caption{Structure of the $20$-node $4$-$4$-MAJORITY-BN. The blue nodes are the control nodes.}\label{fig:example_4_4_majority_BN}
\centering
\includegraphics[width=0.6\textwidth]{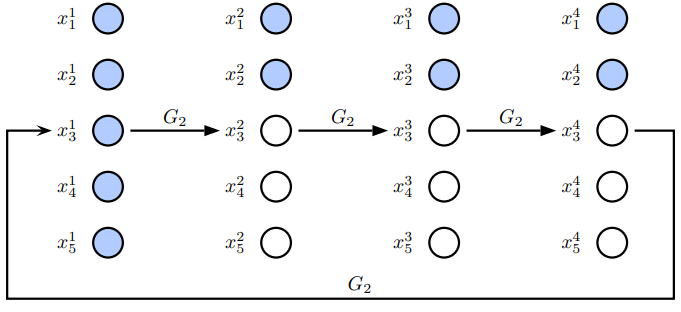}
\end{figure}

\FloatBarrier

Suppose that we want to drive the $4$-$4$-MAJORITY-BN from the initial state $\mathbf{a} \in \{ 0, 1 \}^{20}$ to the target state $\mathbf{b} \in \{ 0, 1 \}^{20}$ defined below.
\begin{table*}[!h]
\centering
\begin{tabular}{ccc}  
  \begin{tabular}{l|llll}
  \multirow{5}{*}{$\mathbf{a}$}	& 0 & 1 & 1 & 0 \\
						& 0 & 1 & 0 & 0 \\
						& 1 & 0 & 0 & 0 \\
						& 0 & 0 & 1 & 0 \\
						& 0 & 0 & 0 & 1 \\
  \end{tabular}
  & \phantom{hello} &
  \begin{tabular}{l|llll}
  \multirow{5}{*}{$\mathbf{b}$}	& 1 & 1 & 1 & 0 \\
						& 1 & 0 & 0 & 0 \\
						& 0 & 1 & 1 & 0 \\
						& 0 & 1 & 0 & 1 \\
						& 1 & 1 & 1 & 0 \\
  \end{tabular}
  \\
\end{tabular}  
\end{table*}

\FloatBarrier

We can apply our control strategy in Algorithm \ref{alg:control_scheme_2k_MAJORITY} to drive the BN from 
$\mathbf{a}$ (time $0$) to $\mathbf{b}$ (time $m = 4$), which is illustrated below.
\begin{table*}[!h]
\centering
\begin{tabular}{rrr}  
  \begin{tabular}{l|llll}
  \multirow{5}{*}{$\mathbf{x}(0)$}	& 0 & 1 & 1 & 0 \\
						& 0 & 1 & 0 & 0 \\
						& 1 & 0 & 0 & 0 \\
						& 0 & 0 & 1 & 0 \\
						& 0 & 0 & 0 & 1 \\
  \end{tabular}
  &
  \begin{tabular}{l|llll}
  \multirow{5}{*}{$F(\mathbf{x}(0))$}	& 0 & 0 & 0 & 0 \\
						& 0 & 0 & 0 & 1 \\
						& 0 & 0 & 1 & 1 \\
						& 0 & 0 & 1 & 0 \\
						& 0 & 0 & 1 & 1 \\
  \end{tabular}
  &
  \begin{tabular}{l|llll}
  \multirow{5}{*}{$\mathbf{x}(1)$}	& 0 & 0 & 0 & 0 \\
						& 0 & 0 & 0 & 0 \\
						& 1 & 0 & 1 & 1 \\
						& 0 & 0 & 1 & 0 \\
						& 1 & 0 & 1 & 1 \\
  \end{tabular}
  \\
  \phantom{a}  & \phantom{a}  & \phantom{a} \\
  \begin{tabular}{l|llll}
  \multirow{5}{*}{$F(\mathbf{x}(1))$}	& 1 & 1 & 0 & 1 \\
						& 1 & 1 & 0 & 1 \\
						& 0 & 0 & 0 & 1 \\
						& 1 & 1 & 0 & 1 \\
						& 0 & 0 & 0 & 1 \\
  \end{tabular}
  &
  \begin{tabular}{l|llll}
  \multirow{5}{*}{$\mathbf{x}(2)$}	& 0 & 1 & 0 & 1 \\
						& 0 & 0 & 0 & 1 \\
						& 1 & 0 & 0 & 1 \\
						& 0 & 1 & 0 & 1 \\
						& 1 & 0 & 0 & 1 \\
  \end{tabular}
  &
  \begin{tabular}{l|llll}
  \multirow{5}{*}{$F(\mathbf{x}(2))$}	& 1 & 1 & 0 & 0 \\
						& 1 & 1 & 1 & 0 \\
						& 1 & 0 & 1 & 0 \\
						& 1 & 1 & 0 & 0 \\
						& 1 & 0 & 1 & 0 \\
  \end{tabular}
  \\
  \phantom{a}  & \phantom{a}  & \phantom{a} \\
  \begin{tabular}{l|llll}
  \multirow{5}{*}{$\mathbf{x}(3)$}	& 1 & 1 & 0 & 0 \\
						& 1 & 0 & 0 & 0 \\
						& 1 & 0 & 1 & 0 \\
						& 1 & 1 & 0 & 0 \\
						& 1 & 0 & 1 & 0 \\
  \end{tabular}
  &
  \begin{tabular}{l|llll}
  \multirow{5}{*}{$F(\mathbf{x}(3))$}	& 0 & 1 & 0 & 1 \\
						& 0 & 1 & 1 & 1 \\
						& 0 & 1 & 1 & 0 \\
						& 0 & 1 & 0 & 1 \\
						& 0 & 1 & 1 & 0 \\
  \end{tabular}
  &
  \begin{tabular}{l|llll}
  \multirow{5}{*}{$\mathbf{x}(4)$}	& 1 & 1 & 1 & 0 \\
						& 1 & 0 & 0 & 0 \\
						& 0 & 1 & 1 & 0 \\
						& 0 & 1 & 0 & 1 \\
						& 1 & 1 & 1 & 0 \\
  \end{tabular}
  \\
\end{tabular}  
\end{table*}

\FloatBarrier
\end{example}

\subsubsection{$2k$-$2k$-MTBI-BNs}\label{section:bcub_2k_2k_MTBI_BNs}

Let $k$ be a positive integer. Denote the $2k$-ary MTBI function by
$g_{k} : \{ 0, 1 \}^{2k} \to \{ 0, 1 \}$. In the following, we construct a $2k$-$2k$-MTBI-BN with 
$n  =2mk$ nodes, where $m$ is a positive integer. We label the nodes as $x^{1}_{1}$, $x^{1}_{2}$, \ldots, $x^{1}_{2k}$, $x^{2}_{1}$, $x^{2}_{2}$, \ldots, $x^{2}_{2k}$, \ldots, $x^{m}_{1}$, $x^{m}_{2}$, \ldots, $x^{m}_{2k}$. The dependencies among the nodes are defined below.
\begin{align}
x^{2}_{1}(t + 1) &=
g_{k}(x^{1}_{1}(t), x^{1}_{2}(t), x^{1}_{3}(t), \ldots, x^{1}_{2k}(t)) \\
x^{2}_{2}(t + 1) &=
g_{k}(x^{1}_{2}(t), x^{1}_{1}(t), x^{1}_{3}(t), \ldots, x^{1}_{2k}(t)) \\
& \quad \vdots \nonumber \\
x^{2}_{2k}(t + 1) &=
g_{k}(x^{1}_{2k}(t), x^{1}_{1}(t), x^{1}_{2}(t), \ldots, x^{1}_{2k-1}(t)) \\
x^{3}_{1}(t + 1) &=
g_{k}(x^{2}_{1}(t), x^{2}_{2}(t), x^{2}_{3}(t), \ldots, x^{2}_{2k}(t)) \\
x^{3}_{2}(t + 1) &=
g_{k}(x^{2}_{2}(t), x^{2}_{1}(t), x^{2}_{3}(t), \ldots, x^{2}_{2k}(t)) \\
& \quad \vdots \nonumber  \\ 
x^{3}_{2k}(t + 1) &=
g_{k}(x^{2}_{2k}(t), x^{2}_{1}(t), x^{2}_{2}(t), \ldots, x^{2}_{2k-1}(t)) \\
& \quad \vdots \nonumber  \\ 
x^{1}_{1}(t + 1) &=
g_{k}(x^{m}_{1}(t), x^{m}_{2}(t), x^{m}_{3}(t), \ldots, x^{m}_{2k}(t)) \\
x^{1}_{2}(t + 1) &=
g_{k}(x^{m}_{2}(t), x^{m}_{1}(t), x^{m}_{3}(t), \ldots, x^{m}_{2k}(t)) \\
& \quad \vdots \nonumber \\ 
x^{1}_{2k}(t + 1) &=
g_{k}(x^{m}_{2k}(t), x^{m}_{1}(t), x^{m}_{2}(t), \ldots, x^{m}_{2k-1}(t))
\end{align}

Denote the function $\mathbf{x}^{1}(t) \mapsto \mathbf{x}^{2}(t + 1)$ by $G_{3} : \{ 0, 1 \}^{2k} \to \{ 0, 1 \}^{2k}$. 
Then, the dependencies of this BN can be succinctly represented as
\begin{align}
\mathbf{x}^{2}(t + 1) &= 
G_{3}(\mathbf{x}^{1}(t)) \\
\mathbf{x}^{3}(t + 1) &=
G_{3}(\mathbf{x}^{2}(t)) \\
& \vdots \nonumber \\ 
\mathbf{x}^{1}(t + 1) &= 
G_{3}(\mathbf{x}^{m}(t))
\end{align}

The graphical representation of this BN is given in Figure \ref{fig:general_2k_2k_MTBI_BN}.

\begin{figure}[h]
\caption{Graphical Representation of the $2k$-$2k$-MTBI-BN}\label{fig:general_2k_2k_MTBI_BN}
\centering
\includegraphics[width=0.7\textwidth]{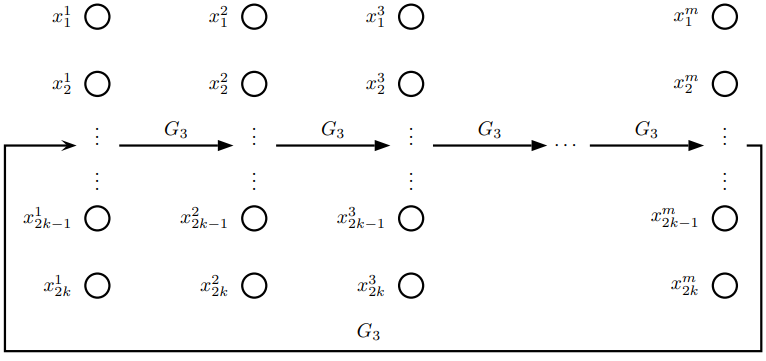}
\end{figure}

\FloatBarrier

The control-node set is defined to be
\begin{equation}\label{eq:2k_2k_MTBI_BN_set_U}
U \coloneqq 
\{ x^{1}_{1}, x^{1}_{2}, \ldots, x^{1}_{2k} \} \cup 
\{ x^{2}_{1}, \ldots, x^{2}_{k-1} \} \cup
\{ x^{3}_{1}, \ldots, x^{3}_{k-1} \} \cup \ldots \cup
\{ x^{m}_{1}, \ldots, x^{m}_{k-1} \}.
\end{equation}
Note that $|U| = 2k + (m - 1)(k - 1)$ and $\frac{|U|}{n} = \frac{2k + (m - 1)(k - 1)}{2mk}$.

The following proposition is the key to understanding why this BN is controllable.

\begin{proposition}\label{prop:applying_G3}
For all $\mathbf{y} \in \{ 0, 1 \}^{2k}$,
if the number of $1$-entries of $\mathbf{y}$ is at least $k + 1$, then $G_{3}(\mathbf{y}) = \mathbf{1}_{2k}$;
if the number of $1$-entries of $\mathbf{y}$ is at most $k - 1$, then $G_{3}(\mathbf{y}) = \mathbf{0}_{2k}$;
if the number of $1$-entries of $\mathbf{y}$ is equal to $k$, then 
$G_{3}(\mathbf{y}) = \mathbf{y}$.
\end{proposition}

\begin{proof}
This proposition follows directly from the definition of $2k$-ary MTBI function (Definition \ref{def:2k-ary_MTBI_function}).
\end{proof}

By considering Proposition \ref{prop:applying_G3}, we can intuitively understand the evolution of the BN. Suppose that we do not exert any control signals, and that at time $t$, the states of the nodes $x^{1}_{1}$, $x^{1}_{2}$, \ldots, $x^{1}_{2k-1}$, $x^{1}_{2k}$ are given by some $\mathbf{y}$ in $\{ 0, 1 \}^{2k}$. Consider the following three cases:
\begin{enumerate}
\item if the number of 1-entries in $\mathbf{y}$ is at least $k+1$, then for $s = 1, 2, \ldots, m-1$, the states of the nodes $x^{s+1}_{1}$, $x^{s+1}_{2}$, \ldots, $x^{s+1}_{2k-1}$, $x^{s+1}_{2k}$ at time $t + s$ will be $\mathbf{1}_{2k}$;

\item if the number of 1-entries in $\mathbf{y}$ is at most $k-1$, then for $s = 1, 2, \ldots, m-1$, the states of the nodes $x^{s+1}_{1}$, $x^{s+1}_{2}$, \ldots, $x^{s+1}_{2k-1}$, $x^{s+1}_{2k}$ at time $t + s$ will be $\mathbf{0}_{2k}$;

\item if the number of 1-entries in $\mathbf{y}$ equals $k$, then for $s = 1, 2, \ldots, m-1$, the states of the nodes $x^{s+1}_{1}$, $x^{s+1}_{2}$, \ldots, $x^{s+1}_{2k-1}$, $x^{s+1}_{2k}$ at time $t + s$ will be $\mathbf{y}$.
\end{enumerate}

Before presenting our control strategy of this $2k$-$2k$-MTBI-BN, we need to provide the following definition.
\begin{definition}\label{def:alpha_three_v}
Let $\mathbf{v}$ be a vector in $\{ 0, 1 \}^{k+1}$ which contains at least one $0$-entry and at least one $1$-entry. Let $d$ be the number of $1$-entries in $\mathbf{v}$. Clearly, $1 \leq d \leq k$. We define the length-$(k-1)$ vector
\begin{equation}\label{eq:alpha_three_v}
\alpha_{3}(\mathbf{v}) \coloneqq 
\underbrace{11\ldots1}_{\substack{\text{$k - d$} \\ \text{times}}}\overbrace{00\ldots0}^{\substack{\text{$d - 1$} \\ \text{times}}}
\end{equation}
\end{definition}

\begin{theorem}\label{thm:2k_MTBI_bcub}
The constructed $2mk$-node $2k$-$2k$-MTBI-BN with control-node set $U$ 
of size $2k + (m - 1)(k - 1)$ defined in Eq.\@ (\ref{eq:2k_2k_MTBI_BN_set_U}) is controllable. More specifically, for all 
$\mathbf{a}, \mathbf{b} \in \{ 0, 1 \}^{2mk}$, there exists a control scheme which drives the BN from $\mathbf{a}$ (time $0$) to $\mathbf{b}$ (time $m$).
\end{theorem}

\begin{proof}
The control strategy is shown in Algorithm \ref{alg:control_scheme_2k_MTBI}.

\begin{algorithm}
\caption{The Control Strategy for a $2k$-$2k$-MTBI-BN}
\label{alg:control_scheme_2k_MTBI}
\begin{algorithmic}[1]
\Require An initial state $\mathbf{a} \in \{ 0, 1 \}^{2mk}$ and a target state $\mathbf{b} \in \{ 0, 1 \}^{2mk}$.

\State Set $\mathbf{x}(0) \gets \mathbf{a}$.

\LComment{Iteration phase}

\For{$t = 0, 1, \ldots, m - 2$}
\State Use the update rules of the BN to update the BN from state $\mathbf{x}(t)$. Denote the resulting state by $\mathbf{y}(t) = F(\mathbf{x}(t))$.

\If{$b^{m-t}_{k} = b^{m-t}_{k+1} = \cdots = b^{m-t}_{2k} = 1$}

\State Overwrite the states of the nodes $x^{1}_{1}$, $x^{1}_{2}$, \ldots, $x^{1}_{2k}$ in $\mathbf{y}(t)$ to form $\mathbf{x}(t+1)$ which satisfies $x^{1}_{1}(t + 1) = x^{1}_{2}(t + 1) = \cdots = x^{1}_{2k}(t + 1) = 1$.

\ElsIf{$b^{m-t}_{k} = b^{m-t}_{k+1} = \cdots = b^{m-t}_{2k} = 0$}

\State Overwrite the states of the nodes $x^{1}_{1}$, $x^{1}_{2}$, \ldots, $x^{1}_{2k}$ in $\mathbf{y}(t)$ to form $\mathbf{x}(t+1)$ which satisfies $x^{1}_{1}(t + 1) = x^{1}_{2}(t + 1) = \cdots = x^{1}_{2k}(t + 1) = 0$.

\ElsIf{$\mathbf{p} \coloneqq (b^{m-t}_{k}, b^{m-t}_{k+1}, \ldots, b^{m-t}_{2k})$ contains both 0-entry and 1-entry}

\State Overwrite the states of the nodes $x^{1}_{1}$, $x^{1}_{2}$, \ldots, $x^{1}_{2k}$ in $\mathbf{y}(t)$ to form 
$\mathbf{x}(t+1)$ which satisfies 
$(x^{1}_{1}(t + 1), x^{1}_{2}(t + 1), \ldots, x^{1}_{2k}(t + 1)) = (\alpha_{3}(\mathbf{p}), \mathbf{p})$.
Note that
$(\alpha_{3}(\mathbf{p}), \mathbf{p})$
contains exactly $k$ 1-entries.
\EndIf
\EndFor

\LComment{Termination phase}
\State Use the update rules of the BN to update the BN from state $\mathbf{x}(m-1)$.
Denote the resulting state by $\mathbf{y}(m-1) = F(\mathbf{x}(m-1))$.

\State Overwrite the states of all control nodes in $\mathbf{y}(m-1)$ to form $\mathbf{x}(m)$ which satisfies 
$\mathbf{x}(m) = \mathbf{b}$.
\end{algorithmic}
\end{algorithm}

\FloatBarrier
\end{proof}

\begin{example}
Let $k = 3$ and $m = 4$.
The corresponding $6$-$6$-MTBI-BN has $2mk = 24$ nodes and is shown in Figure \ref{fig:example_6_6_MTBI_BN}.
\begin{figure}[h]
\caption{Structure of the $24$-node $6$-$6$-MTBI-BN. The blue nodes are the control nodes.}\label{fig:example_6_6_MTBI_BN}
\centering
\includegraphics[width=0.6\textwidth]{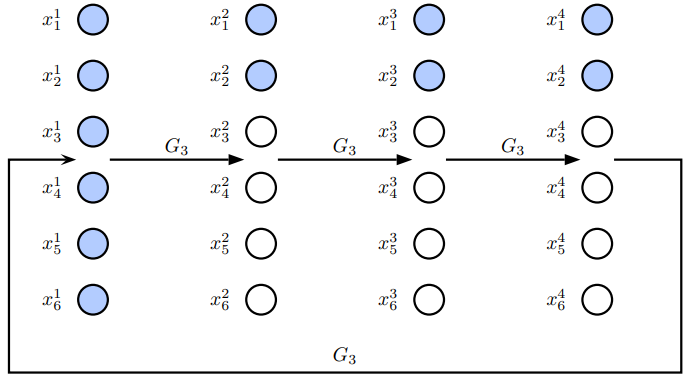}
\end{figure}

\FloatBarrier

Suppose that we want to drive the $6$-$6$-MTBI-BN from the initial state $\mathbf{a} \in \{ 0, 1 \}^{24}$ to the target state $\mathbf{b} \in \{ 0, 1 \}^{24}$ defined below.
\begin{table*}[!h]
\centering
\begin{tabular}{ccc}  
  \begin{tabular}{l|llll}
  \multirow{6}{*}{$\mathbf{a}$}	& 1 & 1 & 1 & 1 \\
						& 0 & 0 & 0 & 0 \\
						& 1 & 0 & 1 & 0 \\
						& 1 & 0 & 0 & 1 \\
						& 0 & 0 & 1 & 1 \\
						& 1 & 1 & 0 & 0 \\
  \end{tabular}
  & \phantom{hello} &
  \begin{tabular}{l|llll}
  \multirow{6}{*}{$\mathbf{b}$}	& 0 & 0 & 0 & 1 \\
						& 1 & 1 & 1 & 0 \\
						& 1 & 0 & 0 & 1 \\
						& 1 & 0 & 1 & 0 \\
						& 0 & 0 & 0 & 1 \\
						& 0 & 0 & 1 & 1 \\
  \end{tabular}
  \\
\end{tabular}  
\end{table*}

\FloatBarrier

We can apply our control strategy in Algorithm \ref{alg:control_scheme_2k_MTBI} to drive the BN from $\mathbf{a}$ (time $0$) to $\mathbf{b}$ (time $m = 4$), which is illustrated below.
\begin{table*}[!h]
\centering
\begin{tabular}{rrr}  
  \begin{tabular}{l|llll}
  \multirow{6}{*}{$\mathbf{x}(0)$}	& 1 & 1 & 1 & 1 \\
						& 0 & 0 & 0 & 0 \\
						& 1 & 0 & 1 & 0 \\
						& 1 & 0 & 0 & 1 \\
						& 0 & 0 & 1 & 1 \\
						& 1 & 1 & 0 & 0 \\
  \end{tabular}
  &
  \begin{tabular}{l|llll}
  \multirow{6}{*}{$F(\mathbf{x}(0))$}	& 1 & 1 & 0 & 1 \\
						& 0 & 1 & 0 & 0 \\
						& 0 & 1 & 0 & 1 \\
						& 1 & 1 & 0 & 0 \\
						& 1 & 1 & 0 & 1 \\
						& 0 & 1 & 0 & 0 \\
  \end{tabular}
  &
  \begin{tabular}{l|llll}
  \multirow{6}{*}{$\mathbf{x}(1)$}	& 0 & 1 & 0 & 1 \\
						& 0 & 1 & 0 & 0 \\
						& 1 & 1 & 0 & 1 \\
						& 0 & 1 & 0 & 0 \\
						& 1 & 1 & 0 & 1 \\
						& 1 & 1 & 0 & 0 \\
  \end{tabular}
  \\
  \phantom{a}  & \phantom{a}  & \phantom{a} \\
  \begin{tabular}{l|llll}
  \multirow{6}{*}{$F(\mathbf{x}(1))$}	& 1 & 0 & 1 & 0 \\
						& 0 & 0 & 1 & 0 \\
						& 1 & 1 & 1 & 0 \\
						& 0 & 0 & 1 & 0 \\
						& 1 & 1 & 1 & 0 \\
						& 0 & 1 & 1 & 0 \\
  \end{tabular}
  &
  \begin{tabular}{l|llll}
  \multirow{6}{*}{$\mathbf{x}(2)$}	& 1 & 0 & 1 & 0 \\
						& 0 & 0 & 1 & 0 \\
						& 0 & 1 & 1 & 0 \\
						& 1 & 0 & 1 & 0 \\
						& 0 & 1 & 1 & 0 \\
						& 1 & 1 & 1 & 0 \\
  \end{tabular}
  &
  \begin{tabular}{l|llll}
  \multirow{6}{*}{$F(\mathbf{x}(2))$}	& 0 & 1 & 0 & 1 \\
						& 0 & 0 & 0 & 1 \\
						& 0 & 0 & 1 & 1 \\
						& 0 & 1 & 0 & 1 \\
						& 0 & 0 & 1 & 1 \\
						& 0 & 1 & 1 & 1 \\
  \end{tabular}
  \\
  \phantom{a}  & \phantom{a}  & \phantom{a} \\
  \begin{tabular}{l|llll}
  \multirow{6}{*}{$\mathbf{x}(3)$}	& 0 & 1 & 0 & 1 \\
						& 0 & 0 & 0 & 1 \\
						& 0 & 0 & 1 & 1 \\
						& 0 & 1 & 0 & 1 \\
						& 0 & 0 & 1 & 1 \\
						& 0 & 1 & 1 & 1 \\
  \end{tabular}
  &
  \begin{tabular}{l|llll}
  \multirow{6}{*}{$F(\mathbf{x}(3))$}	& 1 & 0 & 1 & 0 \\
						& 1 & 0 & 0 & 0 \\
						& 1 & 0 & 0 & 1 \\
						& 1 & 0 & 1 & 0 \\
						& 1 & 0 & 0 & 1 \\
						& 1 & 0 & 1 & 1 \\
  \end{tabular}
  &
  \begin{tabular}{l|llll}
  \multirow{6}{*}{$\mathbf{x}(4)$}	& 0 & 0 & 0 & 1 \\
						& 1 & 1 & 1 & 0 \\
						& 1 & 0 & 0 & 1 \\
						& 1 & 0 & 1 & 0 \\
						& 0 & 0 & 0 & 1 \\
						& 0 & 0 & 1 & 1 \\
  \end{tabular}
  \\
\end{tabular}  
\end{table*}

\FloatBarrier
\end{example}

\subsection{A Best-Case Upper Bound Related to Boolean Threshold Functions with Both Positive and Negative Coefficients}\label{section:threshold_BNs_pos_neg_coef}

Let $k \geq 3$ be a positive integer, and $k_{1}$, $k_{2}$ be positive integers such that $k_{1} + k_{2} = k$. Consider any Boolean function $\varphi_{k} : \{ 0, 1 \}^{k} \to \{ 0, 1 \}$ such that
\begin{enumerate}
\item all inputs of $\varphi_{k}$ are relevant/essential, which means that for $i = 1, 2, \ldots, k$, there exist $a_{1}, \ldots, a_{i-1}, a_{i+1}, \ldots, a_{k} \in \{ 0, 1 \}$ satisfying 
$\varphi_{k}(a_{1}, \ldots, a_{i-1}, 0, a_{i+1}, \ldots, a_{k}) \neq 
\varphi_{k}(a_{1}, \ldots, a_{i-1}, 1, a_{i+1}, \ldots, a_{k})$; and

\item $\varphi_{k}$ is a Boolean threshold function of the form
\begin{equation}\label{eq:phi_k_threshold}
\varphi_{k}(\mathbf{x}) \coloneqq
\begin{cases}
1	& \text{if $a_{1}x_{1} + \cdots + a_{k}x_{k} \geq b$} \\
0	& \text{otherwise}
\end{cases}
\end{equation}
where exactly $k_{1}$ of $a_{1}$, $a_{2}$, \ldots, $a_{k}$ are positive and exactly $k_{2}$ of $a_{1}$, $a_{2}$, \ldots, $a_{k}$ are negative.
\end{enumerate}

If an $n$-node BN satisfies that each of its Boolean functions is $\varphi_{k}$ and each node has $k$ input nodes and $k$ output nodes, then the BN is said to be an $n$-node $k$-$k$-$\varphi_{k}$-BN.

In this section, we show that for all $k \geq 3$, there exists a Boolean function 
$\varphi_{k} : \{ 0, 1 \}^{k} \to \{ 0, 1 \}$ such that $\varphi_{k}$ satisfies points (1) and (2) above and for all positive integers $m$, there exists a controllable $mk$-node $k$-$k$-$\varphi_{k}$-BN with $k$ control nodes.
We remark that majority functions and MTBI functions given in Definitions \ref{def:k-ary_majority_function} and \ref{def:2k-ary_MTBI_function} satisfy point (1) but not point (2).

Let $k \geq 3$ be a positive integer. Consider the Boolean function $\varphi_{k} : \{ 0, 1 \}^{k} \to \{ 0, 1 \}$ such that for all $\mathbf{x} \in \{ 0, 1 \}^{k}$,
\begin{equation}\label{eq:special_phi_k}
\varphi_{k}(\mathbf{x}) \coloneqq
\begin{cases}
1	& \text{if $(k - 2)x_{1} - x_{2} - \cdots - x_{k} \geq 0$} \\
0	& \text{otherwise}
\end{cases}
\end{equation}
We call $x_{1}$ \textit{the key input of $\varphi_{k}$}.

The following proposition directly follows from the definition of $\varphi_{k}$ (Eq.\@ (\ref{eq:special_phi_k})).

\begin{proposition}\label{prop:phi_k_characterization}
$\varphi_{k}(\mathbf{1}_{k}) = 0$, $\varphi_{k}(\mathbf{0}_{k}) = 1$, and
for all $\mathbf{x} \in \{ 0, 1 \}^{k} \setminus \{ \mathbf{1}_{k}, \mathbf{0}_{k} \}$, 
$\varphi_{k}(\mathbf{x}) = x_{1}$.
\end{proposition}

Now, we show that every input of $\varphi_{k}$ is relevant.

\begin{proposition}\label{prop:every_input_is_relevant}
Every input of $\varphi_{k} : \{ 0, 1 \}^{k} \to \{ 0, 1 \}$ is relevant.
\end{proposition}

\begin{proof}
Consider the input $x_{1}$ of $\varphi_{k}$. Let $\mathbf{v} \in \{ 0, 1 \}^{k-1}$ be any vector such that $\mathbf{v}$ contains both $0$-entry and $1$-entry. By Proposition \ref{prop:phi_k_characterization}, we have $\varphi_{k}(0, \mathbf{v}) = 0$ and $\varphi_{k}(1, \mathbf{v}) = 1$. Therefore, the input $x_{1}$ is relevant.

Next, we show that for $i = 2, 3, \ldots, k$, the input $x_{i}$ of $\varphi_{k}$ is relevant. For all $j \in \{ 1, 2, \ldots, k\} \setminus \{ i \}$, let $a_{j} \coloneqq 0$. By Proposition \ref{prop:phi_k_characterization}, we have
$\varphi_{k}(a_{1}, \ldots, a_{i-1}, 0, a_{i+1}, \ldots, a_{k}) = \varphi(\mathbf{0}_{k}) = 1$ and
$\varphi_{k}(a_{1}, \ldots, a_{i-1}, 1, a_{i+1}, \ldots, a_{k}) = a_{1} = 0$. Therefore, the input $x_{i}$ is relevant.

The proof is complete.
\end{proof}

Now, we define the function $G_{4} : \{ 0, 1 \}^{k} \to \{ 0, 1 \}^{k}$ such that for all $\mathbf{x} \in \{ 0, 1 \}^{k}$,
\begin{equation}\label{eq:G4_def}
G_{4}(\mathbf{x}) = 
(\varphi_{k}(x_{1}, x_{2}, \ldots, x_{k}), 
\varphi_{k}(x_{2}, x_{3}, \ldots, x_{k}, x_{1}),
\ldots,
\varphi_{k}(x_{k}, x_{1}, \ldots, x_{k-2}, x_{k-1})).
\end{equation}
The following proposition follows directly from Proposition \ref{prop:phi_k_characterization}.

\begin{proposition}\label{prop:G4_characterization}
$G_{4}(\mathbf{1}_{k}) = \mathbf{0}_{k}$, $G_{4}(\mathbf{0}_{k}) = \mathbf{1}_{k}$,
and for all $\mathbf{x} \in \{ 0, 1 \}^{k} \setminus \{ \mathbf{1}_{k}, \mathbf{0}_{k} \}$, $G_{4}(\mathbf{x}) = \mathbf{x}$.
\end{proposition}

We are now ready to present our controllable $k$-$k$-$\varphi_{k}$-BN with a small control-node set. 
Let $m$ be a positive integer. Our BN consists of $n = mk$ nodes, which we label as 
$x^{1}_{1}$, $x^{1}_{2}$, \ldots, $x^{1}_{k}$, $x^{2}_{1}$, $x^{2}_{2}$, \ldots, $x^{2}_{k}$, \ldots,
$x^{m}_{1}$, $x^{m}_{2}$, \ldots, $x^{m}_{k}$. The dependencies among the nodes are defined below.
\begin{align}
(x^{2}_{1}(t + 1), x^{2}_{2}(t + 1), \ldots, x^{2}_{k}(t + 1)) 	&= G_{4}(x^{1}_{1}(t), x^{1}_{2}(t), \ldots, x^{1}_{k}(t)) \\
(x^{3}_{1}(t + 1), x^{3}_{2}(t + 1), \ldots, x^{3}_{k}(t + 1)) 	&= G_{4}(x^{2}_{1}(t), x^{2}_{2}(t), \ldots, x^{2}_{k}(t)) \\
														& \vdots \nonumber \\
(x^{1}_{1}(t + 1), x^{1}_{2}(t + 1), \ldots, x^{1}_{k}(t + 1)) &= G_{4}(x^{m}_{1}(t), x^{m}_{2}(t), \ldots, x^{m}_{k}(t))
\end{align}

The graphical representation of our BN is given in Figure \ref{fig:general_k_k_phi_k_BN}.

\begin{figure}[h]
\caption{Graphical Representation of the $k$-$k$-$\varphi_{k}$-BN. The blue nodes are the control nodes.}\label{fig:general_k_k_phi_k_BN}
\centering
\includegraphics[width=0.7\textwidth]{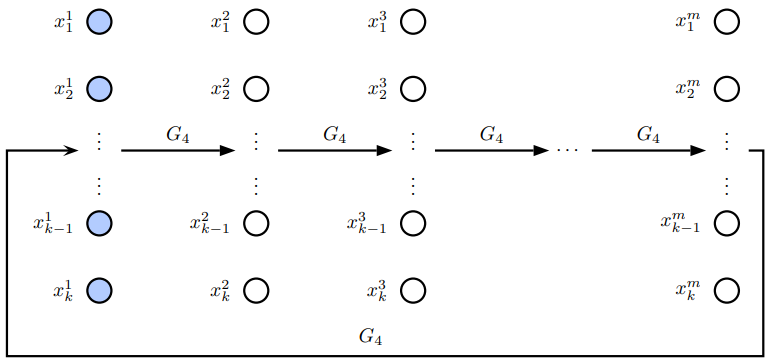}
\end{figure}

\FloatBarrier

The control-node set $U$ is $\{ x^{1}_{1}, x^{1}_{2}, \ldots, x^{1}_{k} \}$. 
Hence, $|U| = k$ and $\frac{|U|}{n} = \frac{1}{m}$.

\begin{theorem}\label{thm:phi_k_BN_bcub}
The constructed $mk$-node $k$-$k$-$\varphi_{k}$-BN with control-node set 
$U \coloneqq \{ x^{1}_{1}, x^{1}_{2}, \ldots, x^{1}_{k} \}$ is controllable.
More specifically, for all 
$\mathbf{a}, \mathbf{b} \in \{ 0, 1 \}^{mk}$, there exists a control scheme which drives the BN from $\mathbf{a}$ (time $0$) to $\mathbf{b}$ (time $m$).
\end{theorem}

\begin{proof}
The control strategy is shown in Algorithm \ref{alg:control_scheme_phi_k}.

\begin{algorithm}
\caption{The Control Strategy for a  $k$-$k$-$\varphi_{k}$-BN}
\label{alg:control_scheme_phi_k}
\begin{algorithmic}[1]
\Require An initial state $\mathbf{a} \in \{ 0, 1 \}^{mk}$ and a target state $\mathbf{b} \in \{ 0, 1 \}^{mk}$.

\State Set $\mathbf{x}(0) \gets \mathbf{a}$.

\For{$t = 0, 1, \ldots, m - 1$}
\State Use the update rules of the BN to update the BN from state $\mathbf{x}(t)$. Denote the resulting state by $\mathbf{y}(t) = F(\mathbf{x}(t))$.

\If{$m - t$ is odd and $\mathbf{b}^{m-t} = \mathbf{1}_{k}$}

\State Overwrite the states of the nodes $x^{1}_{1}$, $x^{1}_{2}$, \ldots, $x^{1}_{k}$ in $\mathbf{y}(t)$ to form 
$\mathbf{x}(t+1)$ which satisfies $\mathbf{x}^{1}(t+1) = \mathbf{1}_{k}$.

\ElsIf{$m - t$ is odd and $\mathbf{b}^{m-t} = \mathbf{0}_{k}$}

\State Overwrite the states of the nodes $x^{1}_{1}$, $x^{1}_{2}$, \ldots, $x^{1}_{k}$ in $\mathbf{y}(t)$ to form 
$\mathbf{x}(t+1)$ which satisfies $\mathbf{x}^{1}(t+1) = \mathbf{0}_{k}$.

\ElsIf{$m - t$ is even and $\mathbf{b}^{m-t} = \mathbf{1}_{k}$}

\State Overwrite the states of the nodes $x^{1}_{1}$, $x^{1}_{2}$, \ldots, $x^{1}_{k}$ in $\mathbf{y}(t)$ to form 
$\mathbf{x}(t+1)$ which satisfies $\mathbf{x}^{1}(t+1) = \mathbf{0}_{k}$.

\ElsIf{$m - t$ is even and $\mathbf{b}^{m-t} = \mathbf{0}_{k}$}

\State Overwrite the states of the nodes $x^{1}_{1}$, $x^{1}_{2}$, \ldots, $x^{1}_{k}$ in $\mathbf{y}(t)$ to form 
$\mathbf{x}(t+1)$ which satisfies $\mathbf{x}^{1}(t+1) = \mathbf{1}_{k}$.

\Else

\State Overwrite the states of the nodes $x^{1}_{1}$, $x^{1}_{2}$, \ldots, $x^{1}_{k}$ in $\mathbf{y}(t)$ to form 
$\mathbf{x}(t+1)$ which satisfies $\mathbf{x}^{1}(t+1) = \mathbf{b}^{m-t}$.

\EndIf
\EndFor
\end{algorithmic}
\end{algorithm}

\FloatBarrier
\end{proof}

\begin{example}
Let $k = 5$ and $m = 4$. The corresponding $5$-$5$-$\varphi_{5}$-BN has $mk = 20$ nodes and is shown in Figure \ref{fig:example_5_5_phi_BN}.

\begin{figure}[h]
\caption{Structure of the $20$-node $5$-$5$-$\varphi_{5}$-BN. The blue nodes are the control nodes.}\label{fig:example_5_5_phi_BN}
\centering
\includegraphics[width=0.6\textwidth]{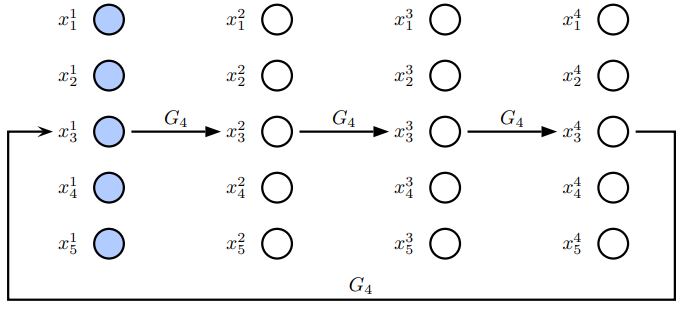}
\end{figure}

\FloatBarrier

Suppose that we want to drive the $5$-$5$-$\varphi_{5}$-BN from the initial state $\mathbf{a} \in \{ 0, 1 \}^{20}$ to the target state $\mathbf{b} \in \{ 0, 1 \}^{20}$ defined below.
\begin{table*}[!h]
\centering
\begin{tabular}{ccc}  
  \begin{tabular}{l|llll}
  \multirow{5}{*}{$\mathbf{a}$}	& 0 & 1 & 1 & 0 \\
						& 0 & 1 & 0 & 0 \\
						& 1 & 0 & 0 & 0 \\
						& 0 & 0 & 1 & 0 \\
						& 0 & 0 & 0 & 1 \\
  \end{tabular}
  & \phantom{hello} &
  \begin{tabular}{l|llll}
  \multirow{5}{*}{$\mathbf{b}$}	& 1 & 1 & 0 & 0 \\
						& 1 & 1 & 0 & 0 \\
						& 0 & 1 & 0 & 0 \\
						& 0 & 1 & 0 & 1 \\
						& 1 & 1 & 0 & 1 \\
  \end{tabular}
  \\
\end{tabular}  
\end{table*}

\FloatBarrier

We can apply our control strategy in Algorithm \ref{alg:control_scheme_phi_k} to drive the BN from 
$\mathbf{a}$ (time $0$) to $\mathbf{b}$ (time $m = 4$), which is illustrated below.
\begin{table*}[!h]
\centering
\begin{tabular}{rrr}  
  \begin{tabular}{l|llll}
  \multirow{5}{*}{$\mathbf{x}(0)$}	& 0 & 1 & 1 & 0 \\
						& 0 & 1 & 0 & 0 \\
						& 1 & 0 & 0 & 0 \\
						& 0 & 0 & 1 & 0 \\
						& 0 & 0 & 0 & 1 \\
  \end{tabular}
  &
  \begin{tabular}{l|llll}
  \multirow{5}{*}{$F(\mathbf{x}(0))$}	& 0 & 0 & 1 & 1 \\
						& 0 & 0 & 1 & 0 \\
						& 0 & 1 & 0 & 0 \\
						& 0 & 0 & 0 & 1 \\
						& 1 & 0 & 0 & 0 \\
  \end{tabular}
  &
  \begin{tabular}{l|llll}
  \multirow{5}{*}{$\mathbf{x}(1)$}	& 0 & 0 & 1 & 1 \\
						& 0 & 0 & 1 & 0 \\
						& 0 & 1 & 0 & 0 \\
						& 1 & 0 & 0 & 1 \\
						& 1 & 0 & 0 & 0 \\
  \end{tabular}
  \\
  \phantom{a}  & \phantom{a}  & \phantom{a} \\
  \begin{tabular}{l|llll}
  \multirow{5}{*}{$F(\mathbf{x}(1))$}	& 1 & 0 & 0 & 1 \\
						& 0 & 0 & 0 & 1 \\
						& 0 & 0 & 1 & 0 \\
						& 1 & 1 & 0 & 0 \\
						& 0 & 1 & 0 & 0 \\
  \end{tabular}
  &
  \begin{tabular}{l|llll}
  \multirow{5}{*}{$\mathbf{x}(2)$}	& 0 & 0 & 0 & 1 \\
						& 0 & 0 & 0 & 1 \\
						& 0 & 0 & 1 & 0 \\
						& 0 & 1 & 0 & 0 \\
						& 0 & 1 & 0 & 0 \\
  \end{tabular}
  &
  \begin{tabular}{l|llll}
  \multirow{5}{*}{$F(\mathbf{x}(2))$}	& 1 & 1 & 0 & 0 \\
						& 1 & 1 & 0 & 0 \\
						& 0 & 1 & 0 & 1 \\
						& 0 & 1 & 1 & 0 \\
						& 0 & 1 & 1 & 0 \\
  \end{tabular}
  \\
  \phantom{a}  & \phantom{a}  & \phantom{a} \\
  \begin{tabular}{l|llll}
  \multirow{5}{*}{$\mathbf{x}(3)$}	& 0 & 1 & 0 & 0 \\
						& 0 & 1 & 0 & 0 \\
						& 0 & 1 & 0 & 1 \\
						& 0 & 1 & 1 & 0 \\
						& 0 & 1 & 1 & 0 \\
  \end{tabular}
  &
  \begin{tabular}{l|llll}
  \multirow{5}{*}{$F(\mathbf{x}(3))$}	& 0 & 1 & 0 & 0 \\
						& 0 & 1 & 0 & 0 \\
						& 1 & 1 & 0 & 0 \\
						& 0 & 1 & 0 & 1 \\
						& 0 & 1 & 0 & 1 \\
  \end{tabular}
  &
  \begin{tabular}{l|llll}
  \multirow{5}{*}{$\mathbf{x}(4)$}	& 1 & 1 & 0 & 0 \\
						& 1 & 1 & 0 & 0 \\
						& 0 & 1 & 0 & 0 \\
						& 0 & 1 & 0 & 1 \\
						& 1 & 1 & 0 & 1 \\
  \end{tabular}
  \\
\end{tabular}  
\end{table*}

\FloatBarrier
\end{example}

\section{Boolean Networks Consisting of XOR Functions}\label{section:XOR_BNs}

In this section, we study the controllability of XOR-BNs whose definition is given below. 

\begin{definition}\label{def:XOR-BNs}
An $n$-node BN is said to be an XOR-BN if for all nodes $x_{i} \in \{ x_{1}, x_{2}, \ldots, x_{n} \}$, there exist distinct nodes $x_{i_{1}}, x_{i_{2}}, \ldots, x_{i_{l}} \in \{ x_{1}, x_{2}, \ldots, x_{n} \}$ such that
\begin{equation}\label{eq:XOR-BN_update_rules}
x_{i}(t + 1) = x_{i_{1}}(t) \oplus x_{i_{2}}(t) \oplus \ldots \oplus x_{i_{l}}(t)
\end{equation}
where $\oplus$ is the XOR logical operator.
\end{definition}

In Section \ref{section:lin_alg_theory}, we use linear algebra to establish some important theorems and algorithms related to the control of XOR-BNs. More specifically, we prove a useful necessary and sufficient condition (Theorem \ref{thm:Wu_spans_Fn_2}) for checking whether a subset of nodes of an XOR-BN can render the BN controllable, and provide algorithms for computing a control-node set of any XOR-BN (Algorithm \ref{alg:construct_control_node_set}) and producing control signals (Algorithm \ref{alg:control_scheme}).
In Section \ref{section:XOR_BNs_bcub}, we provide best-case upper bounds for XOR-BNs with in-degree and out-degree constraints.

\subsection{Linear Algebra Theory of the Control of XOR-BNs}\label{section:lin_alg_theory}

We begin this section by introducing a new perspective to look at XOR-BNs.
Since the set $\{ 0, 1 \}$ can be thought of as the finite field $F_{2}$ of order $2$ and the XOR operator $\oplus$ can be thought of as the addition operation on $F_{2}$, the functional dependencies of an $n$-node XOR-BN can be viewed as a linear map from $F^{n}_{2}$ to $F^{n}_{2}$, or equivalently an $n \times n$ matrix over $F_{2}$.

\begin{example}\label{ex:3_node_XOR_BN}
Consider the $3$-node XOR-BN defined as
\begin{align}
x_{1}(t + 1) &= x_{2}(t) \oplus x_{3}(t) \\
x_{2}(t + 1) &= x_{1}(t) \oplus x_{2}(t) \oplus x_{3}(t) \\
x_{3}(t + 1) &= x_{1}(t)
\end{align}
Then, the matrix $A \in F^{3 \times 3}_{2}$ associated with this XOR-BN is given by
\begin{equation}
A \coloneqq
\left(
\begin{array}{ccc}
0 & 1 & 1 \\
1 & 1 & 1 \\
1 & 0 & 0
\end{array}
\right)
\end{equation}
and the associated linear map $T : F^{3}_{2} \to F^{3}_{2}$ is given by $T(\mathbf{x}) \coloneqq A\mathbf{x}$.

Note that the functional dependencies of the XOR-BN can be succinctly represented as
\begin{equation}
\mathbf{x}(t + 1) = T(\mathbf{x}(t)) = A\mathbf{x}(t).
\end{equation}
\end{example}

Now, we prove an important necessary and sufficient condition for checking whether a subset of nodes of an XOR-BN can render the XOR-BN controllable. This theorem is part of the essential theoretical foundation for our algorithms for computing control-node sets of XOR-BNs and producing control signals.

\begin{theorem}\label{thm:Wu_spans_Fn_2}
Let $\mathcal{B}$ be an $n$-node XOR-BN with control-node set $U$.
Let $A \in F^{n \times n}_{2}$ be the matrix associated with $\mathcal{B}$.
Then, $\mathcal{B}$ is controllable if and only if
\begin{equation}
W_{U} \coloneqq 
\text{Span}\left( \bigcup_{i\, :\, x_{i} \in U} \{ A^{k}\mathbf{e}_{i} : k \in \mathbb{Z}_{\geq 0} \} \right)
= F^{n}_{2}.
\end{equation}
Here, $\mathbf{e}_{1}, \mathbf{e}_{2}, \ldots, \mathbf{e}_{n} \in F^{n}_{2}$ denote the standard unit vectors.
\end{theorem}

\begin{proof}
Let $V \coloneqq \text{Span}(\{ \mathbf{e}_{i} : x_{i} \in U \})$.

First, we prove the ``only if'' direction. Assume that $\mathcal{B}$ is controllable. 
Then, for all $\mathbf{a} \in F^{n}_{2}$, there exists a control scheme 
$\mathbf{u}(0), \mathbf{u}(1), \ldots, \mathbf{u}(t) \in V$ which drives the BN from $\mathbf{0}_{n}$ (time $0$) to
$\mathbf{a}$ (time $t + 1$). Note that
\begin{align}
\mathbf{a} 
&= A( \ldots ( A ( A \mathbf{0}_{n} + \mathbf{u}(0) ) + \mathbf{u}(1) ) \ldots ) + \mathbf{u}(t) \\
&= A^{t+1}\mathbf{0}_{n} + A^{t}\mathbf{u}(0) + A^{t-1}\mathbf{u}(1) + \cdots + A^{0}\mathbf{u}(t) \\
&= A^{t}\mathbf{u}(0) + A^{t-1}\mathbf{u}(1) + \cdots + A^{0}\mathbf{u}(t).
\end{align}
Because for $s = 0, 1, \ldots, t$,
\begin{equation}
A^{s}\mathbf{u}(t - s) \in \{ A^{s}\mathbf{v} : \mathbf{v} \in V \} \subseteq W_{U},
\end{equation}
we have
\begin{equation}
\mathbf{a} = A^{t}\mathbf{u}(0) + A^{t-1}\mathbf{u}(1) + \cdots + A^{0}\mathbf{u}(t) \in W_{U}.
\end{equation}
Therefore, we have $F^{n}_{2} \subseteq W_{U}$. On the other hand, $W_{U} \subseteq F^{n}_{2}$ clearly holds.
Hence, $W_{U} = F^{n}_{2}$. This proves the ``only if'' direction of the theorem.

Second, we prove the ``if'' direction of the theorem. Since $W_{U} = F^{n}_{2}$, we can find 
$A^{k_{1}}\mathbf{e}_{\alpha_{1}}$, $A^{k_{2}}\mathbf{e}_{\alpha_{2}}$, \ldots,
$A^{k_{n}}\mathbf{e}_{\alpha_{n}}$ such that $0 \leq k_{1} \leq k_{2} \leq \cdots \leq k_{n}$,
$x_{\alpha_{1}}, x_{\alpha_{2}}, \ldots, x_{\alpha_{n}} \in U$ and 
$A^{k_{1}}\mathbf{e}_{\alpha_{1}}$, $A^{k_{2}}\mathbf{e}_{\alpha_{2}}$, \ldots,
$A^{k_{n}}\mathbf{e}_{\alpha_{n}}$ form a basis of $F^{n}_{2}$.
Fix arbitrary $\mathbf{a}, \mathbf{b} \in F^{n}_{2}$. We will construct a control scheme which drives the BN from 
$\mathbf{a}$ (time $0$) to $\mathbf{b}$ (time $k_{n} + 1$). Note that there exist
$A^{l_{1}}\mathbf{e}_{\beta_{1}}, A^{l_{2}}\mathbf{e}_{\beta_{2}}, \ldots, A^{l_{m}}\mathbf{e}_{\beta_{m}} \in
\{ A^{k_{1}}\mathbf{e}_{\alpha_{1}}, A^{k_{2}}\mathbf{e}_{\alpha_{2}}, \ldots, A^{k_{n}}\mathbf{e}_{\alpha_{n}} \}$
such that $0 \leq l_{1} \leq l_{2} \leq \ldots \leq l_{m} \leq k_{n}$ and
\begin{align}
&\mathbf{b} - A^{k_{n} + 1}\mathbf{a} = 
A^{l_{1}}\mathbf{e}_{\beta_{1}} + A^{l_{2}}\mathbf{e}_{\beta_{2}} + \cdots + A^{l_{m}}\mathbf{e}_{\beta_{m}} \\
\Rightarrow\ 
&\mathbf{b} = A^{k_{n} + 1}\mathbf{a} + \left( A^{l_{1}}\mathbf{e}_{\beta_{1}} 
+ A^{l_{2}}\mathbf{e}_{\beta_{2}} + \cdots + A^{l_{m}}\mathbf{e}_{\beta_{m}} \right).
\end{align}
By combining terms in $A^{l_{1}}\mathbf{e}_{\beta_{1}} + A^{l_{2}}\mathbf{e}_{\beta_{2}} + \cdots 
+ A^{l_{m}}\mathbf{e}_{\beta_{m}}$ with the same coefficient $A^{l}$, we write
\begin{equation}
\mathbf{b} = A^{k_{n} + 1}\mathbf{a} + \left( A^{s_{1}}\mathbf{v}_{1} + A^{s_{2}}\mathbf{v}_{2} + \cdots 
+ A^{s_{d}}\mathbf{v}_{d} \right)
\end{equation}
where $0 \leq s_{1} < s_{2} < \cdots < s_{d} \leq k_{n}$ and 
$\mathbf{v}_{1}, \mathbf{v}_{2}, \ldots, \mathbf{v}_{d} \in V$. Then, we can see that the control scheme 
$\mathbf{u}_{0}$, $\mathbf{u}(1)$, \ldots, $\mathbf{u}(k_{n})$ defined as 
$\mathbf{u}(k_{n} - s_{1}) \coloneqq \mathbf{v}_{1}$, $\mathbf{u}(k_{n} - s_{2}) \coloneqq \mathbf{v}_{2}$, \ldots,
$\mathbf{u}(k_{n} - s_{d}) \coloneqq \mathbf{v}_{d}$ and $\mathbf{u}(t) \coloneqq \mathbf{0}_{n}$ for all
$t \in \{ 0, 1, \ldots, k_{n} \} \setminus \{ k_{n} - s_{1}, k_{n} - s_{2}, \ldots, k_{n} - s_{d} \}$
can drive the BN from state $\mathbf{a}$ (time $0$) to state $\mathbf{b}$ (time $k_{n} + 1$).
This proves the ``if'' part of the theorem.
\end{proof}

Before we provide our algorithms related to the control of XOR-BNs, we need to present three more lemmas upon which these algorithms are based.

\begin{lemma}\label{lemma:M1}
Let $\mathbf{v}$ be a non-zero vector in $F^{n}_{2}$ and fix arbitrary $A \in F^{n \times n}_{2}$.
There exists a non-negative integer $k$ such that $(\mathbf{v}, A\mathbf{v}, \ldots, A^{k}\mathbf{v})$ forms a basis of $\text{Span}(\left\{ A^{l}\mathbf{v} : l \in \mathbb{Z}_{\geq 0} \right\})$.
\end{lemma}

\begin{proof}
Let $k$ be the greatest non-negative integer such that $\mathbf{v}$, $A\mathbf{v}$, \ldots, $A^{k}\mathbf{v}$ are linearly independent. To prove that $(\mathbf{v}, A\mathbf{v}, \ldots, A^{k}\mathbf{v})$ forms a basis of 
$\text{Span}(\left\{ A^{l}\mathbf{v} : l \in \mathbb{Z}_{\geq 0} \right\})$, it suffices to show that for all integers
$l \geq k + 1$, $A^{l}\mathbf{v} \in \text{Span}(\{ \mathbf{v}, A\mathbf{v}, \ldots, A^{k}\mathbf{v} \})$. Let's call this statement $P(l)$. We will provide an inductive proof.

Consider the base case $P(k+1)$. Since $\mathbf{v}$, $A\mathbf{v}$, \ldots, $A^{k}\mathbf{v}$ are linearly independent and $\mathbf{v}$, $A\mathbf{v}$, \ldots, $A^{k}\mathbf{v}$, $A^{k+1}\mathbf{v}$ are linearly dependent, we must have 
$A^{k+1}\mathbf{v} \in \text{Span}(\{ \mathbf{v}, A\mathbf{v}, \ldots, A^{k}\mathbf{v} \})$. Therefore, $P(k+1)$ is true.

We now turn to the inductive step. Assume that $l \geq k + 1$ is an integer such that $P(l)$ is true, which means that $A^{l}\mathbf{v} \in \text{Span}(\{ \mathbf{v}, A\mathbf{v}, \ldots, A^{k}\mathbf{v} \})$. Then, 
$A^{l+1}\mathbf{v} \in \text{Span}(\{ A\mathbf{v}, A^{2}\mathbf{v}, \ldots, A^{k+1}\mathbf{v} \})$. 
By $P(k+1)$, $A^{k+1}\mathbf{v} \in \text{Span}(\{ \mathbf{v}, A\mathbf{v}, \ldots, A^{k}\mathbf{v} \})$. Hence,
\begin{equation}
A^{l+1}\mathbf{v} 
\in \text{Span}(\{ A\mathbf{v}, A^{2}\mathbf{v}, \ldots, A^{k+1}\mathbf{v} \})
\subseteq \text{Span}(\{ \mathbf{v}, A\mathbf{v}, \ldots, A^{k}\mathbf{v} \}).
\end{equation}
Therefore, $P(l + 1)$ is true, and the inductive step is complete.

By mathematical induction, we conclude that for all integers $l \geq k + 1$, $P(l)$ is true.
\end{proof}

\begin{lemma}\label{lemma:M2a}
Let $d \geq 1$ be an integer and $\mathbf{v}_{1}$, $\mathbf{v}_{2}$, \ldots, $\mathbf{v}_{d+1}$ be non-zero vectors in $F^{n}_{2}$. Fix arbitrary matrix $A \in F^{n \times n}_{2}$. If
\begin{equation}
\mathbf{v}_{d+1} \in \sum^{d}_{j = 1} \text{Span}(\{ A^{l}\mathbf{v}_{j} : l \in \mathbb{Z}_{\geq 0} \}) \eqqcolon W_{1},
\end{equation}
then for all positive integers $m$, $A^{m}\mathbf{v}_{d+1} \in W_{1}$.
\end{lemma}

\begin{proof}
Suppose that $\mathbf{v}_{d+1} \in W_{1}$. Let $m$ be a positive integer. We have
\begin{equation}
A^{m}\mathbf{v}_{d+1} \in \sum^{d}_{j = 1} \text{Span}(\{ A^{l}\mathbf{v}_{j} : l \in \mathbb{Z}_{\geq m} \})
\subseteq W_{1}.
\end{equation}
\end{proof}

\begin{lemma}\label{lemma:M2b}
Let $d \geq 1$ be an integer and $\mathbf{v}_{1}$, $\mathbf{v}_{2}$, \ldots, $\mathbf{v}_{d+1}$ be non-zero vectors in $F^{n}_{2}$. Fix arbitrary matrix $A \in F^{n \times n}_{2}$. If $k$ is a positive integer such that
\begin{equation}\label{eq:W2}
A^{k}\mathbf{v}_{d+1} \in 
\text{Span}(\{ \mathbf{v}_{d+1}, A\mathbf{v}_{d+1}, \ldots, A^{k-1}\mathbf{v}_{d+1} \}) +
\sum^{d}_{j = 1} \text{Span}(\{ A^{l}\mathbf{v}_{j} : l \in \mathbb{Z}_{\geq 0} \}) \eqqcolon W_{2},
\end{equation}
then for all integers $m > k$, $A^{m}\mathbf{v}_{d+1} \in W_{2}$.
\end{lemma}

\begin{proof}
Assume that $k$ is a positive integer such that $A^{k}\mathbf{v}_{d+1} \in W_{2}$. For all integers $m > k$, let $P(m)$ be the statement that $A^{m}\mathbf{v}_{d+1} \in W_{2}$. We are going to prove by induction that $P(m)$ holds for all integers $m > k$.

We begin by proving the base case. By multiplying both sides of Eq.\@ (\ref{eq:W2}) by $A$, we get
\begin{equation}
A^{k+1}\mathbf{v}_{d+1} \in 
\text{Span}(\{ A\mathbf{v}_{d+1}, A^{2}\mathbf{v}_{d+1}, \ldots, A^{k}\mathbf{v}_{d+1} \}) +
\sum^{d}_{j = 1} \text{Span}(\{ A^{l}\mathbf{v}_{j} : l \in \mathbb{Z}_{\geq 1} \}) \eqqcolon W_{3}.
\end{equation}
Note that $W_{3} \subseteq W_{2}$ because $A^{k}\mathbf{v}_{d+1} \in W_{2}$. This proves $P(k+1)$.

To prove the inductive step, we assume that $P(m)$ holds for some integer $m > k$. 
Hence, $A^{m}\mathbf{v}_{d+1} \in W_{2}$, which implies that
$A^{m+1}\mathbf{v}_{d+1} \in W_{3}$.
Since $W_{3} \subseteq W_{2}$, we deduce that $A^{m+1}\mathbf{v}_{d+1} \in W_{2}$.
Hence, $P(m+1)$ is true, and the inductive step is complete.

By mathematical induction, $P(m)$ is true for all integers $m > k$.
\end{proof}

Based on Theorem \ref{thm:Wu_spans_Fn_2} and Lemmas \ref{lemma:M1}, \ref{lemma:M2a}, \ref{lemma:M2b}, we propose an algorithm for computing a control-node set of any XOR-BN 
(Algorithm \ref{alg:construct_control_node_set}).

\begin{algorithm}
\caption{Construct Control-Node Sets of XOR-BNs}
\label{alg:construct_control_node_set}
\begin{algorithmic}[1]
\Require An $n$-node XOR-BN with associated matrix $A \in F^{n \times n}_{2}$.

\Ensure A control-node set $U \subseteq \{ x_{1}, x_{2}, \ldots, x_{n} \}$ which renders the BN controllable.

\Statex

\State Set $U \gets \varnothing$, $V \gets \varnothing$, $i \gets 1$.

\While{$|V| < n$}

\If{$\mathbf{e}_{i} \notin \text{Span}(V)$}

\State Set $U \gets U \cup \{ x_{i} \}$, $V \gets V \cup \{ \mathbf{e}_{i} \}$, $k \gets 1$.

\While{$A^{k}\mathbf{e}_{i} \notin \text{Span}(V)$}

\State Set $V \gets V \cup \{ A^{k}\mathbf{e}_{i} \}$.

\State Increment $k$ by $1$.

\EndWhile 
\EndIf

\State Increment $i$ by $1$.
\EndWhile 

\State \Output $U$
\end{algorithmic}
\end{algorithm}

\FloatBarrier

The major computational cost of Algorithm \ref{alg:construct_control_node_set} comes from testing whether a vector in $F^{n}_{2}$ lies in the span of $V$ (lines 3 and 5), which requires $O(n^{3})$ time per execution. Moreover, note that after lines 2--8 are completely executed, $V$ will become a basis of the subspace spanned by $\mathbf{e}_{1}$, $\mathbf{e}_{2}$, \ldots, $\mathbf{e}_{n}$, $A\mathbf{e}_{1}$, $A\mathbf{e}_{2}$, \ldots, $A\mathbf{e}_{n}$, $A^{2}\mathbf{e}_{1}$, $A^{2}\mathbf{e}_{2}$, \ldots, $A^{2}\mathbf{e}_{n}$, \ldots (i.e., $V$ is a basis of $F^{n}_{2}$ and hence $|V| = n$). Therefore, we can see that lines 3 and 5 of the algorithm will be executed for $O(n)$ number of times. We conclude that the worst-case time complexity of Algorithm \ref{alg:construct_control_node_set} is $O(n^{4})$.

Next, we present an algorithm for constructing control schemes for any XOR-BNs (Algorithm \ref{alg:control_scheme}), which relies on another intermediate algorithm (Algorithm \ref{alg:intermediate}). 
Algorithms \ref{alg:intermediate} and \ref{alg:control_scheme} are also based on Theorem \ref{thm:Wu_spans_Fn_2} and Lemmas \ref{lemma:M1}, \ref{lemma:M2a}, \ref{lemma:M2b}.

\begin{algorithm}
\caption{An Intermediate Algorithm Needed for the Construction of Control Signals}
\label{alg:intermediate}
\begin{algorithmic}[1]
\Require An $n$-node controllable XOR-BN with control-node set $U = \{ x_{j_{1}}, x_{j_{2}}, \ldots, x_{j_{l}} \}$ and
associated matrix $A \in F^{n \times n}_{2}$.

\Ensure A list $[(i_{1}, k_{1}), (i_{2}, k_{2}), \ldots, (i_{n}, k_{n})]$ such that 
$i_{1}, i_{2}, \ldots, i_{n} \in \{ j_{1}, j_{2}, \ldots, j_{l} \}$, $k_{1}, k_{2}, \ldots, k_{n} \in \mathbb{Z}_{\geq 0}$ and
$(A^{k_{1}}\mathbf{e}_{i_{1}}, A^{k_{2}}\mathbf{e}_{i_{2}}, \ldots, A^{k_{n}}\mathbf{e}_{i_{n}})$ forms a basis of 
$F^{n}_{2}$.

\Statex

\State Set $L$ to an empty list. Set $V \gets \varnothing$, $d \gets 1$.

\While{the length of $L$ is less than $n$}

\State Set $k \gets 0$.

\While{$A^{k}\mathbf{e}_{j_{d}} \notin \text{Span}(V)$}

\State Set $V \gets V \cup \{ A^{k}\mathbf{e}_{j_{d}} \}$. Append $(j_{d}, k)$ to $L$.

\State Increment $k$ by $1$.
\EndWhile 

\State Increment $d$ by $1$.
\EndWhile 

\State \Output $L$
\end{algorithmic}
\end{algorithm}

\FloatBarrier

The worst-case time complexity of Algorithm \ref{alg:intermediate} is $O(n^{4})$, and the major computational cost of the algorithm comes from testing whether a vector in $F^{n}_{2}$ lies in the span of $V$ (line 4).

\begin{algorithm}
\caption{Construct a Control Scheme for any XOR-BN}
\label{alg:control_scheme}
\begin{algorithmic}[1]
\Require
An $n$-node controllable XOR-BN with control-node set $U = \{ x_{j_{1}}, x_{j_{2}}, \ldots, x_{j_{l}} \}$ and associated matrix $A \in F^{n \times n}_{2}$; \\
an initial state $\mathbf{a} \in F^{n}_{2}$ and a target state $\mathbf{b} \in F^{n}_{2}$; \\
$i_{1}, i_{2}, \ldots, i_{n} \in \{ j_{1}, j_{2}, \ldots, j_{l} \}$ and $k_{1}, k_{2}, \ldots, k_{n} \in \mathbb{Z}_{\geq 0}$ (which can be obtained using Algorithm \ref{alg:intermediate}) such that 
$(A^{k_{1}}\mathbf{e}_{i_{1}}, A^{k_{2}}\mathbf{e}_{i_{2}}, \ldots, A^{k_{n}}\mathbf{e}_{i_{n}})$ forms a basis of 
$F^{n}_{2}$.

\Ensure A control scheme which drives the BN from state $\mathbf{a}$ to state $\mathbf{b}$.

\Statex

\State Compute $k^{*} \coloneqq \max(k_{1}, k_{2}, \ldots, k_{n})$.

\State Find $c_{1}, c_{2}, \ldots, c_{n} \in F_{2}$ such that 
$\mathbf{b} - A^{k^{*} + 1}\mathbf{a} = 
c_{1}A^{k_{1}}\mathbf{e}_{i_{1}} + c_{2}A^{k_{2}}\mathbf{e}_{i_{2}} + \cdots + c_{n}A^{k_{n}}\mathbf{e}_{i_{n}}$, which is a system of $n$ linear equations in $n$ variables.

\State For all integers $t \in [0, k^{*}]$ and $j \in \{ j_{1}, j_{2}, \ldots, j_{l} \}$, set the control signal $u_{j}(t) \gets 0$.

\For{$s = 1, 2, \ldots, n$}
\If{$c_{s} = 1$}
\State Set $u_{i_{s}}(k^{*} - k_{s}) \gets 1$.
\EndIf
\EndFor

\State \Output the control signals $\mathbf{u}(0)$, $\mathbf{u}(1)$, \ldots, $\mathbf{u}(k^{*})$
\end{algorithmic}
\end{algorithm}

\FloatBarrier

In Algorithm \ref{alg:control_scheme}, computing $A^{k^{*} + 1}\mathbf{a}$ (line 2) requires $O((\max(k_{1}, k_{2}, \ldots, k_{n}) + 1)n^{2})$ time.
Computing $A^{k_{1}}\mathbf{e}_{i_{1}}$, $A^{k_{2}}\mathbf{e}_{i_{2}}$, \ldots, $A^{k_{n}}\mathbf{e}_{i_{n}}$ (line 2) requires $O(\max(k_{1}, k_{2}, \ldots, k_{n})n^{3})$ time. 
Solving the system of linear equations $\mathbf{b} - A^{k^{*} + 1}\mathbf{a} = 
c_{1}A^{k_{1}}\mathbf{e}_{i_{1}} + c_{2}A^{k_{2}}\mathbf{e}_{i_{2}} + \cdots + c_{n}A^{k_{n}}\mathbf{e}_{i_{n}}$ (line 2) requires $O(n^{3})$ time. Since the major computational cost of Algorithm \ref{alg:control_scheme} comes from line 2, we deduce that the worst-case time complexity of the algorithm is 
$O((\max(k_{1}, k_{2}, \ldots, k_{n}) + 1)n^{3})$.

Note that if $i_{1}, i_{2}, \ldots, i_{n}, k_{1}, k_{2}, \ldots, k_{n}$ are obtained using Algorithm \ref{alg:intermediate}, then $k_{1}, k_{2}, \ldots, k_{n} < n$. To see this, we assume for a contradiction that $k_{m} \geq n$ for some $m \in \{ 1, 2, \ldots, n \}$. From Algorithm \ref{alg:intermediate}, we can see that if $A^{k_{m}}\mathbf{e}_{i_{m}}$ is inside the set $V$ after Algorithm \ref{alg:intermediate} is fully executed, then $\mathbf{e}_{i_{m}}$, $A\mathbf{e}_{i_{m}}$, \ldots, $A^{k_{m} - 1}\mathbf{e}_{i_{m}}$ are also inside $V$ after the algorithm is fully executed. 
Therefore, $|V| \geq k_{m} + 1 \geq n + 1$.
But $V$ forms a basis of $F^{n}_{2}$ (by Theorem \ref{thm:Wu_spans_Fn_2}) and hence $|V| = n$. A contradiction is reached. Therefore, we conclude that $k_{1}, k_{2}, \ldots, k_{n} < n$, which implies that 
$\max(k_{1}, k_{2}, \ldots, k_{n}) < n$ and the worst-case time complexity of Algorithm \ref{alg:control_scheme} is $O(n^{4})$.

\begin{example}
Consider the BN defined in Example \ref{ex:3_node_XOR_BN}. By Theorem \ref{thm:Wu_spans_Fn_2}, we can deduce that $U = \{ x_{1} \}$ does not render the BN controllable. To see this, note that
\begin{equation}
A \mathbf{e}_{1} = 
\left(
\begin{array}{c}
0 \\
1 \\
1 
\end{array}
\right), \quad
A^{m} \mathbf{e}_{1} = 
\left(
\begin{array}{c}
0 \\
0 \\
0 
\end{array}
\right) \quad \text{for all integers $m \geq 2$.}
\end{equation}
Hence,
$\text{Span}(\{ \mathbf{e}_{1}, A\mathbf{e}_{1}, A^{2}\mathbf{e}_{1}, \ldots \})
= \text{Span}(\{ \mathbf{e}_{1}, A\mathbf{e}_{1} \}) \neq F^{3}_{2}$.

On the other hand,
\begin{equation}
A \mathbf{e}_{2} = 
\left(
\begin{array}{c}
1 \\
1 \\
0 
\end{array}
\right), \quad
A^{m} \mathbf{e}_{2} = 
\left(
\begin{array}{c}
1 \\
0 \\
1 
\end{array}
\right) \quad \text{for all integers $m \geq 2$.}
\end{equation}
Therefore, $\text{Span}(\{ \mathbf{e}_{2}, A\mathbf{e}_{2}, A^{2}\mathbf{e}_{2}, \ldots \})
= \text{Span}(\{ \mathbf{e}_{2}, A\mathbf{e}_{2}, A^{2}\mathbf{e}_{2} \}) = F^{3}_{2}$.
Hence, $(\mathbf{e}_{2}, A\mathbf{e}_{2}, A^{2}\mathbf{e}_{2})$ forms a basis of $F^{3}_{2}$ and
$U = \{ x_{2} \}$ renders the BN controllable.

Suppose that we adopt the control-node set $U \coloneqq \{ x_{2} \}$ and we want to drive the BN from the state 
$\mathbf{a} \coloneqq 001$ to the state $\mathbf{b} \coloneqq 010$. We will illustrate the idea behind constructing the desired control signals using Algortihm \ref{alg:control_scheme}. First, we consider the powers of the coefficient $A$ in $\mathbf{e}_{2}$, $A\mathbf{e}_{2}$, $A^{2}\mathbf{e}_{2}$ and set 
$k^{*} \coloneqq \max(0, 1, 2) = 2$. Second, we solve the following system of linear equations for $c_{1}$, $c_{2}$, $c_{3}$:
\begin{align}
&\mathbf{b}^{T} - A^{k^{*} + 1}\mathbf{a}^{T} 
= c_{1}\mathbf{e}_{2} + c_{2}A\mathbf{e}_{2} + c_{3}A^{2}\mathbf{e}_{2} \label{eq:linalg_star_1} \\
\Leftrightarrow\ &
\left(
\begin{array}{ccc}
0 & 1 & 1 \\
1 & 1 & 0 \\
0 & 0 & 1 
\end{array}
\right)
\left(
\begin{array}{c}
c_{1} \\
c_{2} \\
c_{3}
\end{array}
\right)
=
\left(
\begin{array}{c}
1 \\
1 \\
1
\end{array}
\right) \\
\Leftrightarrow\ &c_{1} = 1, c_{2} = 0, c_{3} = 1.
\end{align}
Then, we re-write Eq.\@ (\ref{eq:linalg_star_1}) in the following way:
\begin{align}
\mathbf{b}^{T} 
&= c_{1}\mathbf{e}_{2} + c_{2}A\mathbf{e}_{2} + c_{3}A^{2}\mathbf{e}_{2} + A^{3}\mathbf{a}^{T} \\
&= c_{1}\mathbf{e}_{2} + A\left[ c_{2}\mathbf{e}_{2} + c_{3}A\mathbf{e}_{2} + A^{2}\mathbf{a}^{T} \right] \\
&= c_{1}\mathbf{e}_{2} + A\left[ c_{2}\mathbf{e}_{2} + A\left[ c_{3}\mathbf{e}_{2} + A\mathbf{a}^{T} \right] \right] \\
&= A\left[ A \left[ A\mathbf{a}^{T} + 1 \cdot \mathbf{e}_{2} \right] + 0 \cdot \mathbf{e}_{2} \right] 
+ 1 \cdot \mathbf{e}_{2}. \label{eq:linalg_star_2}
\end{align}
From Eq.\@ (\ref{eq:linalg_star_2}), we can see that the control scheme 
$\mathbf{u}(0) = 010 = 1 \cdot \mathbf{e}^{T}_{2}$, $\mathbf{u}(1) = 000 = 0 \cdot \mathbf{e}^{T}_{2}$,
$\mathbf{u}(2) = 010 = 1 \cdot \mathbf{e}^{T}_{2}$ can drive the BN from state $\mathbf{a}$ (time $0$) 
to state $\mathbf{b}$ (time $3$).
\end{example}

\subsection{Best-Case Upper Bounds}\label{section:XOR_BNs_bcub}

Let $k \geq 2$ be an integer. If an $n$-node XOR-BN satisfies that each node has exactly $k$ inputs and $k$ outputs, then the BN is said to be a $k$-$k$-XOR-BN.
We remark that in the $n \times n$ matrix associated with a $k$-$k$-XOR-BN, each row has exactly $k$ $1$-entries and each column has exactly $k$ $1$-entries. In this section, we will present some best-case upper bounds for $k$-$k$-XOR-BNs 
(Theorems \ref{thm:k-k-XOR-BN_(k-1)-bound} and \ref{thm:1_control_node}).

\begin{theorem}\label{thm:k-k-XOR-BN_(k-1)-bound}
Let $k \geq 2$ and $n > k$ be integers. Consider the following $n$-node $k$-$k$-XOR-BN with control node set $U \coloneqq \{ x_{n-k+2}, x_{n-k+3}, \ldots, x_{n} \}$ (hence, $|U| = k - 1$):
\begin{align}
x_{1}(t + 1)	&= x_{1}(t) \oplus x_{2}(t) \oplus \cdots \oplus x_{k}(t) \\
x_{2}(t + 1)	&= x_{2}(t) \oplus x_{3}(t) \oplus \cdots \oplus x_{k+1}(t) \\
		& \vdots \nonumber \\
x_{n}(t + 1)	&= x_{n}(t) \oplus x_{1}(t) \oplus \cdots \oplus x_{k-1}(t) 
\end{align}
This BN is controllable.
\end{theorem}

\begin{proof}
Let $A \in F^{n \times n}_{2}$ be the matrix associated with this BN. For $i = 1, 2, \ldots, n$, let $m_{i}$ be the unique integer in $\{ n - k + 2, n - k + 3, \ldots, n \}$ such that $m_{i} - i \equiv 0 \pmod{k - 1}$. Note that
$\mathbf{v}^{(i)} \coloneqq A^{\frac{m_{i} - i}{k - 1}} \mathbf{e}_{m_{i}}$ satisfies $v^{(i)}_{i} = 1$ and
$v^{(i)}_{j} = 0$ for all integers $j \in [1, i - 1]$.
Consider the matrix
\begin{equation}
B \coloneqq
\left(
\begin{array}{cccc}
|			& |				&		& | \\
\mathbf{v}^{(1)}	& \mathbf{v}^{(2)}	& \cdots 	& \mathbf{v}^{(n)} \\
|			& |				&		& | 
\end{array}
\right)
\in F^{n \times n}_{2}.
\end{equation}
Because $B$ is a lower triangular matrix with all diagonal entries equal to $1$, $B$ is invertible.
Hence, $\text{Span}(\mathbf{v}^{(1)}, \mathbf{v}^{(2)}, \ldots, \mathbf{v}^{(n)}) = F^{n}_{2}$.
Since $\mathbf{v}^{(1)}, \mathbf{v}^{(2)}, \ldots, \mathbf{v}^{(n)} \in W_{U}$, we have $W_{U} = F^{n}_{2}$.
By Theorem \ref{thm:Wu_spans_Fn_2}, the $k$-$k$-XOR-BN is controllable.
\end{proof}

\begin{example}
Consider the following $6$-node $3$-$3$-XOR-BN with control-node set $U = \{ x_{5}, x_{6} \}$:
\begin{align}
x_{1}(t + 1)	&= x_{1}(t) \oplus x_{2}(t) \oplus x_{3}(t) \\
x_{2}(t + 1)	&= x_{2}(t) \oplus x_{3}(t) \oplus x_{4}(t) \\
x_{3}(t + 1)	&= x_{3}(t) \oplus x_{4}(t) \oplus x_{5}(t) \\
x_{4}(t + 1)	&= x_{4}(t) \oplus x_{5}(t) \oplus x_{6}(t) \\
x_{5}(t + 1)	&= x_{5}(t) \oplus x_{6}(t) \oplus x_{1}(t) \\
x_{6}(t + 1)	&= x_{6}(t) \oplus x_{1}(t) \oplus x_{2}(t)
\end{align}
The matrix associated with this BN is
\begin{equation}
A = 
\left(
\begin{array}{cccccc}
1 & 1 & 1 & 0 & 0 & 0 \\
0 & 1 & 1 & 1 & 0 & 0 \\
0 & 0 & 1 & 1 & 1 & 0 \\
0 & 0 & 0 & 1 & 1 & 1 \\
1 & 0 & 0 & 0 & 1 & 1 \\
1 & 1 & 0 & 0 & 0 & 1
\end{array}
\right)
\in F^{6 \times 6}_{2}.
\end{equation}
The values of $m_{1}$, $m_{2}$, \ldots, $m_{6}$ are shown in the following table:
\begin{center}
\begin{tabular}{c|c|c|c|c|c|c}
 $i$		& 1	& 2	& 3	& 4	& 5	& 6 	\\ \hline
$m_{i}$	& 5	& 6	& 5	& 6	& 5	& 6
\end{tabular}
\end{center}
We have
\begin{align}
\mathbf{v}^{(6)} &\coloneqq A^{\frac{m_{6} - 6}{3 - 1}} \mathbf{e}_{m_{6}} = A^{0} \mathbf{e}_{6} 
= [0, 0, 0, 0, 0, 1]^{T} \\
\mathbf{v}^{(5)} &\coloneqq A^{\frac{m_{5} - 5}{3 - 1}} \mathbf{e}_{m_{5}} = A^{0} \mathbf{e}_{5} 
= [0, 0, 0, 0, 1, 0]^{T} \\
\mathbf{v}^{(4)} &\coloneqq A^{\frac{m_{4} - 4}{3 - 1}} \mathbf{e}_{m_{4}} = A \mathbf{e}_{6} 
= [0, 0, 0, 1, 1, 1]^{T} \\
\mathbf{v}^{(3)} &\coloneqq A^{\frac{m_{3} - 3}{3 - 1}} \mathbf{e}_{m_{3}} = A \mathbf{e}_{5} 
= [0, 0, 1, 1, 1, 0]^{T} \\
\mathbf{v}^{(2)} &\coloneqq A^{\frac{m_{2} - 2}{3 - 1}} \mathbf{e}_{m_{2}} = A^{2} \mathbf{e}_{6} 
= [0, 1, 0, 1, 0, 1]^{T} \\
\mathbf{v}^{(1)} &\coloneqq A^{\frac{m_{1} - 1}{3 - 1}} \mathbf{e}_{m_{1}} = A^{2} \mathbf{e}_{5} 
= [1, 0, 1, 0, 1, 0]^{T}
\end{align}
It can be easily seen that $\text{Span}(\mathbf{v}^{(1)}, \mathbf{v}^{(2)}, \mathbf{v}^{(3)}, \mathbf{v}^{(4)},
\mathbf{v}^{(5)}, \mathbf{v}^{(6)}) = F^{6}_{2}$. By Theorem \ref{thm:Wu_spans_Fn_2}, the BN is controllable.
\end{example}

Theorem \ref{thm:k-k-XOR-BN_(k-1)-bound} states that for all integers $k$, $n$ such that $n > k \geq 2$, there exists an $n$-node controllable $k$-$k$-XOR-BN with $k-1$ control nodes. From now on, we are going to show that if $n = 2^{m}$ (where $m \geq 2$ is an integer) and $k \in [3, n-1]$ is odd, then there exists an $n$-node controllable $k$-$k$-XOR-BN with $1$ control node (Theorem \ref{thm:1_control_node}). 
To do so, first we present some definitions and facts in ring theory.

\begin{definition}\label{def:commutative_ring}
A commutative ring $R$ is a set with two binary operations $+$, $\cdot$ and two distinct elements $0$, $1$ such that
\begin{enumerate}
\item $(R, +, 0)$ forms an abelian group;

\item $(R, \cdot, 1)$ forms a monoid and for all $a, b \in R$, $a \cdot b = b \cdot a$;

\item for all $a, b, c \in R$, $a \cdot (b + c) = (a \cdot b) + (a \cdot c)$.
\end{enumerate}
\end{definition}

\begin{example}
The set $F_{2}[t]$ of all polynomials in the variable $t$ over $F_{2}$ equipped with the usual addition operation $+$ and the usual multiplication operation $\cdot$ forms a commutative ring. This ring is called the polynomial ring over $F_{2}$.
\end{example}

\begin{definition}
An element $u$ of a commutative ring $R$ is said to be a unit if there exists $v \in R$ such that $vu = uv = 1$.
\end{definition}

\begin{example}
$1$ is the only unit in $F_{2}[t]$.
\end{example}

\begin{definition}
An integral domain (ID) is a commutative ring such that the product of any two non-zero elements is non-zero.
\end{definition}

\begin{definition}
An irreducible element of an ID is a non-zero, non-unit element which is not the product of two non-unit elements in the ID.
\end{definition}

\begin{example}
$F_{2}[t]$ is an ID. $t + 1$ is an irreducible element of $F_{2}[t]$. $t^{4} + 1$ is not an irreducible element of $F_{2}[t]$ because $t^{4} + 1 = (t + 1)^{4}$.
\end{example}

\begin{definition}
A unique factorization domain (UFD) is an ID $R$ in which every non-zero, non-unit element $x$ of $R$ can be expressed as a finite product of irreducible elements $p_{1}$, $p_{2}$, \ldots, $p_{n}$ of $R$ (where $n \geq 1$):
\begin{equation}
x = p_{1} p_{2} \ldots p_{n},
\end{equation}
and if $q_{1}$, $q_{2}$, \ldots, $q_{m}$ are irreducible elements of $R$ such that $x = q_{1} q_{2} \ldots q_{m}$, then $m = n$, and there exists a bijection $f : \{ 1, 2, \ldots, n \} \to \{ 1, 2, \ldots, n \}$ and units 
$u_{1}, u_{2}, \ldots, u_{n} \in R$ satisfying $p_{i} = u_{i} q_{f(i)}$ for $i = 1, 2, \ldots, n$.
\end{definition}

\begin{definition}
A Euclidean domain (ED) is an ID $R$ such that there exists a function 
$f : R \setminus \{ 0 \} \to \mathbb{Z}_{\geq 0}$ satisfying that for all $a, b \in R$ with $b \neq 0$, there exist
$q, r \in R$ such that $a = bq + r$ and either $r = 0$ or $f(r) < f(b)$. Such a function is said to be a Euclidean function.
\end{definition}

\begin{example}
$F_{2}[t]$ is an ED. A Euclidean function $f : F_{2}[t] \setminus \{ 0 \} \to \mathbb{Z}_{\geq 0}$ is given by
$f(p(t)) \coloneqq \text{deg}(p(t))$.
\end{example}

\begin{definition}
Let $R$ be a commutative ring. An ideal is a subset $I$ of $R$ such that $(I, +)$ is a subgroup of $(R, +)$, and for all $r \in R$, $x \in I$, $rx \in I$.
\end{definition}

\begin{definition}
Let $R$ be a commutative ring. For all $a \in R$, we define $\langle a \rangle \coloneqq \{ ar : r \in R \}$.
$\langle a \rangle$ is called the principal ideal of $R$ generated by $a$. It can be easily seen that 
$\langle a \rangle$ is an ideal of $R$.
\end{definition}

\begin{definition}
A principal ideal domain (PID) is an ID in which every ideal is a principal ideal.
\end{definition}

\begin{theorem}
$F_{2}[t]$ is an ED, PID and UFD.
\end{theorem}

\begin{definition}
Let $R$ be a commutative ring and $I$ be an ideal of $R$. For all $a \in R$, we define
$I + a \coloneqq \{ r + a : r \in I \}$.
\end{definition}

\begin{proposition}
Let $R$ be a commutative ring and $I$ be an ideal of $R$. For all $a, b \in R$, we either have $I + a = I + b$ or 
$(I + a) \cap (I + b) = \varnothing$.
\end{proposition}

\begin{definition}
Let $R$ be a commutative ring and $I$ be an ideal of $R$. The quotient ring 
$\frac{R}{I} \coloneqq \{ I + a : a \in R \}$ is a commutative ring equipped with addition and multiplication defined as: for all $a, b \in R$, $(I + a) + (I + b) \coloneqq I + (a + b)$ and $(I + a)(I + b) \coloneqq I + ab$.
It can be checked that these definitions are well-defined. Moreover, $I$ is the zero-element of $\frac{R}{I}$ and $I + 1$ is the $1$-element of $\frac{R}{I}$. It can be easily verified that $\frac{R}{I}$ satisfies the properties of commutative rings stipulated in Definition \ref{def:commutative_ring}.
\end{definition}

\begin{definition}
Let $R$ and $S$ be two commutative rings. A ring homomorphism $f : R \to S$ is a function which satisfies that for all $a, b \in R$, $f(a + b) = f(a) + f(b)$, $f(ab) = f(a)f(b)$ and $f(1_{R}) = 1_{S}$.
\end{definition}

\begin{definition}
Let $R$ be a commutative ring and $I$ be an ideal of $R$. The function $\rho : R \to \frac{R}{I}$ defined by 
$\rho(a) \coloneqq I + a$ is called the natural quotient map from $R$ to $\frac{R}{I}$. It can be easily seen that 
$\rho$ is a surjective ring homomorphism.
\end{definition}

\begin{definition}
For all integers $m \geq 2$, we use $Q_{m}$ to denote the commutative ring 
$\frac{F_{2}[t]}{\langle t^{2^{m}} + 1 \rangle}$.
\end{definition}

\begin{remark}
From this point onwards, the symbol $0$ has multiple possible meanings. It can refer to the zero element in $F_{2}$, the zero element of $F_{2}[t]$, or the zero element in $Q_{m}$. The right meaning of the symbol $0$ at any place should be clear from the context.
\end{remark}

\begin{definition}
Let $m \geq 2$ be an integer. Fix arbitrary $p(t) \in Q_{m}$. If $q(t) \in F_{2}[t]$ satisfies that $q(p(t)) = 0$,
then $q(t)$ is said to be an annihilating polynomial of $p(t)$.
\end{definition}

\begin{example}
For all integers $m \geq 2$, $q(t) \coloneqq t^{2^{m}} + 1 \in F_{2}[t]$ is an annihilating polynomial of 
$p(t) \coloneqq t^{2} + t + 1 \in Q_{m}$ because
\begin{align}
q(p(t)) 
&= (t^{2} + t + 1)^{2^{m}} + 1 \\
&= (t^{2^{m+1}} + t^{2^{m}} + 1) + 1 \\
&= (t^{2^{m}})(t^{2^{m}}) + t^{2^{m}} + 1 + 1 \\
&= 1 \cdot 1 + 1 + 1 + 1 \\
&= 0
\end{align}
where the 4\textsuperscript{th} inequality follows from the fact that the two elements $1, t^{2^{m}} \in Q_{m}$ are equal.
\end{example}

\begin{proposition}
Let $m \geq 2$ be an integer and $p(t) \in Q_{m}$. There exists exactly one monic non-zero annihilating polynomial $q^{*}(t) \in F_{2}[t]$ such that for all non-zero annihilating polynomial $q(t)$ of $p(t)$,
$\text{deg}(q^{*}(t)) \leq \text{deg}(q(t))$. $q^{*}(t)$ is called the minimal polynomial of $p(t)$.
\end{proposition}

\begin{proof}
Consider the set $X$ of all non-zero annihilating polynomials of $p(t)$. First, we prove that $X$ is non-empty. 
Note that $\{ p(t)^{k} : k \in \mathbb{Z}_{\geq 0} \} \subseteq Q_{m}$ and $Q_{m}$ is a finite set. By the pigeonhole principle, there exist non-negative integers $k_{1} \neq k_{2}$ such that $p(t)^{k_{1}} = p(t)^{k_{2}}$. 
Define $r(t) \coloneqq t^{k_{1}} + t^{k_{2}} \in F_{2}[t]$. We can see that $r(p(t)) = 0$ and therefore $r(t) \in X$, so $X$ is non-empty.

By the well-ordering principle, there exists a polynomial $q^{*}(t) \in X$ with the smallest degree. This shows the existence part of the proposition. To prove the uniqueness part, assume for a contradiction that there exist two distinct minimal polynomials $q^{*}_{1}(t)$, $q^{*}_{2}(t)$ of $p(t)$. Clearly, 
$\text{deg}(q^{*}_{1}(t)) = \text{deg}(q^{*}_{2}(t))$. Let $l$ be the degree of $q^{*}_{1}(t)$. We write
$q^{*}_{1}(t) = t^{l} + a_{l-1} t^{l-1} + \ldots + a_{0}$ and $q^{*}_{2}(t) = t^{l} + b_{l-1} t^{l-1} + \ldots + b_{0}$,
where $a_{0}, a_{1}, \ldots, a_{l-1}, b_{0}, b_{1}, \ldots, b_{l-1} \in F_{2}$.
Note that $q^{*}_{1}(t) + q^{*}_{2}(t) \neq 0$, $q^{*}_{1}(p(t)) + q^{*}_{2}(p(t)) = 0$ and
$\text{deg}(q^{*}_{1}(t) + q^{*}_{2}(t)) < l$. Hence, $q^{*}_{1}(t) + q^{*}_{2}(t)$ is a non-zero annihilating polynomial of $p(t)$ with degree smaller than $q^{*}_{1}(t)$. A contradiction is reached. Therefore, the uniqueness of the minimal polynomial of $p(t)$ is proved.
\end{proof}

\begin{proposition}\label{prop:min_poly_generates_all_annihilating_poly}
Let $m \geq 2$ be an integer and $p(t) \in Q_{m}$. Let $q(t) \in F_{2}[t]$ be the minimal polynomial of $p(t)$.
The set of all annihilating polynomials of $p(t)$ equals the principal ideal $\langle q(t) \rangle$.
\end{proposition}

\begin{proof}
We need to show that for all $r(t) \in F_{2}[t]$, $r(t)$ is an annihilating polynomial of $p(t)$ if and only if there exists $s(t) \in F_{2}[t]$ such that $r(t) = s(t)q(t)$ (i.e., $r(t) \in \langle q(t) \rangle$).

First, we show the ``if'' direction. If $r(t) = s(t)q(t)$ for some $s(t) \in F_{2}[t]$, then
\begin{equation}
r(p(t)) = s(p(t)) q(p(t)) = s(p(t)) \cdot 0 = 0.
\end{equation}
Therefore, $r(t)$ is an annihilating polynomial of $p(t)$.

Second, we show the ``only if'' direction. Since $F_{2}[t]$ is a ED, there exist $a(t), b(t) \in F_{2}[t]$ such that
$r(t) = a(t)q(t) + b(t)$ and either $b(t) = 0$ or $\text{deg}(b(t)) < \text{deg}(q(t))$. Note that
\begin{equation}
b(p(t)) = r(p(t)) - a(p(t)) q(p(t)) = 0 - a(p(t)) \cdot 0 = 0.
\end{equation}
Hence, $b(t)$ is an annihilating polynomial of $p(t)$. $b(t)$ must equal $0$; otherwise, $b(t)$ is a non-zero annihilating polynomial of $p(t)$ such that $\text{deg}(b(t)) < \text{deg}(q(t))$, which contradicts with the fact that $q(t)$ is the minimal polynomial of $p(t)$. Therefore, $r(t) = a(t)q(t) \in \langle q(t) \rangle$.
\end{proof}

We are now ready to prove the main theorem which leads us to the construction of controllable $k$-$k$-XOR-BNs with $1$ control node.

\begin{theorem}\label{thm:important_thm}
Let $m \geq 2$ be an integer and $\rho : F_{2}[t] \to Q_{m}$ be the natural quotient map.
Fix arbitrary $v(t) \in F_{2}[t]$ such that $v(1) = 1$. Define $s(t) \coloneqq 1 + (t + 1)v(t)$.
Then, $\rho(v(t))$ is a unit in $Q_{m}$. Moreover, the minimal polynomial of $\rho(s(t))$ is
$t^{2^{m}} + 1 \in F_{2}[t]$.
\end{theorem}

\begin{proof}
First, we show that $\rho(v(t))$ is a unit in $Q_{m}$. 
Because $F_{2}[t]$ is a ED and $v(1) = 1$, there exists $\beta(t) \in F_{2}[t]$ such that $v(t) = \beta(t)(t + 1) + 1$.
Therefore, $\text{gcd}(v(t), t + 1) = 1$. Because $t + 1$ is the only irreducible factor of 
$(t + 1)^{2^{m}} = t^{2^{m}} + 1$ in $F_{2}[t]$, we have $\text{gcd}(v(t), t^{2^{m}} + 1) = 1$.
By using the Euclidean algorithm, we can find $a(t), b(t) \in F_{2}[t]$ such that $a(t)v(t) + b(t)(t^{2^{m}} + 1) = 1$.
Hence,
\begin{align}
& \rho(a(t)v(t) + b(t)(t^{2^{m}} + 1)) = \rho(1) \\
\Rightarrow\ &\rho(a(t))\rho(v(t)) + \rho(b(t))\rho(t^{2^{m}} + 1) = \rho(1) \\
\Rightarrow\ &\rho(a(t))\rho(v(t)) = \rho(1). \label{eq:hashtag}
\end{align}
where the second implication holds because $\rho(t^{2^{m}} + 1) = \langle t^{2^{m}} + 1 \rangle + t^{2^{m}} + 1 = \langle t^{2^{m}} + 1 \rangle$ is the zero element of the quotient ring $Q_{m}$.
Therefore, $\rho(v(t))$ is a unit in $Q_{m}$.

Second, we show that the minimal polynomial of $\rho(s(t))$ is $t^{2^{m}} + 1$. Note that
\begin{align}
\rho(s(t))^{2^{m}} + \rho(1)
&= \rho(s(t)^{2^{m}} + 1) \\
&= \rho((1 + (t + 1)v(t))^{2^{m}} + 1) \\
&= \rho(1^{2^{m}} + (t + 1)^{2^{m}}(v(t))^{2^{m}} + 1) \\
&= \rho((t + 1)^{2^{m}}(v(t))^{2^{m}}) \\
&= \rho(t^{2^{m}} + 1) \rho((v(t))^{2^{m}}) \\
&= 0 \cdot \rho((v(t))^{2^{m}}) \\
&= 0.
\end{align}
Therefore, $t^{2^{m}} + 1$ is an annihilating polynomial of $\rho(s(t))$. Let $z(t) \in F_{2}[t]$ be the minimal polynomial of $\rho(s(t))$. By Proposition \ref{prop:min_poly_generates_all_annihilating_poly}, there exists
$\gamma(t) \in F_{2}[t]$ such that $(t + 1)^{2^{m}} = t^{2^{m}} + 1 = \gamma(t) z(t)$. We deduce that
$z(t) = (t + 1)^{d}$ for some $d \in \{ 0, 1, \ldots, 2^{m} \}$. Clearly, $z(t) \neq 1$ and hence $d \neq 0$.
On the other hand, note that
\begin{align}
(\rho(s(t)) + \rho(1))^{2^{m}-1}
&= \rho((s(t) + 1)^{2^{m} - 1}) \\
&= \rho(((t+1)v(t))^{2^{m} - 1}) \\
&= \rho((t + 1)^{2^{m} - 1}(v(t))^{2^{m} - 1}) \\
&= \rho((t + 1)^{2^{m} - 1})\rho(v(t))^{2^{m} - 1}. \label{eq:star}
\end{align}
Suppose for a contradiction that $(\rho(s(t)) + \rho(1))^{2^{m}-1} = 0$. Then, by Eq.\@ (\ref{eq:star}),
\begin{align}
&\rho((t + 1)^{2^{m} - 1})\rho(v(t))^{2^{m} - 1} = 0 \\
\Rightarrow\ & \rho(a(t))^{2^{m} - 1}\rho(v(t))^{2^{m} - 1}\rho((t + 1)^{2^{m} - 1}) 
= \rho(a(t))^{2^{m} - 1} \cdot 0 = 0 \\
\Rightarrow\ &\rho((t + 1)^{2^{m} - 1}) = 0 \quad \text{(by Eq.\@ (\ref{eq:hashtag}))} \\
\Rightarrow\ &(t + 1)^{2^{m} - 1} \in \langle t^{2^{m}} + 1 \rangle.
\end{align}
But the degree of any non-zero polynomial inside $\langle t^{2^{m}} + 1 \rangle$ is at least $2^{m}$. A contradiction is reached. Therefore, $(\rho(s(t)) + \rho(1))^{2^{m}-1} \neq 0$.
Hence, for $\delta = 0, 1, \ldots, 2^{m} - 1$, $(\rho(s(t)) + \rho(1))^{\delta} \neq 0$.
Therefore, $z(t) = (t + 1)^{2^{m}} = t^{2^{m}} + 1$ is the minimal polynomial of $\rho(s(t))$. 
The proof is complete.
\end{proof}

Before using Theorem \ref{thm:important_thm} to construct the desired $k$-$k$-XOR-BNs (Theorem \ref{thm:1_control_node}), we need to state a definition and a trivial lemma.

\begin{definition}
Let $m \geq 2$ be an integer. We define $P_{m} \in F^{2^{m} \times 2^{m}}_{2}$ to be the permutation matrix corresponding to the permutation function
$
\left(
\begin{array}{cccccc}
1		& 2	& 3	& \cdots 	& 2^{m}-1	& 2^{m} \\
2^{m}	& 1	& 2	& \cdots 	& 2^{m}-2	& 2^{m} - 1
\end{array}
\right)
$.
\end{definition}

\begin{example}
Consider the case $m = 2$. Then,
\begin{align}
P_{2} &=
\left(
\begin{array}{cccc}
0 & 1 & 0 & 0 \\
0 & 0 & 1 & 0 \\
0 & 0 & 0 & 1 \\
1 & 0 & 0 & 0
\end{array}
\right)
& 
P^{2}_{2} &= 
\left(
\begin{array}{cccc}
0 & 0 & 1 & 0 \\
0 & 0 & 0 & 1 \\
1 & 0 & 0 & 0 \\
0 & 1 & 0 & 0
\end{array}
\right) \\
P^{3}_{2} &=
\left(
\begin{array}{cccc}
0 & 0 & 0 & 1 \\
1 & 0 & 0 & 0 \\
0 & 1 & 0 & 0 \\
0 & 0 & 1 & 0
\end{array}
\right)
& 
P^{4}_{2} &= 
\left(
\begin{array}{cccc}
1 & 0 & 0 & 0 \\
0 & 1 & 0 & 0 \\
0 & 0 & 1 & 0 \\
0 & 0 & 0 & 1
\end{array}
\right)
\end{align}
\end{example}

\begin{lemma}
Let $m \geq 2$ be an integer. 
$(\mathbf{w}_{1}, \mathbf{w}_{2}, \ldots, \mathbf{w}_{2^{m}}) \coloneqq
(\langle t^{2^{m}} + 1 \rangle + t^{2^{m} - 1}, \ldots, 
\langle t^{2^{m}} + 1 \rangle + t, \langle t^{2^{m}} + 1 \rangle + 1)$
forms an ordered basis of $Q_{m}$.
\end{lemma}

\begin{theorem}\label{thm:1_control_node}
Let $m \geq 2$ be an integer and $k \in [3, 2^{m} - 1]$ be an odd integer. There exists a $2^{m}$-node 
controllable $k$-$k$-XOR-BN with 1 control node.
\end{theorem}

\begin{proof}
We split into two cases to consider: $k \equiv 3 \pmod{4}$ and $k \equiv 1 \pmod{4}$.

Consider the first case: $k \equiv 3 \pmod{4}$. Define $v_{1}(t) \coloneqq t + t^{3} + \ldots + t^{k-2} \in F_{2}[t]$, which has an odd number of terms and hence $v_{1}(1) = 1$. Define
$s_{1}(t) \coloneqq 1 + (t + 1)v_{1}(t) = 1 + t + \ldots + t^{k-1}$, which clearly has $k$ terms. 
By Theorem \ref{thm:important_thm}, the minimal polynomial of $\rho(s_{1}(t))$ is 
$t^{2^{m}} + 1 \in F_{2}[t]$. Hence, for all non-zero polynomial $q(t) \in F_{2}[t]$ with 
$\text{deg}(q(t)) \leq 2^{m} - 1$, $q(\rho(s_{1}(t))) \neq \langle t^{2^{m}} + 1 \rangle$ (i.e., $q(\rho(s_{1}(t)))$ is not the zero element of $Q_{m}$). Therefore, $\rho(1)$, $\rho(s_{1}(t))$, \ldots, 
$(\rho(s_{1}(t)))^{2^{m} - 1}$ are linearly independent in $Q_{m}$. Because $\text{dim}(Q_{m}) = 2^{m}$, 
$\rho(1)$, $\rho(s_{1}(t))$, \ldots, $(\rho(s_{1}(t)))^{2^{m} - 1}$ form a basis of $Q_{m}$.

Consider the linear map $T_{1} : Q_{m} \to Q_{m}$ such that for all $a(t) \in F_{2}[t]$, 
$T_{1}(\rho(a(t))) = \rho(s_{1}(t))\rho(a(t))$. Note that
\begin{align}
T_{1}(\langle t^{2^{m}} + 1 \rangle + 1) 	&= \langle t^{2^{m}} + 1 \rangle + t^{k-1} + \cdots + t + 1 \\
T_{1}(\langle t^{2^{m}} + 1 \rangle + t) 	&= \langle t^{2^{m}} + 1 \rangle + t^{k} + \cdots + t^{2} + t \\
							&\vdots \\
T_{1}(\langle t^{2^{m}} + 1 \rangle + t^{2^{m} - 1}) 
&= \langle t^{2^{m}} + 1 \rangle + t^{2^{m} + k - 2} + \cdots + t^{2^{m}} + t^{2^{m} - 1} \\
&= \langle t^{2^{m}} + 1 \rangle + t^{k - 2} + \cdots + t + 1 + t^{2^{m} - 1}
\end{align}
which can be alternatively written as
\begin{align}
T_{1}(\mathbf{w}_{2^{m}}) 
&= \mathbf{w}_{2^{m} - k + 1} + \cdots + \mathbf{w}_{2^{m} - 1} + \mathbf{w}_{2^{m}} \\
T_{1}(\mathbf{w}_{2^{m} - 1}) 
&= \mathbf{w}_{2^{m} - k} + \cdots + \mathbf{w}_{2^{m} - 2} + \mathbf{w}_{2^{m} - 1} \\
&\vdots \\
T_{1}(\mathbf{w}_{1}) 
&= \mathbf{w}_{2^{m} - k + 2} + \cdots + \mathbf{w}_{2^{m} - 1} + \mathbf{w}_{2^{m}} + \mathbf{w}_{1}
\end{align}
Therefore, the transformation matrix $A_{1} \in F^{2^{m} \times 2^{m}}_{2}$ of $T_{1} : Q_{m} \to Q_{m}$ with respect to the ordered basis $(\mathbf{w}_{1}, \mathbf{w}_{2}, \ldots, \mathbf{w}_{2^{m}})$ is given by
$A_{1} = I + P_{m} + \cdots + P^{k-1}_{m}$, which is a circulant matrix. Moreover, each column and row of $A_{1}$ has exactly $k$ $1$-entries. As mentioned before, $\mathbf{w}_{2^{m}} = \rho(1)$, 
$T_{1}(\mathbf{w}_{2^{m}}) = \rho(s_{1}(t))$, \ldots, 
$T^{2^{m} - 1}_{1}(\mathbf{w}_{2^{m}}) = (\rho(s_{1}(t)))^{2^{m} - 1}$ form a basis of $Q_{m}$. Therefore,
$\mathbf{e}_{2^{m}}$, $A_{1}\mathbf{e}_{2^{m}}$, \ldots, $A^{2^{m} - 1}_{1}\mathbf{e}_{2^{m}}$ form a basis of 
$F^{2^{m}}_{2}$ where $\mathbf{e}_{2^{m}} \in F^{2^{m}}_{2}$ is the standard unit vector whose last entry is equal to $1$. Then, by Theorem \ref{thm:Wu_spans_Fn_2}, the $2^{m}$-node $k$-$k$-XOR-BN with associated matrix $A_{1}$ and control-node set $U = \{ x_{2^{m}} \}$ is controllable.

Consider the second case: $k \equiv 1 \pmod{4}$. Define $v_{2}(t) \coloneqq 1 + t^{2} + \ldots + t^{k-1} \in F_{2}[t]$, which has an odd number of terms and hence $v_{2}(1) = 1$. Define
$s_{2}(t) \coloneqq 1 + (t + 1)v_{2}(t) = t + t^{2} + \ldots + t^{k}$, which clearly has $k$ terms. 
By Theorem \ref{thm:important_thm}, the minimal polynomial of $\rho(s_{2}(t))$ is 
$t^{2^{m}} + 1 \in F_{2}[t]$. Hence, for all non-zero polynomial $q(t) \in F_{2}[t]$ with 
$\text{deg}(q(t)) \leq 2^{m} - 1$, $q(\rho(s_{2}(t))) \neq \langle t^{2^{m}} + 1 \rangle$ (i.e., $q(\rho(s_{2}(t)))$ is not the zero element of $Q_{m}$). Therefore, $\rho(1)$, $\rho(s_{2}(t))$, \ldots, 
$(\rho(s_{2}(t)))^{2^{m} - 1}$ are linearly independent in $Q_{m}$. Because $\text{dim}(Q_{m}) = 2^{m}$, 
$\rho(1)$, $\rho(s_{2}(t))$, \ldots, $(\rho(s_{2}(t)))^{2^{m} - 1}$ form a basis of $Q_{m}$.

Consider the linear map $T_{2} : Q_{m} \to Q_{m}$ such that for all $a(t) \in F_{2}[t]$, 
$T_{2}(\rho(a(t))) = \rho(s_{2}(t))\rho(a(t))$. 
Note that
%
\begin{align}
T_{2}(\mathbf{w}_{2^{m}}) 
&= \mathbf{w}_{2^{m} - k} + \cdots + \mathbf{w}_{2^{m} - 2} + \mathbf{w}_{2^{m} - 1} \\
T_{2}(\mathbf{w}_{2^{m} - 1}) 
&= \mathbf{w}_{2^{m} - k - 1} + \cdots + \mathbf{w}_{2^{m} - 3} + \mathbf{w}_{2^{m} - 2} \\
&\vdots \\
T_{2}(\mathbf{w}_{1}) 
&= \mathbf{w}_{2^{m} - k + 1} + \cdots + \mathbf{w}_{2^{m} - 1} + \mathbf{w}_{2^{m}}
\end{align}
Therefore, the transformation matrix $A_{2} \in F^{2^{m} \times 2^{m}}_{2}$ of $T_{2} : Q_{m} \to Q_{m}$ with respect to the ordered basis $(\mathbf{w}_{1}, \mathbf{w}_{2}, \ldots, \mathbf{w}_{2^{m}})$ is given by
$A_{2} = P_{m} + P^{2}_{m} + \cdots + P^{k}_{m}$, which is a circulant matrix. Moreover, each column and row of $A_{2}$ has exactly $k$ $1$-entries. As mentioned before, $\mathbf{w}_{2^{m}} = \rho(1)$, 
$T_{2}(\mathbf{w}_{2^{m}}) = \rho(s_{2}(t))$, \ldots, 
$T^{2^{m} - 1}_{2}(\mathbf{w}_{2^{m}}) = (\rho(s_{2}(t)))^{2^{m} - 1}$ form a basis of $Q_{m}$. Therefore,
$\mathbf{e}_{2^{m}}$, $A_{2}\mathbf{e}_{2^{m}}$, \ldots, $A^{2^{m} - 1}_{2}\mathbf{e}_{2^{m}}$ form a basis of 
$F^{2^{m}}_{2}$. Then, by Theorem \ref{thm:Wu_spans_Fn_2}, the $2^{m}$-node $k$-$k$-XOR-BN with associated matrix $A_{2}$ and control-node set $U = \{ x_{2^{m}} \}$ is controllable.

The proof is complete.
\end{proof}

\begin{example}
Let $m = 3$ and $k = 3$. The $8$-node controllable $3$-$3$-XOR-BN with $1$ control node constructed in the proof of Theorem \ref{thm:1_control_node} is associated with the matrix
\begin{equation}
A_{1} \coloneqq
\left(
\begin{array}{cccccccc}
1 & 1 & 1 & 0 & 0 & 0 & 0 & 0 \\
0 & 1 & 1 & 1 & 0 & 0 & 0 & 0 \\
0 & 0 & 1 & 1 & 1 & 0 & 0 & 0 \\
0 & 0 & 0 & 1 & 1 & 1 & 0 & 0 \\
0 & 0 & 0 & 0 & 1 & 1 & 1 & 0 \\
0 & 0 & 0 & 0 & 0 & 1 & 1 & 1 \\
1 & 0 & 0 & 0 & 0 & 0 & 1 & 1 \\
1 & 1 & 0 & 0 & 0 & 0 & 0 & 1
\end{array}
\right) 
\in F^{8 \times 8}_{2}
\end{equation}
and the control-node set $U = \{ x_{8} \}$.
More explicitly, the Boolean update rules of this BN are
\begin{align}
x_{1}(t + 1)	&= x_{1}(t) \oplus x_{2}(t) \oplus x_{3}(t)	& x_{5}(t + 1)	&= x_{5}(t) \oplus x_{6}(t) \oplus x_{7}(t) \\
x_{2}(t + 1)	&= x_{2}(t) \oplus x_{3}(t) \oplus x_{4}(t)	& x_{6}(t + 1)	&= x_{6}(t) \oplus x_{7}(t) \oplus x_{8}(t) \\
x_{3}(t + 1)	&= x_{3}(t) \oplus x_{4}(t) \oplus x_{5}(t)	& x_{7}(t + 1)	&= x_{7}(t) \oplus x_{8}(t) \oplus x_{1}(t) \\
x_{4}(t + 1)	&= x_{4}(t) \oplus x_{5}(t) \oplus x_{6}(t)	& x_{8}(t + 1)	&= x_{8}(t) \oplus x_{1}(t) \oplus x_{2}(t)
\end{align}
Note that
\begin{align}
\mathbf{e}_{8}		&= [0, 0, 0, 0, 0, 0, 0, 1]^{T}	& A^{4}_{1}\mathbf{e}_{8} &= [0, 0, 0, 1, 0, 0, 0, 0]^{T} \\
A_{1}\mathbf{e}_{8}	&= [0, 0, 0, 0, 0, 1, 1, 1]^{T}	& A^{5}_{1}\mathbf{e}_{8} &= [0, 1, 1, 1, 0, 0, 0, 0]^{T} \\
A^{2}_{1}\mathbf{e}_{8}	&= [0, 0, 0, 1, 0, 1, 0, 1]^{T}	& A^{6}_{1}\mathbf{e}_{8} &= [0, 1, 0, 1, 0, 0, 0, 1]^{T} \\
A^{3}_{1}\mathbf{e}_{8}	&= [0, 1, 1, 0, 1, 0, 1, 1]^{T}	& A^{7}_{1}\mathbf{e}_{8} &= [1, 0, 1, 1, 0, 1, 1, 0]^{T}
\end{align}
We can easily write a computer program to check that
\begin{equation}
\det([\mathbf{e}_{8}, A_{1}\mathbf{e}_{8}, A^{2}_{1}\mathbf{e}_{8}, \ldots, A^{7}_{1}\mathbf{e}_{8}]) = 1 \neq 0,
\end{equation}
which implies that $\mathbf{e}_{8}$, $A_{1}\mathbf{e}_{8}$, 
$A^{2}_{1}\mathbf{e}_{8}$, \ldots, $A^{7}_{1}\mathbf{e}_{8}$ form a basis of $F^{8}_{2}$.
\end{example}

\begin{example}
Let $m = 3$ and $k = 5$. The $8$-node controllable $5$-$5$-XOR-BN with $1$ control node constructed in the proof of Theorem \ref{thm:1_control_node} is associated with the matrix
\begin{equation}
A_{2} \coloneqq
\left(
\begin{array}{cccccccc}
0 & 1 & 1 & 1 & 1 & 1 & 0 & 0 \\
0 & 0 & 1 & 1 & 1 & 1 & 1 & 0 \\
0 & 0 & 0 & 1 & 1 & 1 & 1 & 1 \\
1 & 0 & 0 & 0 & 1 & 1 & 1 & 1 \\
1 & 1 & 0 & 0 & 0 & 1 & 1 & 1 \\
1 & 1 & 1 & 0 & 0 & 0 & 1 & 1 \\
1 & 1 & 1 & 1 & 0 & 0 & 0 & 1 \\
1 & 1 & 1 & 1 & 1 & 0 & 0 & 0
\end{array}
\right) 
\in F^{8 \times 8}_{2}
\end{equation}
and the control-node set $U = \{ x_{8} \}$.
Note that
\begin{align}
\mathbf{e}_{8}		&= [0, 0, 0, 0, 0, 0, 0, 1]^{T}	& A^{4}_{2}\mathbf{e}_{8} &= [0, 0, 0, 1, 0, 0, 0, 0]^{T} \\
A_{2}\mathbf{e}_{8}	&= [0, 0, 1, 1, 1, 1, 1, 0]^{T}	& A^{5}_{2}\mathbf{e}_{8} &= [1, 1, 1, 0, 0, 0, 1, 1]^{T} \\
A^{2}_{2}\mathbf{e}_{8}	&= [0, 1, 0, 1, 0, 0, 0, 1]^{T}	& A^{6}_{2}\mathbf{e}_{8} &= [0, 0, 0, 1, 0, 1, 0, 1]^{T} \\
A^{3}_{2}\mathbf{e}_{8}	&= [0, 1, 0, 1, 0, 0, 1, 0]^{T}	& A^{7}_{2}\mathbf{e}_{8} &= [0, 0, 1, 0, 0, 1, 0, 1]^{T}
\end{align}
We can easily write a computer program to check that
\begin{equation}
\det([\mathbf{e}_{8}, A_{2}\mathbf{e}_{8}, A^{2}_{2}\mathbf{e}_{8}, \ldots, A^{7}_{2}\mathbf{e}_{8}]) = 1 \neq 0,
\end{equation}
which implies that $\mathbf{e}_{8}$, $A_{2}\mathbf{e}_{8}$, 
$A^{2}_{2}\mathbf{e}_{8}$, \ldots, $A^{7}_{2}\mathbf{e}_{8}$ form a basis of $F^{8}_{2}$.
\end{example}

\section{Conclusion}\label{section:conclusion}

In this paper, we have investigated the minimal controllability of degree-constrained Boolean networks (BNs). First, we have established lower-bound-related inequalities and general upper bounds on the number of required control nodes for several families of controllable majority-type threshold BNs. Second, we have constructed controllable majority-type BNs and threshold BNs with mixed-sign coefficients that admit small control-node sets. Third, for general XOR-BNs, we have given a linear-algebraic necessary and sufficient condition for controllability and developed polynomial-time algorithms to compute control-node sets and control signals. Finally, leveraging ring theory and linear algebra, we have derived best-case upper bounds for degree-constrained XOR-BNs. Specifically, we have shown that for any positive integers $n$ and $k$ with $n > k \geq 2$, there exists a controllable $n$-node $k$ $k$-XOR-BN with a control-node set of size $k-1$; moreover, we have shown that when $n \geq 4$ is a power of two and $k \in [3, n - 1]$ is an odd integer, there exists a controllable $n$-node $k$-$k$-XOR-BN with a single control node. We expect our research to advance understanding of both minimal interventions in biological systems and opinion dynamics in social networks.

There are a number of possible future directions of research on the minimal controllability problem of BNs. 
Firstly, the gaps between our general lower bounds (Corollary \ref{coroll:k_in_MAJOR_BN_lower_bound} and Theorem \ref{thm:2k_in_MTBI_BN_ineq_thm}) and our best-case upper bounds (Theorems \ref{thm:(2k+1)_bcub}, \ref{thm:2k_MAJOR_bcub}, \ref{thm:2k_MTBI_bcub}) for $(2k+1)$-$(2k+1)$-MAJORITY-BNs, $2k$-$2k$-MAJORITY-BNs and $2k$-$2k$-MTBI-BNs are small when the value of $m \geq 2$ is small, and future investigation can focus on narrowing these gaps for large values of $m$.
Secondly, we note that most of our results for threshold BNs in this paper are related to majority-type Boolean threshold functions. In the future, the research community can further study the minimal control of threshold BNs consisting of non-majority-type Boolean threshold functions with mixed-sign mixed-magnitude coefficients. This is an important direction of research because such BNs can model biological systems in which both activation and inhibition interactions are present. 
Thirdly, canalyzing Boolean functions can also effectively model many biological interactions \cite{HarrisSawhill, KauffmanCanalyzing}, and the minimal controllability problem of BNs consisting of canalyzing Boolean functions can also be investigated to increase our understanding of certain biological systems. 
Fourthly, further research can also be conducted on whether for any positive integers $n$ and $k$ such that $n > k \geq 2$, there exists a controllable $n$-node $k$-$k$-XOR-BN with a control-node set of size $1$. Our paper only resolves this question for the case that $k = 2$ and the case that $n \geq 4$ is a power of $2$ and $k \in [3, n - 1]$ is an odd integer.



\end{document}